\documentclass[11pt]{article}

\usepackage{amsmath, amssymb, amsthm}
\usepackage{algorithm}
\usepackage{algpseudocode}
\usepackage{float, graphicx, multirow,subcaption, setspace}
\usepackage{comment}
\usepackage{url}
\usepackage{enumitem}
\usepackage{tikz}
\usepackage{bm, bbm}
\usetikzlibrary{shapes.geometric, positioning}
\usetikzlibrary{quotes, angles}
\usepackage{rotating}
\usepackage{hyperref}
\usepackage{natbib}
\usepackage{cleveref}

\definecolor{SkyBlue}{RGB}{14, 118, 188}
\definecolor{BrightRed}{RGB}{223, 82, 78}
\definecolor{Green638}{RGB}{165,255,118} 
\definecolor{RevisedColor}{RGB}{25, 200, 150}

\hypersetup{pdfborder = {0 0 0.5 [3 3]}, colorlinks = true, linkcolor = BrightRed, citecolor = SkyBlue}

\newcommand{\E}{\mathbb{E}} 
\def\P{\mathbb{P}} 

\newcommand{\bx}{\bm{x}}
\newcommand{\bz}{\bm{z}}

\newcommand{\bY}{\bm{Y}}
\newcommand{\bX}{\bm{X}}
\newcommand{\bZ}{\bm{Z}}

\theoremstyle{plain}
\newtheorem{proposition}{Proposition}
\newtheorem{theorem}{Theorem}
\newtheorem{lemma}[theorem]{Lemma}
\newtheorem{corollary}{Corollary}

\theoremstyle{definition}

\newtheorem{ex}{Example}

\newlist{assumenum}{enumerate}{1}
\setlist[assumenum]{label=\textbf{A\theassumenumi}, ref=A\theassumenumi}
\newcounter{assumption}
\crefname{assumenumi}{Assumption}{Assumptions}
\Crefname{assumenumi}{Assumption}{Assumptions}

\newlist{bssumenum}{enumerate}{1}
\setlist[bssumenum]{label=\textbf{B\thebssumenumi}, ref=B\thebssumenumi}

\crefname{bssumenumi}{Assumption}{Assumptions}
\Crefname{bssumenumi}{Assumption}{Assumptions}

\newlist{cssumenum}{enumerate}{1}
\setlist[cssumenum]{label=\textbf{C\thecssumenumi}, ref=C\thecssumenumi}

\crefname{cssumenumi}{Assumption}{Assumptions}
\Crefname{cssumenumi}{Assumption}{Assumptions}

\newlist{pssumenum}{enumerate}{1}
\setlist[pssumenum]{label=\textbf{P\arabic*}, ref=P\arabic*}

\crefname{pssumenumi}{Assumption}{Assumptions}
\Crefname{pssumenumi}{Assumption}{Assumptions}

\usepackage{fullpage, parskip}
\usepackage{amsmath}
\usepackage{mathtools}
\usepackage{bm}
\usepackage{graphicx}
\usepackage{multirow}
\usepackage{booktabs}
\usepackage{tikz}
\usetikzlibrary{arrows.meta,positioning,calc,fit,backgrounds,decorations.pathreplacing}

\usepackage{algorithm}
\usepackage{algpseudocode}
\usepackage{rotating}
\usepackage{comment}
\usepackage{url}
\usepackage{enumitem}
\usepackage{bbm}
\usepackage{siunitx}

\usepackage{cleveref}

\creflabelformat{equation}{#2\textup{#1}#3}

\hypersetup{
  pdfborder = {0 0 0.5 [3 3]},
  colorlinks = true,
  linkcolor = BrightRed,
  citecolor = SkyBlue,
  urlcolor = SkyBlue
}

\newcommand{\includesupp}{1}

\newcommand{\switchref}[2]{%
  \if\includesupp1%
    #1%
  \else%
    #2%
  \fi%
}

\newcommand{\suppref}[2]{%
  \if\includesupp1%
    #1%
  \else%
    #2%
  \fi%
}

\graphicspath{{./figures/}{./plots/}{./art/}}

\title{Multivariate Varying-Coefficient BART with Graphical Horseshoe Priors}

\author{
Soham Ghosh\thanks{Department of Statistics, University of Wisconsin--Madison. \texttt{sghosh39@wisc.edu}}
\and
Sameer K. Deshpande\thanks{Department of Statistics, University of Wisconsin--Madison. \texttt{sameer.deshpande@wisc.edu}}
}

\begin{document}

\maketitle

\begin{abstract}
Modern multivariate regression problems involve several related outcomes whose regression effects are not only nonlinear, heterogeneous, and outcome-specific, but the residual dependence among outcomes is scientifically meaningful. Existing multivariate Bayesian tree-based methods typically address only part of this problem: some impose substantial sharing of tree architecture across outcomes, which can be overly restrictive when different responses depend on distinct predictors or effect modifiers, while others accommodate residual dependence but retain simpler mean structures. This paper develops \texttt{multiVCBART}, a multivariate varying-coefficient Bayesian additive regression tree framework that jointly models flexible outcome-specific coefficient surfaces and a sparse residual precision matrix. Specifically, each entry of the coefficient matrix $\bm{B}(\bm{x})$ is represented by its own BART ensemble, allowing predictor effects to vary nonlinearly with modifiers $\bm{x}$ and to differ across outcomes, while a Graphical Horseshoe prior on the precision matrix $\Omega$ captures parsimonious residual conditional dependence. This yields a multivariate BART model that allows rich heterogeneous mean structure, and still borrows strength across outcomes through joint estimation of $(\bm{B},\Omega)$.
We introduce a sampler that reduces the multivariate Gaussian likelihood to a sequence of outcome-wise Gaussian pseudo-response updates, thereby permitting efficient BART backfitting within each outcome while updating $\Omega$ through a column-wise graphical shrinkage step. From a theoretical standpoint, we establish posterior contraction for the joint multivariate varying-coefficient model, showing that the posterior adapts to heterogeneous smoothness and sparsity in $\bm{B}$ and $\Omega$. Empirically, \texttt{multiVCBART} performs strongly on synthetic datasets, especially in high-dimensional settings with sparse and outcome-specific signals, where shared-tree multivariate BART formulations and linear SUR-type competitors are the most restrictive. In a re-analysis of the study of Drug Sensitivity in cancer, \texttt{multiVCBART} identifies significant biomarker signals, recovers coherent pharmacologic groupings among the drugs, and estimates an interpretable residual drug-response network after adjustment for tissue context and molecular covariates.

\end{abstract}


\section{Introduction}
\label{sec:intro}

In many modern scientific applications, researchers seek to understand not only whether certain primary covariates influence multiple, possibly correlated, outcomes but also whether the relationship between covariates and outcomes varies with respect to other variables. 
For instance, in pharmacogenetics, researchers observe drug sensitivity
measurements for multiple compounds on the same cancer cell lines, together with molecular
features such as gene expression, copy-number variation, and mutation indicators \citep{garnett2012a,YangGDSC} and try to answer several questions simultaneously: which molecular features (covariates) are associated with sensitivity to each drug? How do these associations vary across tissue contexts (effect modifiers)? And which drug responses remain conditionally dependent after adjusting for observed molecular and tissue information? 

These goals naturally lead to a multivariate varying-coefficient model
\begin{equation}\label{eq:multi-VCM}
\bm{y} \mid \bm{x},\bm{z} \sim \mathcal N_{q} \big(\bm{B}(\bm{x})^\top \bm{z}, \Omega^{-1}\big),
\end{equation}
where $\bm y\in\mathbb R^q$ is a vector of outcomes (e.g., sensitivities of multiple drugs); $\bm z \in\mathbb R^p$ is a vector of covariates of primary interest (e.g., molecular features); and $\bm x\in\mathbb R^d$ is a vector of effect modifiers (e.g., tissue type).
In \Cref{eq:multi-VCM}, $\bm{B}(\bm x)\in\mathbb R^{p\times q}$ is a matrix of unknown coefficient functions evaluated at $\bm x$, and $\Omega$ is a residual precision matrix encoding conditional dependence among the outcomes. 
The model in \Cref{eq:multi-VCM} extends the classical linear varying coefficient \citep[][VCM]{Hastie1993} to the multi-output setting and allows the effects of each covariate $Z_{j}$ on each outcome $Y_{k}$ to vary as a function of the effect modifiers $\bm{X}.$

The model in \Cref{eq:multi-VCM} subsumes several important existing statistical frameworks as special cases. 
First, if the coefficient matrix is constant with respect to $\bm x$ (i.e., $\bm{B}(\bm x_i) \equiv \bm{B}$), the model reduces the standard multivariate linear regression model.
This seamlessly generalizes to settings where predictors differ by outcome; by structuring the model equation-by-equation --- stacking outcome-specific predictors into a single expanded vector $\bm z_i = (\bm z_{i1}^\top, \dots, \bm z_{iq}^\top)^\top$ and defining the coefficient matrix as the direct sum of outcome-specific vectors $\bm B = \bigoplus_{r=1}^q \bm{\beta}_r$ --- the model reduces to \citet{ZellnerSUR}'s celebrated Seemingly Unrelated Regressions (SUR).
Second, when the $\bm{z} = (1, T)^{\top}$ for a binary treatment $T,$ the model in \Cref{eq:multi-VCM} can be used to estimate conditional average treatment effects under suitable identifying assumptions, effectively extending \citet{Hahn2020}'s Bayesian causal forests model to the multi-outcome setting. 

Despite the generality of \eqref{eq:multi-VCM}, fitting such a model in high dimensions poses substantial statistical challenges. 
Different outcomes may depend on distinct subsets of predictors, exhibit unique nonlinear patterns of effect modification, and yet remain strongly dependent even after adjusting for observed covariates.
Ignoring this cross-outcome dependence by fitting separate univariate models forfeits efficiency and distorts joint uncertainty quantification \citep{ZellnerSUR,SmithKohn2000}, while imposing a common mean structure across outcomes can be overly restrictive for distinct scientific processes \citep{baldwin2014analyzing}. 
Furthermore, classical approaches to fitting VCMs like linear smoothers and kernel methods \citep[see, e.g.,][]{Hoover1998, Wu2000, LiRacine2010} or basis expansions \citep{huang2002,Bai2023} struggle when the coefficient functions in $\bm{B}(\bm{x})$ contain strong nonlinearities, higher-order interactions, heterogeneous smoothness, or many irrelevant modifiers \citep{Liu2014,Deshpande2024}.

To operationalize the general multivariate framework in \eqref{eq:multi-VCM} and overcome the limitations of classical smoothers, we propose a Bayesian nonparametric approach that represents each scalar coefficient function \(B_{jr}(\cdot)\) using its own Bayesian additive regression tree ensemble \citep[BART;][]{Chipman2010}. 
To encourage sparsity in $\bm{B},$ we specify global--local shrinkage priors on the tree leaf parameters to suppress irrelevant covariate-outcome surfaces.
We further specify a graphical horseshoe prior on \(\Omega\) then learns a sparse residual conditional-dependence graph together with the nonlinear mean structure.

\subsection{Related work}

A closely related line of work arises from multivariate regression and seemingly unrelated regression models, where one jointly models \(q\) responses through a coefficient matrix and a residual covariance or precision matrix. 
The classical SUR framework of \citet{ZellnerSUR} showed that, when responses are correlated, joint modeling can improve efficiency relative to equation-by-equation estimation. 
In the Bayesian literature, this perspective has been extended in several important directions. 
Early work by \citet{brown1998multivariate} developed multivariate Bayesian variable selection for regression with correlated outcomes, while \citet{BhadraMallick2013} studied sparse high-dimensional Gaussian SUR models with joint selection of predictors and inverse-covariance elements. 
More recent contributions have emphasized scalable and structured joint regularization of regression coefficients and residual dependence, including \texttt{BayesSUR} of \citet{Bottolo2021}, which combines cell-sparse variable selection with sparse covariance selection, and the horseshoe\(+\) based SUR model of \citet{han2023}, which places continuous shrinkage priors on both the regression coefficients and the precision matrix. 
These methods provide an important foundation for our work, but they are largely built around linear mean specifications and therefore do not directly accommodate nonlinear effect heterogeneity.

A second relevant line of work concerns Bayesian tree-based varying-coefficient models. \citet{Deshpande2024} introduced \texttt{VCBART}, which represents each coefficient function in a univariate varying-coefficient model by a BART ensemble, thereby combining coefficient-level interpretability with flexible nonlinear effect modification. 
More recently, \citet{ghosh2025fittingsparsehighdimensionalvaryingcoefficient} proposed \texttt{sparseVCBART}, which extends this idea to high-dimensional settings by placing global--local shrinkage priors on the regression tree leaf outputs and hierarchical sparsity priors on the splitting probabilities. 
This allows the model to learn both which covariates have nonzero effects and which modifiers drive each nonzero coefficient function. 
Our construction builds directly on this sparse varying-coefficient BART philosophy, but moves from a univariate response model to a genuinely multivariate regression setting.

A growing literature has extended BART beyond standard univariate regression to structured and
multivariate settings. 
One important strand uses \emph{shared} tree structures across related
responses or model components. 
For example, \citet{mcjames2023bayesiancausalforestsmultivariate} develop a multivariate extension of Bayesian causal forests, \texttt{mvbcf}, for multiple outcomes that allows all outcomes to share the same tree architecture, and \citet{Um2023} propose a multivariate \texttt{skewBART} model in which a common multivariate BART prior is combined with a skew-normal error model to accommodate within-subject dependence and non-Gaussian responses. 
Related ideas also appear in the shared-forest framework of \citet{Linero2019}, where multiple model components are tied together through a common set of trees. 
A second strand moves closer to the SUR perspective by allowing outcome-specific mean models while still accounting for residual dependence. 
Most notably, \citet{esser2025seeminglyunrelatedbayesianadditive} propose seemingly unrelated BART, \texttt{suBART}, a multivariate extension of BART in which each response is assigned its own ensemble of trees while dependence across outcomes is captured through a joint residual model. 
Multivariate BART ideas have also been developed in time-series settings, such as Bayesian additive vector autoregressive tree models \citep{Huber2022} and multivariate BART
models for tail forecasting \citep{clark2023}. These contributions highlight the promise of joint tree-based modeling for multivariate responses.

Despite these important contributions, existing multivariate BART formulations do not fully address the particular combination of features targeted here. 
Approaches such as \texttt{mvbcf} and multivariate \texttt{skewBART} borrow strength through shared or partially shared tree structures.
This is attractive when outcomes have similar regression structure, but it can be restrictive when different outcomes depend on distinct predictors or effect modifiers. 
Conversely, SUR-style models such as \texttt{suBART} allow response-specific ensembles and residual dependence, but they do not directly target sparse, predictor-specific varying-coefficient structure. In contrast, our goal is to learn a separate nonlinear coefficient surface for each covariate--outcome pair, shrink irrelevant surfaces toward zero, select relevant modifiers within each surface, and estimate the residual dependence graph jointly with the mean functions.

\subsection{Our contributions}

We propose \texttt{multiVCBART}, a Bayesian nonparametric method for the general model in \eqref{eq:multi-VCM} that combines outcome-adaptive varying-coefficient BART components with a sparse residual precision model. 
Concretely, each entry of \(\bm B(\bm x)\) is represented by its own BART ensemble, allowing the effect of each covariate to vary nonlinearly with the modifier vector \(\bm x_i\) in an outcome-specific manner. 
At the same time, the residual precision matrix \(\Omega\) is jointly estimated so that cross-outcome dependence is learned rather than ignored. 
In this way, our method extends the classical benefits of joint multivariate modeling beyond linear SUR-type settings to a substantially richer varying-coefficient regime.

To make this practically useful in high-dimensional settings, we equip the coefficient-function ensembles with sparsity-inducing horseshoe priors \citep{Carvalho2010} that shrink irrelevant covariate--outcome surfaces while preserving important heterogeneous effects. 
We place a complementary Graphical Horseshoe prior \citep{LiJCGS} on \(\Omega\) to encourage a parsimonious residual dependence graph. The resulting framework retains coefficient-level interpretability, permits outcome-specific nonlinear heterogeneity, and still borrows strength across outcomes through joint estimation of \((\bm B,\Omega)\).

Empirically, in \Cref{sec:experiments} and \Cref{sec:multi-experiments}, we show, across a range of controlled synthetic benchmarks, that \texttt{multiVCBART} achieves strong predictive accuracy and uncertainty quantification relative to representative multivariate baselines. In high-dimensional nonlinear regression settings with sparse and outcome-specific signal structure, our method consistently attains the best or near-best estimation and often improved predictive interval coverage relative to competing multivariate BART and Bayesian SUR alternatives. These gains are particularly pronounced when the true signals are sparse, nonlinear, and distributed differently across outcomes, precisely the regime where shared-tree formulations or linear multivariate models are most restrictive.

Finally, beyond methodology and computation, \Cref{sec:theory} establishes, to our knowledge, the first full posterior contraction rates for a multivariate BART model with jointly estimated residual dependence. 
A recognised challenge in the Bayesian asymptotics of continuous shrinkage is that heavy-tailed priors, such as the vanilla horseshoe, often necessitate mathematically artificial posterior truncations to control global metric entropy \citep{vanderpasHS2017}. 
We overcome this by deploying an amplitude shelling argument, proving that the \texttt{multiVCBART} posterior achieves near-minimax adaptation to structural sparsity and functional complexity. Furthermore, by integrating functional restricted eigenvalue and $\beta$-min signal separation conditions, we extend these contraction guarantees beyond the joint predictive surface to establish the optimal recovery of the matrix $\Omega^{-1}\bm{B}(\bx)$, which encodes the effects of each covariate on each outcome \emph{conditionally} given all other outcomes. 

The remainder of this paper is organized as follows. 
\Cref{sec:model} specifies the full Bayesian model, including the prior structure on the coefficient-function ensembles and the residual precision matrix. 
\Cref{sec:comp} details our posterior computation strategy and practical implementation choices. 
Our main theoretical results, including posterior contraction guarantees, are developed in \Cref{sec:theory}. 
Empirical performance is assessed in \Cref{sec:experiments} through a synthetic
benchmark and re-analysis of drug sensitivity data in \Cref{sec:gdsc_realdata}. 
The paper concludes with a brief discussion in \Cref{sec:discussion}.
An open-source \textsf{R} implementation of \texttt{multiVCBART} is publicly available at \url{https://github.com/ghoshstats/multiVCBART}.

\section{Multivariate varying-coefficient BART}
\label{sec:model}
Our primary inferential objective is to jointly estimate the $p \times q$ matrix of unknown coefficient functions $\bm B(\cdot)$ and the $q \times q$ residual precision matrix $\Omega$ in \Cref{eq:multi-VCM}.
In many modern scientific applications, the ambient covariate dimension $p$ and the number of outcomes $q$ may be large relative to the sample size $n$. 
Consequently, flexible and robust estimation in this high-dimensional regime necessitates assuming structural sparsity, both in the set of active coefficient functions within $\bm B(\cdot)$ and in the conditional dependence graph encoded by the off-diagonal entries of $\Omega$.\footnote{We assume the modifier vector satisfies $\bx_i\in[0,1]^d$ only to simplify the presentation. Extensions to other continuous ranges, categorical modifiers, mixed input spaces, and more general splitting rules are natural; see, for example, \citet{flexBART2025}.}

\textbf{BART representation of the coefficient functions.} To flexibly capture nonlinear effect modification, we approximate each scalar coefficient function \(B_{jr}(\cdot)\) by its own BART ensemble:
\begin{equation}\label{eq:BARTrep}
B_{jr}(\bx)
=
\sum_{t=1}^M g(\bx;\mathcal T_{jrt},\mathcal M_{jrt}),
\quad j=1,\dots,p; \ r = 1,\dots,q.
\end{equation}
where \(\mathcal T_{jrt}\) is a regression tree and
\(\mathcal M_{jrt}\) is the corresponding collection of terminal-node, or leaf,
parameters. We write
\(
\mathcal M_{jrt}
=
\{\mu_{jrt\ell}:\ell\in\mathcal L(\mathcal T_{jrt})\},
\)
where \(\mathcal L(\mathcal T_{jrt})\) denotes the set of leaves of
\(\mathcal T_{jrt}\). Each tree induces an axis-aligned partition of the
modifier space,
\(
[0,1]^d
=
\bigcup_{\ell=1}^{L_{jrt}} A_{jrt\ell},
\)
where \(L_{jrt}=|\mathcal L(\mathcal T_{jrt})|\). The tree evaluation map is the
piecewise-constant step function
\[
g(\bx;\mathcal T_{jrt},\mathcal M_{jrt})
=
\sum_{\ell=1}^{L_{jrt}}
\mu_{jrt\ell}\mathbbm 1\{\bx\in A_{jrt\ell}\}.
\]
Consequently, plugging this expression in \eqref{eq:BARTrep},
\[
B_{jr}(\bx)
=
\sum_{t=1}^M
\sum_{\ell=1}^{L_{jrt}}
\mu_{jrt\ell}\mathbbm 1\{\bx\in A_{jrt\ell}\}.
\]

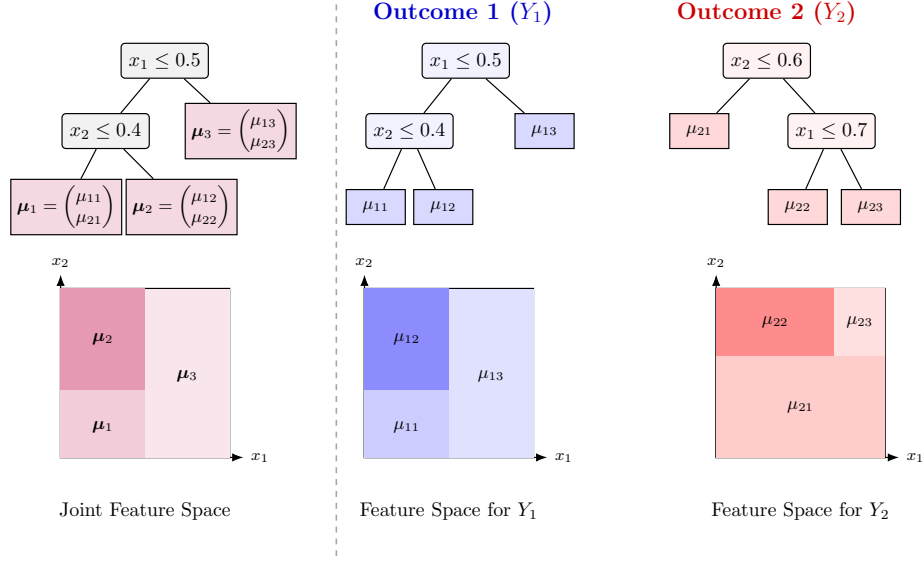
\begin{figure}[htbp]
\centering
\resizebox{0.75\textwidth}{!}{
\begin{tikzpicture}[
    x=1cm,y=1cm,
    >=Latex,
    line width=0.6pt,
    splitnode/.style={
        draw,
        rounded corners=2pt,
        fill=gray!10,
        minimum width=1.55cm,
        minimum height=0.62cm,
        inner sep=2pt,
        font=\small
    },
    leafnodeS/.style={
        draw,
        fill=purple!15,
        minimum width=1.45cm,
        minimum height=0.9cm,
        inner sep=3pt,
        font=\footnotesize,
        align=center
    },
    leafnode1/.style={
        draw,
        fill=blue!15,
        minimum width=1.05cm,
        minimum height=0.62cm,
        inner sep=2pt,
        font=\footnotesize,
        align=center
    },
    leafnode2/.style={
        draw,
        fill=red!15,
        minimum width=1.05cm,
        minimum height=0.62cm,
        inner sep=2pt,
        font=\footnotesize,
        align=center
    },
    paneltitle/.style={font=\bfseries\large},
    outtitle1/.style={font=\bfseries\large, text=blue!80!black},
    outtitle2/.style={font=\bfseries\large, text=red!80!black},
    axlbl/.style={font=\scriptsize},
    smalllbl/.style={font=\footnotesize},
    note/.style={font=\small, align=center}
]

\node[paneltitle] at (2.2,4.8) {};

\node[splitnode] (Sroot) at (2.2,3.6) {$x_1 \le 0.5$};
\node[splitnode] (Sleft) at (1.15,2.35) {$x_2 \le 0.4$};

\node[leafnodeS] (Sright) at (3.55,2.35)
{$\bm{\mu}_3=\begin{pmatrix}\mu_{13}\\ \mu_{23}\end{pmatrix}$};

\node[leafnodeS] (Sll) at (0.45,1.0)
{$\bm{\mu}_1=\begin{pmatrix}\mu_{11}\\ \mu_{21}\end{pmatrix}$};

\node[leafnodeS] (Slr) at (2.50,1.0)
{$\bm{\mu}_2=\begin{pmatrix}\mu_{12}\\ \mu_{22}\end{pmatrix}$};

\draw (Sroot) -- (Sleft);
\draw (Sroot) -- (Sright);
\draw (Sleft) -- (Sll);
\draw (Sleft) -- (Slr);

\begin{scope}[shift={(0.35,-3.45)}]
    \draw[->] (0,0) -- (3.25,0) node[right, axlbl] {$x_1$};
    \draw[->] (0,0) -- (0,3.25) node[above, axlbl] {$x_2$};

    \draw[thick] (0,0) rectangle (3,3);
    \draw[thick] (1.5,0) -- (1.5,3);
    \draw[thick] (0,1.2) -- (1.5,1.2);

    \fill[purple!20] (0,0) rectangle (1.5,1.2);
    \fill[purple!40] (0,1.2) rectangle (1.5,3);
    \fill[purple!10] (1.5,0) rectangle (3,3);

    \node[smalllbl] at (0.75,0.55) {$\bm{\mu}_1$};
    \node[smalllbl] at (0.75,2.1) {$\bm{\mu}_2$};
    \node[smalllbl] at (2.25,1.45) {$\bm{\mu}_3$};

    \node[note] at (1.5,-0.95) {Joint Feature Space};
\end{scope}

\draw[dashed, gray!70, thick] (5.25,4.6) -- (5.25,-5.2);

\node[paneltitle] at (10.2,4.8) {};
\node[outtitle1] at (7.5,4.45) {Outcome 1 ($Y_1$)};
\node[outtitle2] at (12.9,4.45) {Outcome 2 ($Y_2$)};

\node[splitnode, fill=blue!5] (T1root) at (7.55,3.6) {$x_1 \le 0.5$};
\node[splitnode, fill=blue!5] (T1left) at (6.55,2.35) {$x_2 \le 0.4$};
\node[leafnode1] (T1right) at (8.95,2.35) {$\mu_{13}$};
\node[leafnode1] (T1ll) at (5.95,1.0) {$\mu_{11}$};
\node[leafnode1] (T1lr) at (7.15,1.0) {$\mu_{12}$};

\draw (T1root) -- (T1left);
\draw (T1root) -- (T1right);
\draw (T1left) -- (T1ll);
\draw (T1left) -- (T1lr);

\begin{scope}[shift={(5.75,-3.45)}]
    \draw[->] (0,0) -- (3.25,0) node[right, axlbl] {$x_1$};
    \draw[->] (0,0) -- (0,3.25) node[above, axlbl] {$x_2$};

    \draw[thick] (0,0) rectangle (3,3);
    \draw[thick, blue!80!black] (1.5,0) -- (1.5,3);
    \draw[thick, blue!80!black] (0,1.2) -- (1.5,1.2);

    \fill[blue!20] (0,0) rectangle (1.5,1.2);
    \fill[blue!45] (0,1.2) rectangle (1.5,3);
    \fill[blue!12] (1.5,0) rectangle (3,3);

    \node[smalllbl] at (0.75,0.55) {$\mu_{11}$};
    \node[smalllbl] at (0.75,2.1) {$\mu_{12}$};
    \node[smalllbl] at (2.25,1.45) {$\mu_{13}$};

    \node[note] at (1.5,-0.95) {Feature Space for $Y_1$};
\end{scope}

\node[splitnode, fill=red!5] (T2root) at (12.9,3.6) {$x_2 \le 0.6$};
\node[leafnode2] (T2left) at (11.7,2.35) {$\mu_{21}$};
\node[splitnode, fill=red!5] (T2right) at (14.05,2.35) {$x_1 \le 0.7$};
\node[leafnode2] (T2rl) at (13.45,1.0) {$\mu_{22}$};
\node[leafnode2] (T2rr) at (14.75,1.0) {$\mu_{23}$};

\draw (T2root) -- (T2left);
\draw (T2root) -- (T2right);
\draw (T2right) -- (T2rl);
\draw (T2right) -- (T2rr);

\begin{scope}[shift={(12.0,-3.45)}]
    \draw[->] (0,0) -- (3.25,0) node[right, axlbl] {$x_1$};
    \draw[->] (0,0) -- (0,3.25) node[above, axlbl] {$x_2$};

    \draw[thick] (0,0) rectangle (3,3);
    \draw[thick, red!80!black] (0,1.8) -- (3,1.8);
    \draw[thick, red!80!black] (2.1,1.8) -- (2.1,3);

    \fill[red!20] (0,0) rectangle (3,1.8);
    \fill[red!45] (0,1.8) rectangle (2.1,3);
    \fill[red!12] (2.1,1.8) rectangle (3,3);

    \node[smalllbl] at (1.5,0.9) {$\mu_{21}$};
    \node[smalllbl] at (1.05,2.4) {$\mu_{22}$};
    \node[smalllbl] at (2.55,2.4) {$\mu_{23}$};

    \node[note] at (1.5,-0.95) {Feature Space for $Y_2$};
\end{scope}

\end{tikzpicture}
}
\caption{
Schematic comparison of tree partitions.
\textbf{Left:} A shared-tree architecture.
\textbf{Right:} The flexible \texttt{multiVCBART} architecture.
}
\label{fig:partition_comparison}
\end{figure}

Each scalar coefficient function is approximated by an ensemble
of piecewise-constant functions, and the partition used for one coefficient
surface need not be shared with any other surface. For example,
\(B_{jr_1}(\cdot)\) and \(B_{jr_2}(\cdot)\) may use entirely different splitting
variables, cutpoints, tree depths, and terminal regions. Thus, in general,
\(
\{A_{j r_1 t\ell}\}_{t,\ell}
\neq
\{A_{j r_2 t\ell}\}_{t,\ell}.
\)
This distinction is illustrated schematically in \Cref{fig:partition_comparison}. In the shared-tree formulation, a common partition of the modifier space is used across outcomes, and each terminal node carries a vector-valued leaf parameter. In contrast, our formulation assigns separate scalar-valued tree ensembles to each coefficient surface, allowing the induced partitions to adapt to outcome-specific nonlinearities and sparsity patterns.

This is substantially more flexible than shared-tree multivariate BART
formulations, such as \texttt{mvbcf}
\citep{mcjames2023bayesiancausalforestsmultivariate}, in which related outcomes
borrow strength through common or partially common tree partitions. Shared
partitions can be effective when outcomes have similar regression structure, but
they may be restrictive when different outcomes depend on different predictors,
different effect modifiers, or different nonlinear interaction patterns.

\subsection{Prior specification on \texorpdfstring{$(\bm B,\Omega)$}{(B, Omega)}}

We place independent prior components on the varying-coefficient surfaces and
the residual precision matrix,
\(
\Pi(\bm B,\Omega)
=
\Pi(\bm B)\Pi(\Omega).
\)

\paragraph{Prior on \(\bm B\).}
For each coefficient function \(B_{jr}(\cdot)\), let
\(
\mathcal E_{jr}
=
\{(\mathcal T_{jrt},\mathcal M_{jrt}):t=1,\ldots,M\}
\)
denote its BART ensemble. Conditional on the shared global shrinkage scale $\tau_B$ defined below, the ensembles are independent across \((j,r)\).

Each tree \(\mathcal T_{jrt}\) is assigned a depth-penalized Galton--Watson branching-process prior, following \citet{RockovaSaha2019}. Starting from the root, each node \(v\) at depth \(h=\mathrm{depth}(v)\) independently splits into two children with probability
\(
p_{\mathrm{split}}(h)=\gamma^h,\ 0<\gamma<1/2.
\)
Thus, deeper nodes are increasingly unlikely to split, which favors shallow trees and regularizes each tree toward being a weak learner. Conditional on a node splitting, the split variable is drawn from a coefficient-specific probability vector
\(
s(v)\mid \bm\pi_{jr}
\sim
\mathrm{Categorical}(\bm\pi_{jr}),
\)
and the cutpoint is drawn uniformly from the available cutpoints for the selected
modifier.

To encourage modifier sparsity within each coefficient function, we use the
heavy-tailed Dirichlet prior of \citet{Linero2018}. Specifically,
\(
\bm\pi_{jr}\mid \theta_{jr}
\sim
\mathrm{Dirichlet}\left(
{\theta_{jr}}/{d},\ldots,{\theta_{jr}}/{d}
\right),
\)
with hyperpriors
\(
{\theta_{jr}}/{(\theta_{jr}+d)}
\sim
\mathrm{Beta}(1,0.5).
\)
Small values of \(\theta_{jr}\) concentrate \(\bm\pi_{jr}\) near the simplex
corners, so that only a few modifiers are likely to be used for splitting in the
\((j,r)\)-th coefficient surface. The heavy tail on \(\theta_{jr}\) prevents the
prior from forcing sparsity when the data support a richer modifier set.

For each terminal-node parameter, we use a global--local horseshoe prior:
\[
\mu_{jrt\ell}\mid \lambda_{jr}^2,\tau_B^2
\sim
\mathcal N\left(0,\frac{\tau_B^2\lambda_{jr}^2}{M}\right),
\quad
\ell\in\mathcal L(\mathcal T_{jrt}),
\]
with
\[
\lambda_{jr}\sim\mathcal C^+(0,1),
\quad
\tau_B\sim\mathcal C^+(0,\sigma_B).
\]
The factor \(1/M\) stabilizes the prior variance of the ensemble as the number
of trees grows. The global scale \(\tau_B\) is shared across all \(pq\)
coefficient surfaces and controls the overall degree of sparsity, while the local
scale \(\lambda_{jr}\) allows important covariate--outcome surfaces to escape
global shrinkage. Thus, irrelevant coefficient surfaces are aggressively shrunk
toward zero, whereas truly active heterogeneous effects can remain large.

\paragraph{Prior on \(\Omega\).}
To model residual dependence across outcomes, we place a Graphical Horseshoe
prior \citep{LiJCGS} on the residual precision matrix \(\Omega\). The diagonal
entries receive weakly informative gamma priors,
\[
\omega_{rr}\sim\mathrm{Gamma}(a_0,b_0),
\quad r=1,\ldots,q,
\]
while the off-diagonal entries satisfy
\[
\omega_{rs}\mid \lambda_{rs}^2,\tau_\Omega^2
\sim
\mathcal N(0,\tau_\Omega^2\lambda_{rs}^2),
\quad r\neq s,
\]
with
\[
\lambda_{rs}\sim\mathcal C^+(0,1),
\quad
\tau_\Omega\sim\mathcal C^+(0,\sigma_\Omega).
\]
This prior shrinks many off-diagonal elements of \(\Omega\) toward zero,
encouraging a sparse residual conditional-dependence graph while still allowing
large residual associations when supported by the data.

Overall, the prior regularizes the model at two levels. The BART ensembles
flexibly represent the nonlinear varying-coefficient surfaces
\(B_{jr}(\cdot)\), while the global--local shrinkage prior suppresses irrelevant
covariate--outcome pairs and the Dirichlet splitting prior encourage the modifier
selection within active surfaces. Similarly, the Graphical Horseshoe prior
learns a sparse residual dependence structure across outcomes through
\(\Omega\).

\section{Posterior computation}
\label{sec:comp}

We outline the Metropolis-within-Gibbs sampler used to conduct posterior inference
for our multivariate varying-coefficient BART model. The sampler alternates between
updating the coefficient-function ensembles in \(\bm B(\cdot)\), updating the
residual precision matrix \(\Omega\), and updating the associated shrinkage and
splitting-probability hyperparameters. The key computational simplification is
that, conditional on the current value of \(\Omega\), the multivariate Gaussian
likelihood can be reduced to a sequence of scalar Gaussian working regressions.
Full conditional expressions are collected in \suppref{\Cref{sec:multi-algorithm}}{Supplementary Section S2}.

Let
\(
\bm E_i
=
\bm Y_i-\bm B(\bx_i)^\top\bz_i
\in\mathbb R^q
\)
denote the current residual vector for observation \(i\). We write
\(\omega_{rs}\) for the \((r,s)\)-entry of the precision matrix \(\Omega\). 

\paragraph{Outcome-wise pseudo-responses.}
Fixing an outcome \(r\in\{1,\ldots,q\}\), the Gaussian likelihood contribution for
observation \(i\) is proportional to
\(
|\Omega|^{1/2}
\exp\left\{
- \bm E_i^\top\Omega\bm E_i/2
\right\}.
\)
To isolate the terms involving \(E_{ir}\), expand the quadratic form as
\begin{align*}
\bm E_i^\top\Omega\bm E_i
&=
\sum_{a=1}^q\sum_{b=1}^q E_{ia}\omega_{ab}E_{ib}  \\
&=
\omega_{rr}E_{ir}^2
+
2E_{ir}\sum_{k\neq r}\omega_{rk}E_{ik}
+
\sum_{a\neq r}\sum_{b\neq r}E_{ia}\omega_{ab}E_{ib} \\
& = \omega_{rr}
\left(
E_{ir}
+
\omega_{rr}^{-1}\sum_{k\neq r}\omega_{rk}E_{ik}
\right)^2 + C_{i,-r},
\end{align*}
where \(C_{i,-r}\) does not depend on \(E_{ir}\). Therefore,
\[
E_{ir}\mid \bm E_{i,-r},\Omega
\sim
\mathcal N\left(
-\omega_{rr}^{-1}\sum_{k\neq r}\omega_{rk}E_{ik},
\ \omega_{rr}^{-1}
\right).
\]
Substituting \(E_{ir}=Y_{ir}-\eta_{ir}\) yields the scalar pseudo-response
\[
\widetilde Y_{ir}
:=
Y_{ir}
+
\omega_{rr}^{-1}\sum_{k\neq r}\omega_{rk}E_{ik},
\]
which satisfies
\(
\widetilde Y_{ir}\mid \bm E_{i,-r},\Omega
\sim
\mathcal N\left(\eta_{ir},\omega_{rr}^{-1}\right).
\)
Thus, conditional on the current values of the other outcomes and on \(\Omega\),
updating the mean function for outcome \(r\) reduces exactly to a univariate
Gaussian BART regression with working variance \(\omega_{rr}^{-1}\).


\paragraph{Updating the coefficient-function ensembles.}
Conditional on \(\widetilde Y_{ir}\) and \(\omega_{rr}\), we update the ensembles
\(\{B_{jr}(\cdot):j=1,\ldots,p\}\) for outcome \(r\) by Bayesian backfitting. For
the \(t\)-th tree in the \((j,r)\)-th ensemble, define the leave-one-tree-out
partial residual
\[
\widetilde r_{ir}^{(j,t)}
=
\widetilde Y_{ir}
-
\sum_{j'\neq j} z_{ij'}B_{j'r}(\bx_i)
-
\sum_{t'\neq t} z_{ij}\,
g(\bx_i;\mathcal T_{jrt'},\mathcal M_{jrt'}).
\]
Then the working leaf model is
\[
\widetilde r_{ir}^{(j,t)}
=
z_{ij}\mu_{jrt,\ell(\bx_i;\mathcal T_{jrt})}
+
\varepsilon_{ir},
\quad
\varepsilon_{ir}\sim\mathcal N(0,\omega_{rr}^{-1}).
\]
This is exactly the standard scalar Gaussian tree-update problem. Accordingly, each tree is updated by the usual Metropolis--Hastings \textsc{grow/prune} proposals with the leaf parameters integrated out (\suppref{\Cref{eq:multi-tree-mh}}{Eq. S2.6}), followed by Gibbs sampling of the leaf means conditional on the accepted structure (\suppref{\Cref{eq:multi-leaf-full}}{Eq. S2.7}). The split-probability vectors for the $(j,r)$-th ensembles are updated by drawing from their conjugate Dirichlet full conditionals (\suppref{\Cref{eq:multi-pi-full}}{Eq. S2.8}). Finally, the global--local horseshoe shrinkage scales are updated via the conditionally conjugate inverse-gamma augmentation scheme of \citet{Makalic_2016}. Qualitatively, the local, ensemble-specific scales $\lambda_{jr}$ are updated by pooling the squared leaf parameters across all $M$ trees within the $(j,r)$-th ensemble (\suppref{\Cref{eq:multi-lambda-full}}{Eq. S2.10}), allowing the model to shrink irrelevant predictors toward zero on an outcome-specific basis. Concurrently, the global scale $\tau_B$ is updated by aggregating the standardized signals across all $p \times q$ coefficient functions (\suppref{\Cref{eq:multi-tauB-full}}{Eq. S2.12}), actively adapting the baseline regularization to the overall sparsity level of $\bm{B}$.

\paragraph{Computational remark.}
It is useful to distinguish our update from existing related multivariate BART strategies. Shared-tree formulations model responses as joint $q$-dimensional vectors at each leaf, necessitating an $\mathcal{O}(q^3)$ dense matrix inversion to evaluate the marginal likelihood of every candidate split. Specifically, this involves computing the posterior leaf covariance $V_{I} = (\Sigma_{\mu}^{-1} + |I| \Omega)^{-1},$ where $|I|$ is the number of observations falling into leaf $I$, and $\Sigma_{\mu}$ is the $q \times q$ prior covariance matrix of the multivariate leaf parameters. Because this inversion must be computed for every candidate cutpoint, scanning a node of size $n$ incurs a severe $\mathcal{O}(nq^3)$ computational bottleneck. Alternatively, covariance-parameterized outcome-wise models (e.g., \texttt{suBART}; \citealp{esser2025seeminglyunrelatedbayesianadditive}) avoid vector-valued leaves but still necessitate computing dense Schur complements of $\Sigma=\Omega^{-1}$ to form conditional regressions. In contrast, we parameterize the likelihood directly via the precision matrix $\Omega,$ which naturally decouples the system into scalar pseudo-responses $\tilde{Y}_{ir}$ with scalar conditional variances $\omega_{rr}^{-1}$. This eliminates matrix operations from the tree-building phase entirely; given recursively updated sufficient statistics, the marginal likelihood evaluation cost per candidate split is reduced to $\mathcal{O}(1)$, bringing the total cost of scanning a node back down to the standard $\mathcal{O}(n)$ time of univariate BART.
\paragraph{Updating the residual precision matrix.}
After cycling through all coefficient-function ensembles, we recompute the residuals
\(\bm E_i = \bm Y_i - \bm{B}(\bx_i)^\top \bz_i,\)  and define \(S_E := \sum_{i=1}^n \bm E_i \bm E_i^\top.\)

Conditional on the current mean functions, the precision matrix is updated under a Graphical Horseshoe prior using the blocked Gibbs sampler of \citet{LiJCGS}. Writing the posterior kernel as
\[
p(\Omega\mid \bm{B},\ldots)
\propto
\mathbbm 1\{\Omega\succ0\}\,
|\Omega|^{n/2}
\exp \left\{-\frac12\operatorname{tr}(S_E\Omega)\right\}
\pi(\Omega\mid \Lambda,\tau_\Omega),
\]
we update \(\Omega\) column-by-column. For column \(c\), partition \(\Omega\) and \(S_E\) conformably, let
\[
\beta := \omega_{-c,c},
\quad
\gamma := \omega_{cc} - \beta^\top \Omega_{-c,-c}^{-1}\beta,
\]
and write \(\Lambda_c=\mathrm{diag}(\lambda_{rc}^2:r\neq c)\). Then the conditionally conjugate updates are
\[
\beta\mid \ldots
\sim
\mathcal N \big(-C\,s_{-c,c},\,C\big),
\quad
C=
\Big(s_{cc}\,\Omega_{-c,-c}^{-1}+(\tau_\Omega^2\Lambda_c)^{-1}\Big)^{-1},
\]
and
\[
\gamma\mid \ldots
\sim
\mathrm{Ga}\left(\frac n2 + a_0,\ \frac{s_{cc}+2b_0}{2}\right).
\]
The updated column is then recovered by \(\omega_{-c,c}=\beta\) and \(\omega_{cc}=\gamma+\beta^\top \Omega_{-c,-c}^{-1}\beta.\)
The local and global Graphical Horseshoe scales are subsequently updated by the standard \citet{Makalic_2016} inverse-gamma augmentation scheme (additional details in \suppref{\Cref{sec:GHSdetails}}{Supplementary Section S2.5}).

\paragraph{Recommended default hyperparameters.}
Our model depends on several hyperparameters controlling tree depth,
ensemble size, modifier-splitting sparsity, coefficient-surface shrinkage, and
residual-graph shrinkage. In our experiments, we have found the following choices
to work well; additional sensitivity analyses are reported in the supplement (\suppref{\Cref{sec:sensitivity}}{Section S3.4}). We use the same number of trees for every coefficient-function ensemble,
\(M_{jr}\equiv M=20\), for all \(j=1,\ldots,p\) and \(r=1,\ldots,q\). For the global horseshoe shrinkage scales, we use \(\sigma_B= 1\) for coefficient-surface shrinkage and \(\sigma_\Omega=1\)
for residual-graph shrinkage. Finally, for the diagonal entries of the Graphical
Horseshoe prior, we use weakly informative gamma hyperparameters
\(
(a_0,b_0)=(0.01,0.01),
\)
allowing the data to determine the marginal residual precisions.

\section{Theoretical results}
\label{sec:theory}
We establish a contraction rate for our proposed model by showing that the posterior concentrates around the true data-generating parameter $\Theta_0$ at a rate $\varepsilon_n ^{\dagger} \to 0$. Adopting the framework of \citep[GGV;][]{Ghosal2000}, we (i) lower bound prior mass in suitable KL neighborhoods of $\Theta_0$, (ii) construct exponentially consistent tests, and (iii) control the metric entropy of the pseudo-sparse sieve introduced in \suppref{\Cref{sec:revised-sieve}}{Supplementary Section S1}. Broadly, $\varepsilon_n ^{\dagger}$ reflects the balance between approximation accuracy and effective model complexity. For the BART mean component $\eta_i$, we achieve this by pairing tree-ensemble approximation theory for H\"older-smooth coefficient surfaces with entropy bounds that quantify the combinatorial richness of tree partitions and active predictor sets, building on the foundational analyses of \citet{Rockova2019} and \citet{Jeong2023}.

\subsection{Notations and Definitions}
For our working model, the true data-generating parameter is denoted by $\Theta_0=(\eta_0,\Omega_0)$. Throughout the theoretical analysis, the design points \(\{(\bm x_i,\bm z_i)\}_{i=1}^n\) are treated as
fixed. For any \(\Theta=(\eta,\Omega)\), let \(f_{\Theta,i}\) denote the
conditional density of \(\bm Y_i\) given \((\bm x_i,\bm z_i)\) under
\(\Theta\), which is a multivariate Normal density.
We measure the distance between the model and the truth using the empirical Hellinger distance $H(\Theta, \Theta_0)$, defined via $H^2(\Theta,\Theta_0) = n^{-1} \sum_{i=1}^n h^2(f_{\Theta, i}, f_{\Theta_0, i})$, where $h$ is the standard Hellinger distance between two $q$-variate Normal distributions.

For the true mean functions $B_{0,jr}$, we let $J_{0,jr} \in \{1,\dots,d\}$ denote the subset of active coordinates with intrinsic dimension $d_{0,jr} := |J_{0,jr}|$. We say $B_{0,jr} \in \mathcal{H}^{\alpha_{jr}}([0,1]^{d_{0,jr}}; K)$ if it is H\"older-$\alpha_{jr}$ smooth with radius $K$ on its active coordinates. The effective row support is $S_{B,r} := \{j = 1,\dots,p : B_{0,jr} \not\equiv 0\}$ with cardinality $s_r := |S_{B,r}|$. The total number of active mean functions is $S_B := \sum_{r=1}^q s_r$. For the true precision matrix $\Omega_0$, the degree of the $k$-th node is $d_k := \#\{k' \neq k : (\Omega_0)_{kk'} \neq 0\}$, and the total number of active off-diagonal edges is $s_q := \sum_{k=1}^q d_k / 2$. The maximum intrinsic dimension across all functions is $d_{\mathrm{eff}} := \max_{j,r} d_{0,jr}$. 

\subsection*{Model Assumptions}

\begin{assumenum}
\item{\textbf{Smoothness and per-component coordinate sparsity.}
For each response $r \in \{1,\ldots,q\}$ and predictor $j \in \{1,\ldots,p\}$, 
$B_{0,jr}(\bx)=B_{0,jr}(\bx_{J_{0,jr}})$ and $B_{0,jr}\in \mathcal H^{\alpha_{jr}}\big([0,1]^{d_{0,jr}};K\big)$; i.e., $B_{0,jr}$ depends only on the coordinates of $\bx$ indexed by $J_{0,jr}$ and is constant in all other coordinates.}\label{assum:multi-smoothness}

\item{{\textbf{Design constraints.}
\begin{enumerate}
    \item[(A2a)] \emph{Bounded design:} There exists $D<\infty$ such that $\bx_i\in[0,1]^d$ and $\bz_i\in[-D,D]^{p}$ for all $i$.
    \item[(A2b)] \emph{Ensemble size:} The number of trees $M$ used to construct each scalar function $B_{jr}$ is a fixed constant, $M = \mathcal{O}(1)$.
    \item[(A2c)] \emph{k-d regularity:} The design points $\{\bx_i\}_{i=1}^n\subset[0,1]^d$ are \emph{k-d--regular}. That is, there exist constants $0<c_\mathrm{kd}\le C_\mathrm{kd}<\infty$ such that for any axis-aligned rectangle $A\subset[0,1]^d$,
    \[
      c_\mathrm{kd}\,\mathrm{vol}(A) \le \frac{1}{n}\#\{i:\bx_i\in A\} \le C_\mathrm{kd}\,\mathrm{vol}(A).
    \]
\end{enumerate}}}\label{assum:multi-design}

\item{\textbf{Asymptotic Growth and Sparsity Regime.}
Let $\alpha_\star := \min_{j,r} \alpha_{jr}$. The true model complexities grow sufficiently slowly relative to $n$:
\begin{equation*}
  S_B \log p  = o(n), \quad
  s_q n^{A_{\mathrm{diag}}}\log n = \mathcal{O}\big(n^{1/2-\kappa_{\mathrm{spec}}}\big), \quad
  S_B (\log n) = o\Big(n^{\frac{2\alpha_\star}{2\alpha_\star + d_{\mathrm{eff}}}}\Big),
\end{equation*}
for some fixed constant $\kappa_{\mathrm{spec}}>0$. Here $A_{\mathrm{diag}}>0$ is the lower diagonal truncation exponent in \Cref{assum:GHS-scale}. 

Furthermore, we assume $p$ and $q$ to strictly dominate the effective complexities. Specifically, for a fixed constant $c_0 \in (0,1)$,
\begin{equation*}
 \max \left\{ q \log n, \log(pq) \log n \right\} \lesssim n(\varepsilon_n^\dagger)^2  \lesssim \min \left\{q^{2(1-c_0)} \log q, (pq)^{1-c_0} \log(pq) \right\}.
\end{equation*}
Additionally, we require the BART local-scale truncation to be large enough relative to the ambient sparsity, namely
\[
p/S_B \lesssim n^{\kappa_B}
\quad\text{for some}\quad
\kappa_B\in(0,A_B),
\]
where $A_B$ is the polynomial truncation exponent in \Cref{assum:multi-HSscale}.
}\label{assum:multi-growth}

\item{\textbf{True Precision Bounds.}
The true precision matrix $\Omega_0$ has strictly bounded eigenvalues:
\[
  0<\underline\lambda\le \lambda_{\min}(\Omega_0)\le \lambda_{\max}(\Omega_0)\le \bar\lambda<\infty.
\]}\label{assum:multi-eig}

\item{\textbf{Functional restricted eigenvalue (RE).}\label{assump:RE-func}
There exists a constant $\kappa_z>0$ such that for every coefficient-difference function
$\Delta \bm{B}(\cdot) \coloneqq \bm{B}(\cdot)-\bm{B}_{0}(\cdot)$ lying in the inner sieve $\mathcal{F}_n$,
\begin{equation*}
\frac1n\sum_{i=1}^n \|\Delta \bm{B}(\bx_i)^\top \bz_i\|_2^2
 \ge \kappa_z\cdot \frac1n\sum_{i=1}^n \|\Delta \bm{B}(\bx_i)\|_F^2.
\end{equation*}
Equivalently, $\|\Delta \bm{B}\|_{F,2,n}\le \kappa_z^{-1/2}\|\Delta \eta\|_{2,n}$,
where $\Delta \eta_i=\Delta \bm{B}(\bx_i)^\top \bz_i$.}\label{assum:multi-RE}

\item{\textbf{True Signal Separation.}\label{assum:beta-min}
For the effective thresholds $t_{\Omega,n}$ and $u_{B,n}$ defined in \suppref{\Cref{sec:revised-sieve}}{Supplementary Section S1}, there exist constants $a_\Omega, a_B > 0$ such that:
\begin{itemize}
    \item[(A6a)] $\min_{(k,k') \in S_{0,\Omega}} |\omega_{kk'}^{0}| \ge 4t_{\Omega,n}+a_{\Omega}\varepsilon_n^{\dagger}$
    \item[(A6b)] $\min_{(j,r) \in S_{0,B}} \|B_{0,jr}\|_{2,n} \ge 4Mu_{B,n}/\sqrt{\kappa_z} + a_{B} \varepsilon_n^{\dagger}.$
\end{itemize}
}\label{assum:multi-betamin}

\end{assumenum}

To illustrate that the conditions in \Cref{assum:multi-growth} remain highly feasible, consider a sparse high-dimensional regime in which the response dimension grows as $q \asymp n^{1/2}$, and the ambient covariate dimension grows polynomially as $p \asymp n^\beta,$ for some fixed $\beta>1/2$. Let the true active mean complexity satisfy $S_B \asymp q \asymp n^{1/2}$, and let the precision graph remain sparse with $s_q \asymp n^\xi$ for some $\xi<1/2-A_{\mathrm{diag}}$.
Assuming standard smoothness $\alpha_\star=1$ and intrinsic dimension $d_{\mathrm{eff}}=1$, the smoothness bottleneck is $n^{{2\alpha_\star}/{(2\alpha_\star+d_{\mathrm{eff}})}} = n^{2/3},$ which strictly dominates \(S_B \log n \asymp n^{1/2}\log n.\)
Moreover, \(S_B \log p \asymp n^{1/2}\log n = o(n),\) and the graph-sparsity condition holds because
\[
s_q n^{A_{\mathrm{diag}}}\log n
\asymp
n^{\xi+A_{\mathrm{diag}}}\log n
=
\mathcal{O}\big(n^{1/2-\kappa_{\mathrm{spec}}}\big)
\]
for some $\kappa_{\mathrm{spec}}>0$ whenever $\kappa_{\mathrm{spec}}<1/2-A_{\mathrm{diag}}-\xi$. In particular, since $q\asymp n^{1/2}$ and $\xi<1/2-A_{\mathrm{diag}}<1/2$, this condition also implies the usual sparse-graph complexity requirement \(s_q\log q=o(n)\). In addition, since \({p}/{S_B} \asymp n^{\beta-1/2}\), it suffices to choose $\beta$ so that $\beta-1/2< A_B$. 

\Cref{assum:multi-smoothness} restricts the true varying-coefficient functions to H\"older spaces of anisotropic intrinsic dimension. The intrinsic sparsity condition ($d_{0,jr} \ll d$) is fundamentally necessary to bypass the curse of dimensionality, which is a standard requirement in the theoretical analysis of tree-based models \citep{Jeong2023}. \Cref{assum:multi-design} provides the structural regularity needed to transition between continuous function spaces and empirical observations. Crucially, the k-d regularity condition guarantees that the empirical measure of the design points behaves analogously to the Lebesgue measure \citep{RockovaSaha2019}. The condition $S_B \log p = o(n)$ is the well-known necessary threshold for sparse variable selection \citep{castillovdv2012}, while the stronger graph condition \(s_q n^{A_{\mathrm{diag}}}\log n=\mathcal O(n^{1/2-\kappa_{\mathrm{spec}}})\) implies the usual sparse-graph requirement \(s_q\log q=o(n)\) under the displayed dimension-growth bounds \citep{banerjeeghosal2015}. It also serves as a compatibility condition between graph sparsity and the lower diagonal truncation in the graphical horseshoe prior, ensuring that the SPD-restricted prior places at most polynomially small mass near the boundary of the positive-definite cone $\Omega \succ 0$. Furthermore, the upper bounds on the global complexity rate ensure that the ambient dimensions $p$ and $q$ grow fast enough relative to the true signals so that our continuous horseshoe priors can effectively shrink the vast noise space without violating the combinatorial subset capacity \citep{vanderpasHS2017}. \Cref{assum:multi-eig} assumes $\Omega_0$ is well-conditioned, which is ubiquitous in multivariate regression and Gaussian graphical models \citep[e.g.,][]{cai2016precision,ravikumar2011glasso}.

Crucially, Assumptions \ref{assum:multi-RE} and \ref{assum:multi-betamin} are required to untangle the coefficient signals and bypass the limitations of heavy-tailed continuous shrinkage priors. Without further conditions on the design, the map $\bm{B}(\cdot)\mapsto \eta(\cdot)=\bm{B}(\cdot)^\top \bz$ need not be injective along sparse functional directions. Assumption \ref{assum:multi-RE} imposes a functional restricted eigenvalue inequality, which rules out such degeneracies and ensures the mean functions are identifiable. Furthermore, \Cref{assum:multi-betamin} imposes a $\beta$-min separation condition requiring the true nonzero precision entries to exceed the edge threshold by at least \(a_\Omega\varepsilon_n^\dagger\), and true nonzero coefficient surfaces to exceed the mean threshold by at least \(a_B\varepsilon_n^\dagger\). Hence, models that drop such signals are separated from \(\Theta_0\) in Hellinger distance at the contraction scale and can be controlled by the shell-wise tests (see \Cref{lem:shell-separation}).


\subsection*{Prior Assumptions}

We assume the standard BART priors and the graphical horseshoe structure already defined in \Cref{sec:model}. We record only the relevant prior configuration assumptions used in the contraction analysis.

\begin{pssumenum}

\item{\textbf{Horseshoe shrinkage on leaves.}
Assume the following truncations on the BART scale parameters for constants $A_B>0$ and $C_{\tau} > 0$:
\[
\lambda_{jr} \sim \mathcal{C}^{+}(0,1)\mathbbm{1}_{[0,n^{A_B}]}, \quad \tau_{B} \sim \mathcal{C}^{+}(0,\sigma_{B,n}) \mathbbm{1}_{[0,C_{\tau}\sigma_{B,n}]},
\]
where $\sigma_{B,n} \asymp \sum_{r} s_r / p = S_{B}/p$.}\label{assum:multi-HSscale}

\item{\textbf{Graphical horseshoe on $\Omega$.} For the off-diagonal entries $\omega_{kk'}$ where $k \neq k'$ and constants $A_\Omega,C_{\tau} >0$,
\[
\lambda_{\Omega,kk'} \sim \mathcal C^+(0,1)\,\mathbbm 1_{[0,n^{A_\Omega}]}, \quad
\tau_\Omega \sim \mathcal C^+(0,\sigma_{\Omega,n})\,\mathbbm 1_{[0,C_{\tau}\sigma_{\Omega,n}]},
\]
where $\sigma_{\Omega,n}\asymp {s_q}/{(q^2\sqrt{n})}.$

For the diagonals, independently for $k$, 
\[
\omega_{kk}\sim \mathrm{Ga}(a_0,b_0)\mathbbm{1}_{[n^{-A_{\mathrm{diag}}},\,n^{A_{\mathrm{diag}}}]},
\]
where $A_{\mathrm{diag}}>0$ is chosen to satisfy the compatibility condition in \Cref{assum:multi-growth}.
}\label{assum:GHS-scale}
\end{pssumenum}

The explicit polynomial truncations in Assumptions \ref{assum:multi-HSscale} and \ref{assum:GHS-scale} serve distinct technical functions. The exponents $A_B$ and $A_\Omega$ bound the local scales of the BART and graphical horseshoe, respectively, ensuring that the metric entropy of the induced sieve does not explode while remaining large enough to avoid truncating the local scales needed to protect true signals from over-shrinkage. The diagonal exponent $A_{\mathrm{diag}}$ in \Cref{assum:GHS-scale}, together with the graph-sparsity condition $s_q n^{A_{\mathrm{diag}}}\log n=\mathcal O(n^{1/2-\kappa_{\mathrm{spec}}})$ in \Cref{assum:multi-growth}, is used to prove that the prior mass assigned to nearly singular precision matrices is polynomially small.

Finally, we define the posterior contraction rate $\varepsilon_n^\dagger$ as the maximum of the smoothness, mean sparsity, and precision sparsity penalties:
\[
\varepsilon_n^{\dagger}
:= \max\left\{
\sqrt{(\log n)\sum_{r=1}^q \sum_{j\in S_{B_0,r}} n^{-\frac{2\alpha_{jr}}{2\alpha_{jr}+d_{0,jr}}}},\  
\sqrt{\frac{1}{n}\sum_{r=1}^q s_r \log\Big(\frac{e p}{s_r}\Big)},\  
\sqrt{\frac{1}{n}\sum_{k=1}^q d_k \log\Big(\frac{e q}{d_k}\Big)}
\right\}.
\]

\paragraph{Sieve construction.}
A key step in any posterior contraction proof based on the GGV framework is the construction of a sieve, which is a subset of the parameter space that is rich enough to contain suitable Kullback--Leibler neighborhoods of the truth, yet sufficiently regular to admit controlled metric entropy and manageable prior mass outside it. In our setting, the horseshoe priors retain heavy polynomial tails, so a single sieve $(\mathcal{F}_n)$ defined by a hard global upper bound on the magnitudes of $|b_{jrt\ell}|$ and $|\omega_{kk'}|$ is not adequate for proving contraction of the full posterior.

We therefore work with two nested sieves. The \emph{outer sieve} $\mathcal G_n$ controls the structural and polynomial-envelope complexity of the parameter space by restricting the effective support sizes for the precision and mean blocks, a total active-leaf budget, the number of distinct split variables used by each coefficient surface, and the polynomial envelopes induced by the truncations in Assumptions \ref{assum:multi-HSscale}--\ref{assum:GHS-scale}. Inside $\mathcal G_n$, we define an \emph{inner sieve} $\mathcal F_n \subset \mathcal G_n$ by additionally requiring the inactive noise energies of the precision and mean components to remain below baseline tolerances, and by imposing the working spectral envelope $\|\Omega^{-1}\|_{\mathrm{op}}\le \bar R_{\Omega,n}$. Equivalently, if $S_{k,m,a}$ denotes the shell indexed by precision-noise level $k$, mean-noise level $m$, and active-signal amplitude $a$, where $k,m,a\in\{0,1,2,\dots\}$, then
\[
\mathcal F_n=\bigcup_{a=0}^{A_{0,n}} S_{0,0,a},
\]
for a sufficiently large polynomial envelope $A_{0,n}$. Here $k=m=0$ selects the baseline-noise shells, while $A_{0,n}$ is chosen to dominate the polynomial amplitude bounds induced by $\max\{A_B,A_\Omega,A_{\mathrm{diag}}\}$. The full formal definitions are given in \suppref{\Cref{sec:revised-sieve}}{Supplementary Section S1}.

Our contraction analysis proceeds in two stages. We first establish posterior contraction under the posterior restricted to the effective inner sieve $\mathcal F_n$ in \Cref{prop:postconc_sieve}. We then upgrade this to the full posterior in \Cref{thm:postconc_full} by partitioning the outer region $\mathcal G_n\setminus \mathcal F_n$ into shells and showing that, although the horseshoe prior allocates only polynomially decaying mass across these shells, the Hellinger separation and shell-wise tests grow quickly enough to control their total posterior contribution. The geometry of this argument is illustrated schematically in \Cref{fig:ggv-vs-sliced-sieve-compact}.

We first state the contraction result under the posterior restricted to the inner sieve $\mathcal F_n$. For any measurable set \(A\), let
\[
\Pi^{\mathcal F_n}(A) := \frac{\Pi(A\cap \mathcal F_n)}{\Pi(\mathcal F_n)}
\]
denote the prior restricted to the sieve and renormalized, and let $\Pi^{\mathcal F_n}(A\mid \bY)$ be the corresponding posterior distribution.

\begin{proposition}[Sieve-Truncated Posterior Contraction]
\label{prop:postconc_sieve}
Under assumptions \ref{assum:multi-smoothness}--\ref{assum:multi-RE} and prior configurations \ref{assum:multi-HSscale}--\ref{assum:GHS-scale}, there exists a sufficiently large constant \(M<\infty\) such that
\[
  \Pi^{\mathcal F_n}\Big(H(\Theta,\Theta_0)>M\,\varepsilon_n^\dagger\ \Bigm|\ \bY \Big)
  \ \xrightarrow[n\to\infty]{P_{\Theta_0}}\ 0.
\]
\end{proposition}

Although \Cref{prop:postconc_sieve} successfully bounds the error, it requires artificially truncating the posterior to the inner sieve $\mathcal F_n$. This is necessary because under the horseshoe prior configuration in Assumptions \ref{assum:multi-HSscale} and \ref{assum:GHS-scale}, any sieve tight enough to deliver the metric-entropy control required for testing (\suppref{\Cref{lem:entropy-H-correct-global}}{Lemma S1.11}) leaves a complement whose prior mass decays only polynomially, rather than exponentially, thereby violating the standard unconditional GGV requirements \citep{Ghosal2000}. 

To upgrade this result to the full posterior without ad hoc regularization, we employ an amplitude shelling argument, which appears in several places in the literature on contraction theory \citep[see, e.g., Section 6.1 of][]{agapiou2026}. By slicing the structural outer sieve $\mathcal G_n$ into shells indexed by their inactive precision noise, inactive mean noise, and active signal amplitude, we use Assumptions \ref{assum:multi-RE} and \ref{assum:multi-betamin} to show that the Hellinger separation increases with the shell index (\suppref{\Cref{lem:shell-separation}}{Lemma S1.12}). Combined with shell-wise entropy control via the zero net construction (\suppref{\Cref{lem:shell-entropy}}{Lemma S1.13}), this yields exponentially consistent tests on each shell, which are strong enough to dominate the polynomial prior tails and thereby establish contraction of the full posterior.

\begin{figure}[t]
\centering
\begin{tikzpicture}[
    >=Latex,
    font=\small,
    line width=0.8pt,
    title/.style={font=\bfseries\normalsize},
    lab/.style={font=\small\bfseries}
]

\begin{scope}[xshift=0cm]
    \node[title] at (0, 2.8) {Classical GGV Sieve};

    \fill[gray!10] (0,0) circle (2.2);
    \draw[gray!50, thick] (0,0) circle (2.2);
    \node[text=gray!70, font=\scriptsize] at (0, -2.5) {Full Parameter Space};

    \fill[teal!20] (0,0) circle (1.2);
    \draw[teal!80!black, thick] (0,0) circle (1.2);
    \node[lab, text=teal!80!black] at (0,0) {$\mathcal{F}_n$};
\end{scope}

\begin{scope}[xshift=7.5cm]
    \node[title] at (0, 2.8) {Amplitude Shelling};

    \fill[gray!10] (0,0) circle (2.2);
    \draw[gray!50, thick] (0,0) circle (2.2);
    \node[text=gray!70, font=\scriptsize] at (0, -2.5) {Full Parameter Space};

    \fill[orange!10] (0,0) circle (2.0);
    \draw[orange!80!black, dashed, very thick] (0,0) circle (2.0);

    \fill[cyan!20] (0,0) circle (1.6);
    \draw[cyan!60!black, thin] (0,0) circle (1.6);

    \fill[cyan!40] (0,0) circle (1.2);
    \draw[cyan!60!black, thin] (0,0) circle (1.2);

    \fill[cyan!60] (0,0) circle (0.8);
    \draw[cyan!60!black, thin] (0,0) circle (0.8);

    \fill[teal!80!black] (0,0) circle (0.4);
    \node[lab, text=white, font=\scriptsize] at (0,0) {$\mathcal{F}_n$};

    \draw[<-, orange!80!black, thick] (120:2.0) -- (120:2.5)
        node[above left, font=\small\bfseries, inner sep=2pt] {$\mathcal{G}_n$};

    \draw[<-, cyan!80!black, thick] (30:1.4) -- (30:2.55)
        node[right, font=\scriptsize, inner sep=2pt] {Shells $S_{k,m,a}$};
\end{scope}

\end{tikzpicture}
\caption{
Schematic comparison of the classical GGV sieve argument and our two-sieve shelling construction.
\textbf{Left:} The classical approach works with a single sieve $\mathcal F_n$ whose complement receives exponentially decaying prior probability; metric entropy and testing are handled directly on that set.
\textbf{Right:} In our heavy-tailed setting, we instead define a structural outer sieve $\mathcal G_n$ and an effective inner sieve $\mathcal F_n \subset \mathcal G_n$. The intermediate region $\mathcal G_n \setminus \mathcal F_n$ is partitioned into shells indexed by inactive precision-noise, inactive mean-noise, and active-signal amplitude. Shell separation and shell-wise entropy bounds then allow the posterior contribution of these outer shells to be controlled despite the polynomial prior tails.
}
\label{fig:ggv-vs-sliced-sieve-compact}
\end{figure}
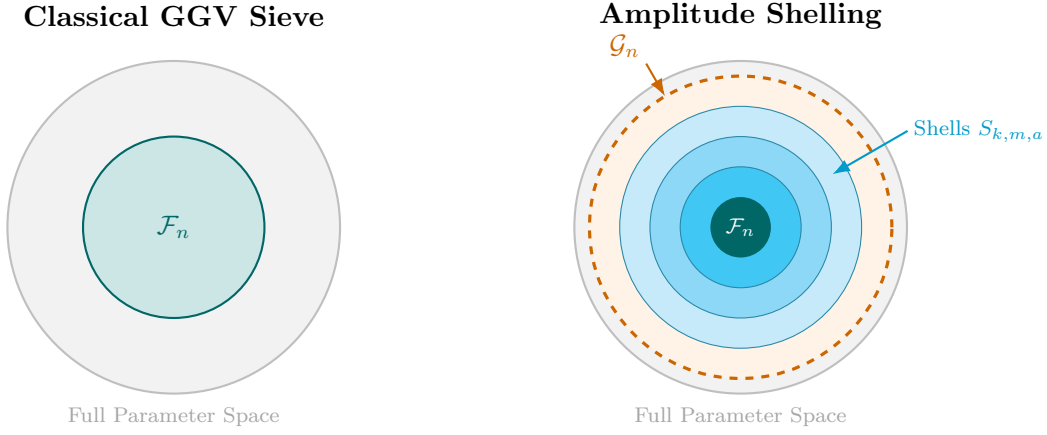

\begin{theorem}[Full Posterior Contraction]
\label{thm:postconc_full}
Under assumptions \ref{assum:multi-smoothness}--\ref{assum:multi-betamin} and prior configurations \ref{assum:multi-HSscale}--\ref{assum:GHS-scale}, there exists a sufficiently large constant \(M<\infty\) such that the full posterior contracts at the rate $\varepsilon_n^\dagger$:
\[
  \Pi\Big(H(\Theta,\Theta_0)>M\,\varepsilon_n^\dagger\ \Bigm|\ \bY \Big)
  \ \xrightarrow[n\to\infty]{P_{\Theta_0}}\ 0.
\]
\end{theorem}
\Cref{thm:postconc_full} has been proved explicitly in \Cref{sec:proofthm2}. Furthermore, while Theorem \ref{thm:postconc_full} establishes the recovery of the joint predictive surface $\eta(\bx_i) = \bm{B}(\bx_i)^\top \bz_i$ and the precision matrix $\Omega$, in many multivariate settings the primary target of inference is the transformed coefficient matrix $G(\bx) := \Omega^{-1}\bm{B}(\bx)^\top$. This quantity represents the direct marginal effects of the covariates $\bz$ on the multivariate response $\bY$, explicitly adjusting for the residual dependence structure captured by $\Omega$. 

Combining the functional RE condition (\Cref{assum:multi-RE}) with \Cref{thm:postconc_full} seamlessly upgrades contraction of the linear predictor $\eta$ to contraction of the underlying coefficient matrix function $\bm{B}(\cdot)$, and hence of its transformation $G(\bx)$. Because Theorem \ref{thm:postconc_full} now holds for the full posterior, this corollary extends directly to the un-truncated distribution.

\begin{corollary}[Contraction of the Transformed Coefficient Surface]
\label{cor:transformed-surface}
Let $\bm{G}(\bx):=\Omega^{-1}\bm{B}(\bx)^\top$ and $\bm{G}_0(\bx):=\Omega_0^{-1}\bm{B}_0(\bx)^\top$.
Define the empirical Frobenius norms over the design points
\[
\|\bm{G}-\bm{G}_0\|_{F,2,n}^2
:=\frac{1}{n}\sum_{i=1}^n \|\bm{G}(\bx_i)-\bm{G}_0(\bx_i)\|_F^2,
\quad
\|\bm{B}-\bm{B}_0\|_{F,2,n}^2
:=\frac{1}{n}\sum_{i=1}^n \|\bm{B}(\bx_i)-\bm{B}_0(\bx_i)\|_F^2.
\]
Assume \ref{assum:multi-smoothness}--\ref{assum:multi-betamin}, and let $\kappa_z>0$ denote the functional restricted eigenvalue constant from \Cref{assum:multi-RE}. Let
\[
R_{B_0} := \max_{1\le i\le n}\|\bm{B}_0(\bx_i)\|_F.
\]
Then there exists a sufficiently large constant $M'<\infty$ such that
\begin{align}
\Pi\Big(
\|\bm{B}-\bm{B}_0\|_{F,2,n}
>
M'\,\kappa_z^{-1/2}\,\varepsilon_n^\dagger
\ \Bigm|\ \bY
\Big)
&\xrightarrow[n\to\infty]{P_{\Theta_0}} 0,
\label{eq:B-conc-cor}\\
\Pi\Big(
\|\bm{G}-\bm{G}_0\|_{F,2,n}
>
M'\big(\kappa_z^{-1/2}+R_{B_0}\big)\varepsilon_n^\dagger
\ \Bigm|\ \bY
\Big)
&\xrightarrow[n\to\infty]{P_{\Theta_0}} 0.
\label{eq:G-conc-cor}
\end{align}
\end{corollary}
The complete proof of \Cref{cor:transformed-surface} is provided in \Cref{sec:prooftransformed-surface}. The additional factor $(\kappa_z^{-1/2}+R_{B_0})$ in the rate of \eqref{eq:G-conc-cor} arises from a simple decomposition: the error $\bm{G}-\bm{G}_0$ contains (i) a term from $\bm{B}-\bm{B}_0$ propagated through $\Omega^{-1}$ and hence controlled by $\kappa_z^{-1/2}\varepsilon_n^\dagger$, and (ii) a term from $\Omega^{-1}-\Omega_0^{-1}$ acting on the true surface $\bm{B}_0$, whose magnitude scales with \( R_{B_0}:=\max_i\|\bm{B}_0(\bx_i)\|_F \).
Under \Cref{assum:multi-smoothness} and row-wise sparsity, \( \|\bm{B}_0(\bx_i)\|_F^2 = \sum_{r=1}^q\sum_{j\in S_{B,r}} B_{0,jr}(\bx_i)^2 \le M^2 S_B \), hence $R_{B_0}\le M\sqrt{S_B}$. Thus, \Cref{cor:transformed-surface} yields \( \|\bm{G}-\bm{G}_0\|_{F,2,n}=o_{P}(1) \) under the full posterior provided \( \sqrt{S_B}\,\varepsilon_n^\dagger \rightarrow 0 \).

\Cref{cor:transformed-surface} is useful in practice because posterior summaries of \(G_{rj}(\bx)\) provide uncertainty-quantified marginal effect surfaces for predictor \(j\) on response \(r\), rather than only the contraction of the joint linear predictor \(\eta\). Thus, variable ranking or effect-surface inference based on posterior means, credible bands, or thresholded posterior inclusion summaries of \(G(\bx)\) is asymptotically justified whenever the additional condition \(\sqrt{S_B}\varepsilon_n^\dagger\to0\) holds.

\section{A synthetic experiment}
\label{sec:experiments}
In this section, we evaluate the predictive accuracy, uncertainty quantification, and computational cost of \texttt{multiVCBART} on a controlled synthetic example. We compare against several representative baselines. First, we fit standard univariate BART independently to each response. Second, we consider \texttt{suBART} \citep{esser2025seeminglyunrelatedbayesianadditive}, which assigns each response its own BART ensemble and couples outcomes through a residual covariance matrix \(\Sigma=\Omega^{-1}\). Following \citet{esser2025seeminglyunrelatedbayesianadditive}, we use the hierarchical covariance prior, following \citet{HuangWand2013} rather than a vanilla inverse-Wishart prior. Third, we include \texttt{mvBART}, taken as the Gaussian non-skew special case of \citet{Um2023}, in which outcomes share tree partitions and terminal nodes have multivariate leaf parameters. Fourth, we include a two-step procedure, \texttt{BART-GHS}, which first estimates the outcome-specific conditional means using independent BART fits and then estimates the residual precision matrix from posterior mean residuals using a Graphical Horseshoe prior. This separates mean estimation from residual graph estimation, unlike \texttt{multiVCBART}, which learns \((\bm B,\Omega)\) jointly. Finally, we include Bayesian linear SUR \citep{AndoBayesSUR}, implemented using standard \texttt{BayesSUR} defaults \citep{ZhaoJSS}, to assess the cost of imposing linearity when the true mean structure is nonlinear. For \texttt{multiVCBART}, we run four independent MCMC chains with 2500 iterations each, discarding the first 500 iterations as burn-in.

Our simulation design mimics the multivariate ``Friedman-type" \citep{FriedmanAOS} simulation philosophy to a sparse varying-coefficient framework. While traditional benchmarks assume each response depends on a small subset of covariates with correlated Gaussian noise, our design rigorously tests the model's ability to recover sparsity across both primary covariates $\bz_i$ and modifiers $\bx_i$. Our base structure is inspired by the experiments in \citet[Section 5.1]{esser2025seeminglyunrelatedbayesianadditive}. For each replication, we generate independent modifier variables $\bx_i=(x_{i1},\ldots,x_{id})^\top\in[0,1]^d$ with $x_{ik}\stackrel{\mathrm{iid}}{\sim}\mathcal{U}(0,1)$ for $d=50$, and primary covariates $\bz_i=(z_{i1},\ldots,z_{ip})^\top\in\mathbb{R}^p$ with $z_{ij}\stackrel{\mathrm{iid}}{\sim}\mathcal{N}(0,1)$ for $p=100$. We consider $q=2$ outcomes. The mean vector $\boldsymbol\eta_0(\bx_i, \bz_i)=(\eta_{0,1}(\bx_i, \bz_i),\eta_{0,2}(\bx_i, \bz_i))^\top$ is constructed using a sparse varying-coefficient structure, $\eta_{0,k}(\bx_i, \bz_i) = \bz_i^\top \boldsymbol\beta_k(\bx_i)$ for $k \in \{1,2\}$.

We distribute the non-linear effects across a small active subset of the primary covariates. Specifically, the non-zero varying coefficients for the first outcome are:
\(
\beta_{1,1}(\bx_i) := 10\sin \big(\pi x_{i1}x_{i2}\big), \ 
\beta_{1,2}(\bx_i) := 20(x_{i3}-1/2)^2, \
\beta_{1,3}(\bx_i) := 10x_{i4}, \
\beta_{1,4}(\bx_i) := 5x_{i5}, \
\beta_{1,5}(\bx_i) := 5,
\)
with $\beta_{1,j}(\bx_i) = 0$ for all remaining $j \in \{6, \ldots, 100\}$. For the second outcome, the active coefficients are shifted to a different subset of primary covariates:
\(
\beta_{2,6}(\bx_i) := 10\cos \big(\pi x_{i3}x_{i4}\big), \
\beta_{2,7}(\bx_i) := 20(x_{i5}-1/2)^2, \
\beta_{2,8}(\bx_i) := 10x_{i6}, \
\beta_{2,9}(\bx_i) := 5x_{i7}, \
\beta_{2,10}(\bx_i) := 5,
\)
with $\beta_{2,j}(\bx_i) = 0$ for $j \notin \{6, \ldots, 10\}$. Hence, the active modifier sets are $\{1,2,3,4,5\}$ for the first response and $\{3,4,5,6,7\}$ for the second, while the remaining $43$ modifiers are pure noise. Simultaneously, only $10$ of the $100$ primary covariates in $\bz_i$ possess non-zero effects. 

We then generate responses as $\bm{Y}_i=\boldsymbol\eta_0(\bx_i, \bz_i)+\boldsymbol\varepsilon_i$, with $\boldsymbol\varepsilon_i\stackrel{\mathrm{iid}}{\sim}\mathcal{N}_2(\mathbf{0},\Sigma_0)$ drawn from a correlated Gaussian noise distribution to emulate a seemingly-unrelated regression structure. We define the true precision matrix $\Omega_0$ such that $(\Omega_0)_{12}=(\Omega_0)_{21}=0.6$ and $(\Omega_0)_{11}=(\Omega_0)_{22}=1$, and set the covariance matrix $\Sigma_0 = 4\Omega_0^{-1}$, ensuring the residual dependence is non-negligible and models that ignore cross-outcome correlation can become mis-calibrated. We generate a training set of size $n_{\mathrm{train}}=300$ and an independent test set of size $n_{\mathrm{test}}=500$, and repeat the full experiment over $25$ independent random seeds. We summarize predictive performance using RMSE and the continuous ranked probability score (CRPS) \citep{GneitingRaftery2007}, which is a proper scoring rule rewarding both calibration and sharpness.
We assess uncertainty quantification using empirical coverage of nominal \(95\%\) predictive intervals, matching the interval summaries computed in our implementation.

\begin{figure}[t]
\centering
\begin{subfigure}[b]{0.24\linewidth}\centering\includegraphics[width=\linewidth]{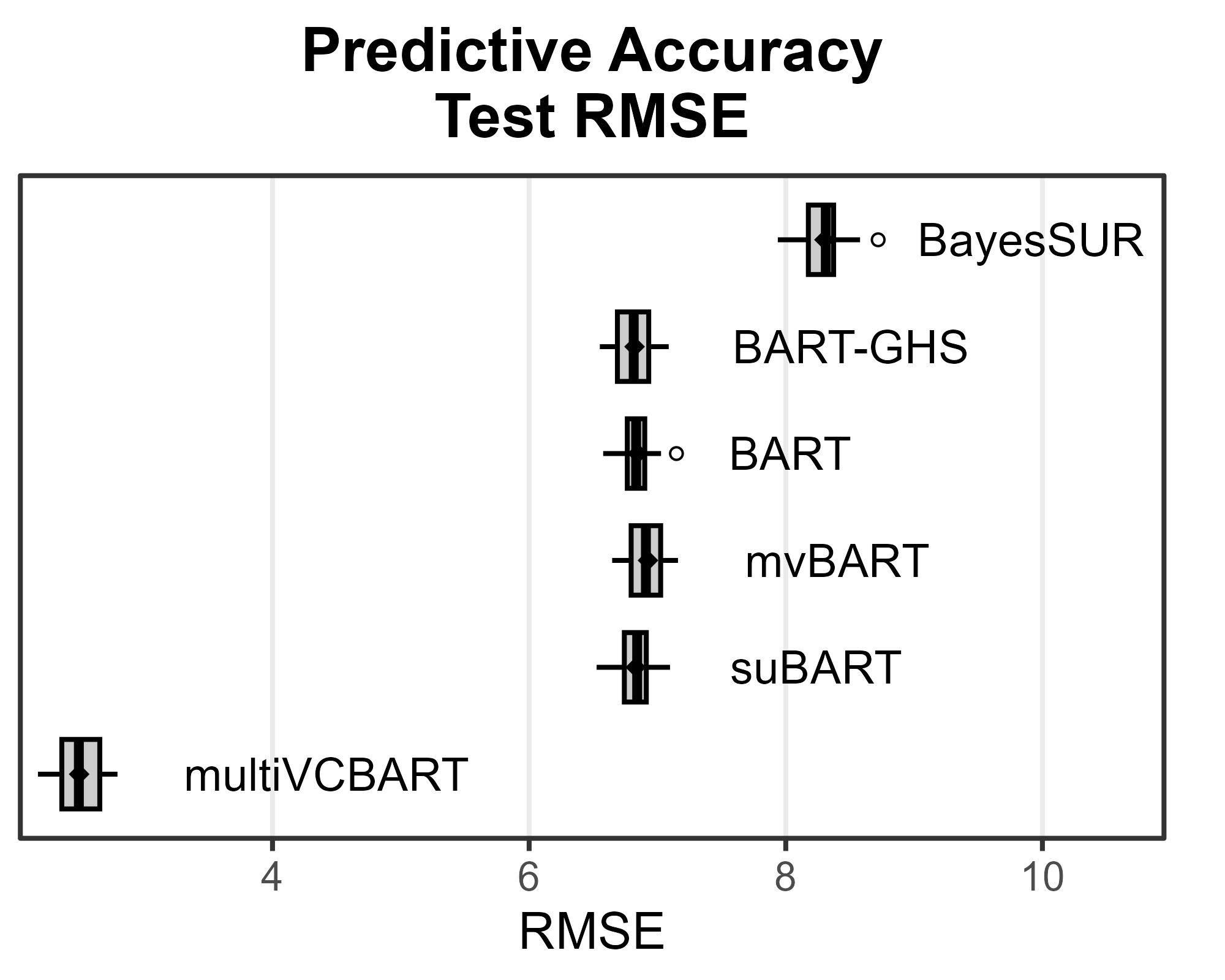}
\caption{RMSE}\label{fig:friedman_rmse}
\end{subfigure}\hfill
\begin{subfigure}[b]{0.24\linewidth}\centering\includegraphics[width=\linewidth]{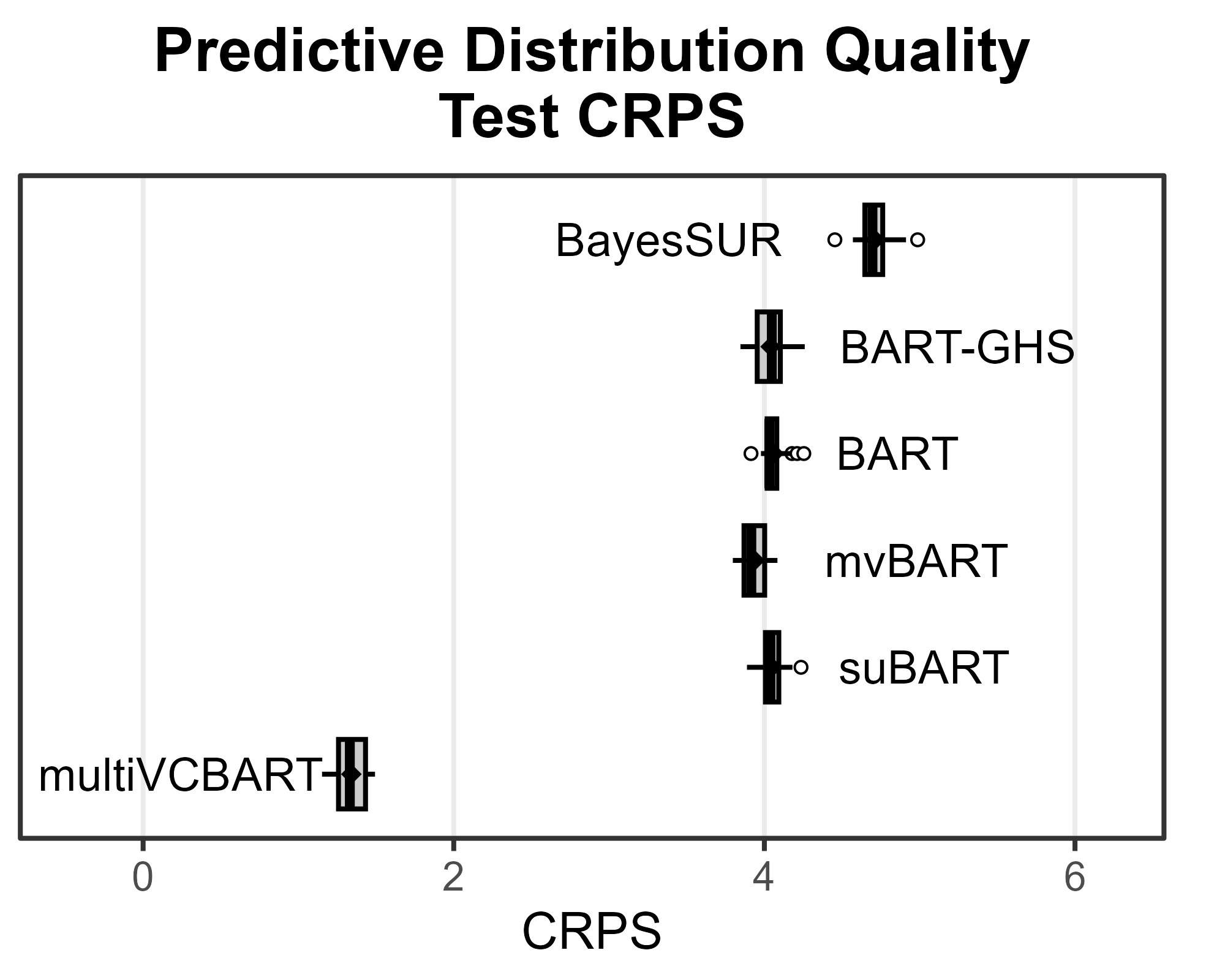}
\caption{CRPS}\label{fig:friedman_crps}
\end{subfigure}\hfill
\begin{subfigure}[b]{0.24\linewidth}\centering\includegraphics[width=\linewidth]{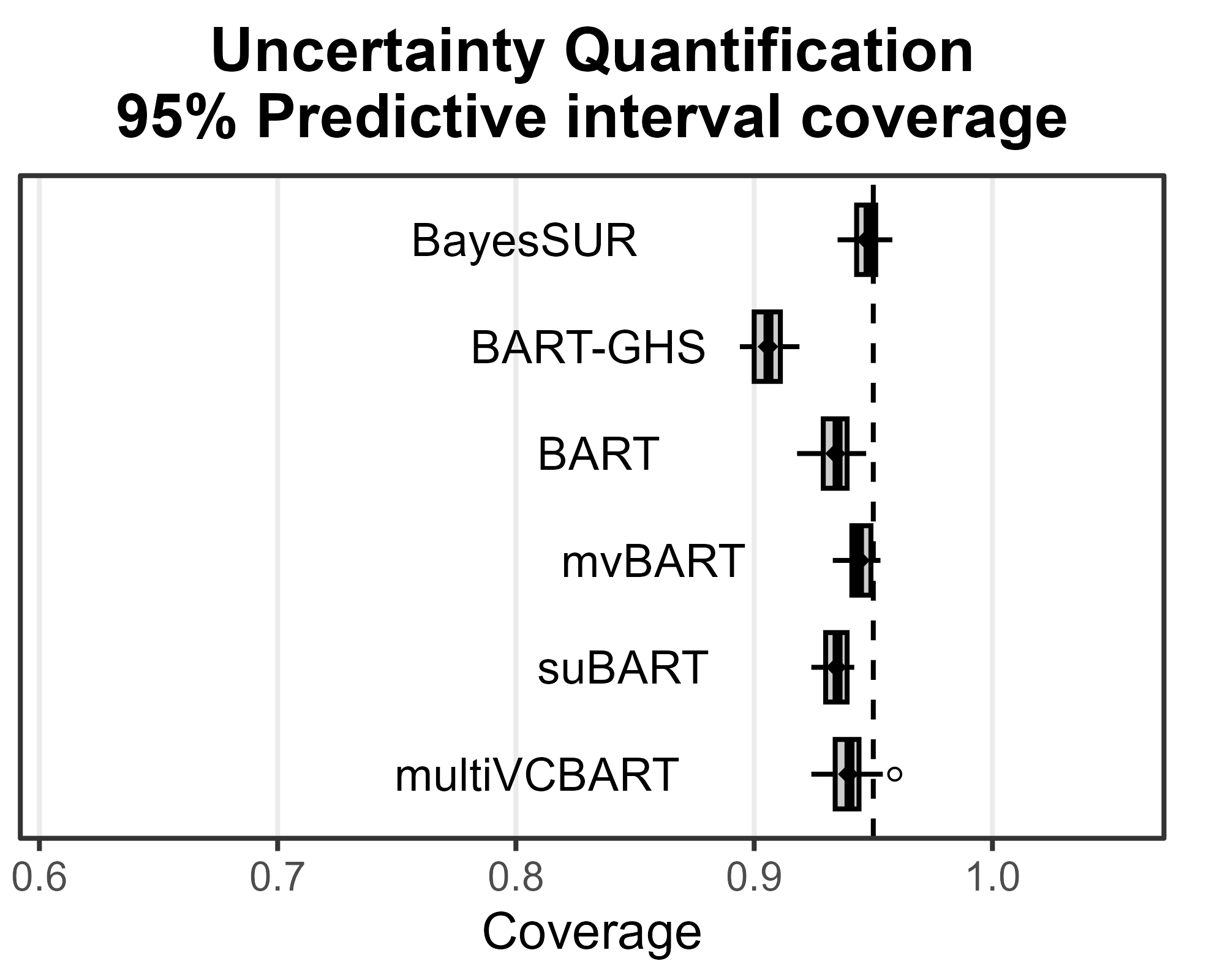}
\caption{Coverage}\label{fig:friedman_cov}
\end{subfigure}\hfill
\begin{subfigure}[b]{0.24\linewidth}\centering\includegraphics[width=\linewidth]{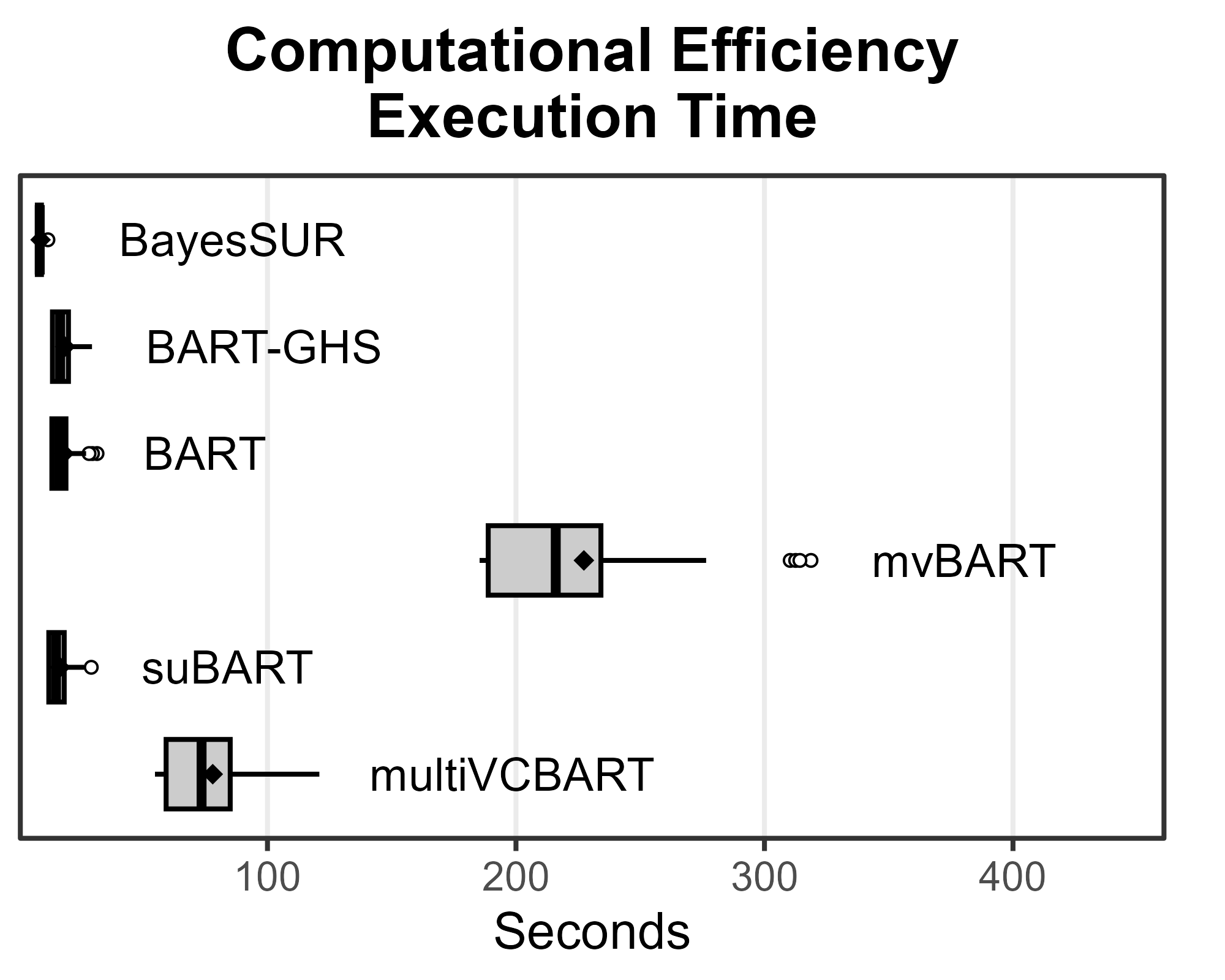}
\caption{Time (s)}\label{fig:friedman_time}\end{subfigure}\caption{High-dimensional Friedman--SUR benchmark ($p=100$, $d=50$, $q=2$). Lower is better for RMSE, CRPS, and Time. The dashed line in (c) marks nominal $95\%$ predictive-interval coverage.}\label{fig:friedman}
\end{figure}

\Cref{fig:friedman} reports the comparative performance in this simulation experiment.\texttt{multiVCBART} substantially improves predictive accuracy relative to all competitors. In \Cref{fig:friedman_rmse}, the distribution of test RMSE for \texttt{multiVCBART} is well separated from the other methods, with much smaller median error and relatively little variability. The remaining tree-based methods, including independent \texttt{BART}, \texttt{BART-GHS}, \texttt{mvBART}, and \texttt{suBART}, exhibit noticeably
larger RMSE, while \texttt{BayesSUR} performs worst. This is expected because \texttt{BayesSUR} can model residual dependence but imposes a linear mean structure, which is misspecified for the nonlinear varying-coefficient surfaces in this simulation design. The same qualitative conclusion appears in \Cref{fig:friedman_crps}. \texttt{multiVCBART} achieves by far the smallest CRPS, indicating that the entire posterior predictive distribution is better calibrated and sharper, not only that the posterior predictive mean is more accurate. Additionally, \Cref{fig:friedman_cov} shows that most methods attain coverage close to the nominal $95\%$ target, although no method is perfectly calibrated. 
Note that the two-step \texttt{BART-GHS} approach improves neither RMSE nor
CRPS relative to the strongest one-step tree baselines, suggesting that estimating the mean functions first and the residual graph second is not enough in this setting. 


Finally, \Cref{fig:friedman_time} suggests that
\texttt{BayesSUR} and independent \texttt{BART} fits are faster because they
fit simpler models. \texttt{multiVCBART} is computationally more expensive than these simpler baselines, reflecting the cost of fitting \(pq\) coefficient ensembles and jointly updating \(\Omega\). Recalling the
computational remark in \Cref{sec:comp} concerning the
per-tree proposal cost, our pseudo-response update avoids dense
\(q\times q\) Cholesky factorizations to evaluate leaf marginal likelihoods inside the tree-search loop. In the present Friedman
design, \(q=2\) is small while \(p=100\), so the dominant cost for
\texttt{multiVCBART} is cycling through the \(pq=200\) scalar coefficient-function ensembles rather than the residual-dependence update. At the same time, \texttt{multiVCBART} is considerably faster than \texttt{mvBART}, whose shared multivariate tree structure incurs heavier dense multivariate tree-update computations.

\section{GDSC drug sensitivity analysis}
\label{sec:gdsc_realdata}
We analyze a subset of the Genomics of Drug Sensitivity in Cancer (GDSC) dataset from a large-scale pharmacogenomic study in \citet{garnett2012a} and \citet{YangGDSC}, which appears in the \texttt{BayesSUR} tutorial \citep{ZhaoJSS}. The dataset consists of \(n=499\) cancer cell lines, \(q=7\) drug-response outcomes, \(d=13\) tissue-type indicators, and \(p=837\) omics covariates. The seven responses are standardized \(\log(\mathrm{IC}_{50})\) measurements for \texttt{Methotrexate, RDEA119, PD.0325901, CI.1040, AZD6244, Nilotinib}, and \texttt{Axitinib}. The predictor matrix is composed of \(343\) gene-expression features, \(426\) copy-number variation (CNV) features, and \(68\) mutation indicators. In our \texttt{multiVCBART} specification, the full omics panel enters as the primary covariates \(\bZ \in \mathbb{R}^{499 \times 837}\), while the \(13\) tissue indicators enter as modifiers \(\bX \in \{0,1\}^{499 \times 13}\), allowing each molecular effect to vary with tissue context.

We fit \texttt{multiVCBART} using all the \(837\) omics covariates and focus on the posterior inferential structure. Specifically, we ran four MCMC chains with \(M=50\) trees for \(2500\) iterations each, discarding the first \(1000\) iterations of each chain as burn-in. We focus on two complementary aspects of the fitted model: first, we examine the posterior median local scale parameters $\lambda_{jr}$, which control the shrinkage of the varying-coefficient functions \(B_{jr}(\bX)\). These are aggregated across the $q$ drug responses to yield a global biomarker importance score for the $j$-th covariate, defined by \(\bar{\lambda}_j = q^{-1}\sum_{r=1}^{q} \lambda_{jr}.\) Second, we inspect the residual drug--drug dependence captured by the estimated precision matrix \(\Omega\).

\paragraph{Biomarker ranking.}
Large values of \(\bar{\lambda}_j\) indicate molecular features whose
coefficient surfaces are repeatedly allowed to escape global shrinkage across the multivariate drug responses. The full posterior distribution of \(\bar{\lambda}_j\) for the top-ranked features is shown in
\suppref{\Cref{fig:gdsc_lambda}}{Fig. S3.1} of the supplement. \texttt{BCR\_ABL.MUT} exhibits the most prominent signal. Its posterior distribution for \(\bar{\lambda}_j\) is markedly shifted to the right relative to all other features, with the largest posterior median. This dominance is biologically well motivated since the \texttt{BCR}--\texttt{ABL1} fusion is the defining oncogenic lesion in chronic myeloid leukaemia and a canonical target of tyrosine-kinase inhibitors, with \texttt{Nilotinib} specifically developed as a potent \texttt{BCR-ABL} inhibitor \citep{Marin2023,Yeung02102016}.
Other highly ranked features include \texttt{LCP1}, \texttt{KRAS.MUT}, \texttt{PRDM16.CNV}, \texttt{MLLT}, \texttt{GATA2.CNV}, and \texttt{MLLT2.CNV}. Most of them have connections to oncogenic signalling \citep{garnett2012a} and genomic instability processes \citep{iorio2016}. In contrast, the lower-ranked blue features have smaller posterior medians. Their heavily overlapping boxplots indicate that these features are comparatively weak signals, and the posterior does not clearly distinguish their importance from one another.

\paragraph{Residual drug--drug dependence.}
\Cref{fig:gdsc_omega} shows the conditional dependency network induced by the estimated precision matrix \(\Omega\). We visualize the posterior mean partial-correlation graph $(\rho_{rs})$ after applying hard-thresholding at $0.05$, that is, setting $\hat{\rho}_{rs} = 0$ whenever $|\hat{\rho}_{rs}| < 0.05$. In the resulting network, positive edges are shown in green, negative edges in red, and thicker edges correspond to stronger magnitude. Importantly, \Cref{fig:gdsc_omega} summarizes dependence after adjusting for all \(837\) omics covariates and the tissue modifiers. Hence, the edges should be interpreted as residual conditional associations among the drug responses that remain unexplained by the observed biomarkers and tissue context.

Several scientifically meaningful structures emerge. First, the four drugs \texttt{RDEA119}, \texttt{PD.0325901}, \texttt{CI.1040}, and \texttt{AZD6244} form a clear positively connected module, with especially strong edges involving \texttt{CI.1040}. This is consistent with the fact that these compounds are all MEK-pathway inhibitors \citep{Cheng2017} and therefore would be expected to share residual sensitivity patterns across cell lines. 

Second, \texttt{Nilotinib} and \texttt{Methotrexate} are linked to the remainder of the network primarily through \texttt{Axitinib}, which shows positive residual associations both with the smaller \texttt{Nilotinib}--\texttt{Methotrexate} group and with the MEK-inhibitor block. Because this bridging role is inferred from the estimated partial-correlation graph, we do not interpret it as a previously established pharmacologic relationship. However, it is biologically plausible that \texttt{Axitinib} occupies a central position in the network, given its kinase inhibitory profile, which overlaps in part with that of \texttt{Nilotinib}, particularly through PDGFR and KIT-related activity \citep{Gunnarsson2017,Blay2011}. By contrast, the residual association with \texttt{Methotrexate}, an antifolate with a distinct mechanism of action, appears to be a more novel feature of the \texttt{multiVCBART} fit and warrants further investigation. The absence of strong negative edges in \Cref{fig:gdsc_omega} suggests that the dominant residual structure is one of shared co-sensitivity rather than residual antagonism. 

\begin{figure}[t]
    \centering
    \includegraphics[width=0.72\linewidth]{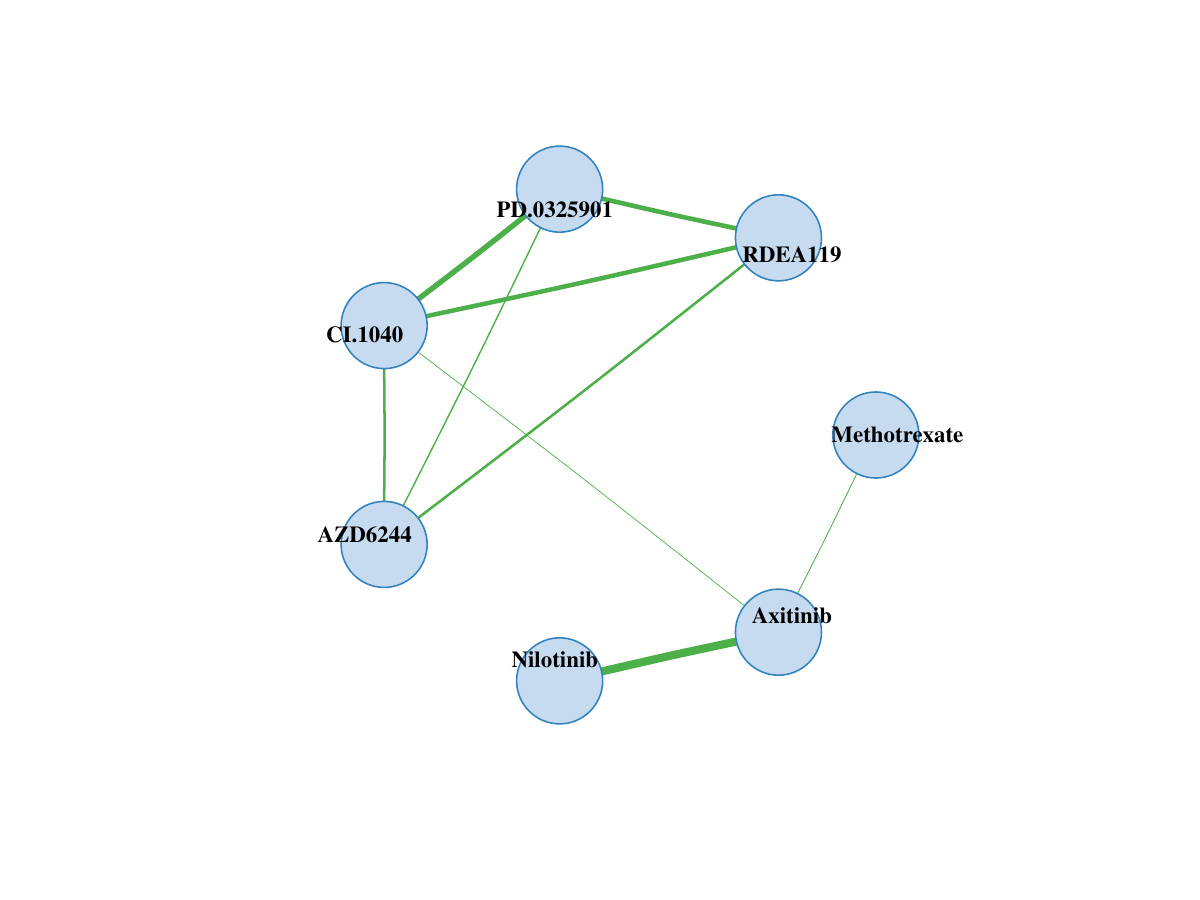}
    \caption{Estimated residual conditional dependency network among the seven GDSC drug responses.}
    \label{fig:gdsc_omega}
\end{figure}

\paragraph{Predictive assessment.} 
We carried out a 5-fold cross-validation analysis comparing \texttt{multiVCBART} with univariate \texttt{BART}, \texttt{suBART}, and \texttt{BayesSUR}. Across the five folds, \texttt{multiVCBART} achieved the lowest mean test RMSE of $\mathbf{0.874}$, indicating the best overall out-of-sample predictive performance among the four methods. The univariate \texttt{BART} fit performed very similarly, with mean test RMSE $0.88$, suggesting that flexible nonlinear regression alone already captures a substantial portion of the predictive signal in this dataset. However, the modest improvement of \texttt{multiVCBART} over univariate \texttt{BART} shows that modeling the seven drug responses jointly, while allowing covariate effects to vary with tissue context, yields an additional predictive gain. The multivariate \texttt{suBART} model performed slightly worse, with mean test RMSE $0.894$, while \texttt{BayesSUR} showed the weakest predictive accuracy, with RMSE $0.955$. Overall, the results are in favor of multivariate tree-based models compared to the more restrictive linear multivariate regression structure of \texttt{BayesSUR}.

\section{Discussion}
\label{sec:discussion}
We introduced \texttt{multiVCBART}, a multivariate varying-coefficient BART model that jointly estimates nonlinear high-dimensional mean structures and sparse residual dependence graphs. Theoretically, \Cref{thm:postconc_full} establishes optimal posterior contraction for $(\eta,\Omega)$ under a truncated prior restricted to the effective parameter sieve in high-dimensional regimes. Under the functional restricted eigenvalue condition (\Cref{assump:RE-func}), this truncated contraction rate optimally transfers to the underlying coefficient surfaces $\bm{B}(\cdot)$ and $G(\bx)$ (\Cref{cor:transformed-surface}).

Empirically, \texttt{multiVCBART} consistently outperformed competitors in predictive accuracy, interval coverage, and precision recovery across extensive simulations, particularly excelling when outcome structures are highly heterogeneous. Applied to GDSC pharmacogenomic data in \Cref{sec:gdsc_realdata}, it improved cross-validated predictions over \texttt{BayesSUR} while successfully uncovering biologically coherent conditional drug dependency networks.

Several extensions are natural. First, while we assumed independent priors across coefficient ensembles, one could instead couple the split-probability vectors across \((j,r)\) to encourage borrowing of strength when outcomes share important modifiers. Second, our current model assumes a covariate-invariant precision matrix \(\Omega\). Allowing \(\Omega\) to vary with modifiers would lead to a multivariate covariate-dependent graphical model. This would provide a novel route to learning outcome networks that evolve across subpopulations while retaining the flexible varying-coefficient mean structure.

There also remain open theoretical questions. Our contraction rate $\varepsilon^{\dagger}$ is consistent with known lower bounds for the constituent sub-problems, but a sharp joint minimax theory for nonlinear multivariate varying-coefficient regression with degree-aware sparse precision estimation is still unavailable. Likewise, stronger variable-selection guarantees for continuous shrinkage priors in multivariate tree models would further strengthen interpretability. We leave these investigations for future work.

\bibliography{refs}

\newpage
\begin{center}
{\large\textbf{Supplementary Materials}}
\end{center}

\appendix

\renewcommand{\thesection}{S\arabic{section}}
\renewcommand{\thefigure}{\thesection.\arabic{figure}}
\renewcommand{\thetable}{\thesection.\arabic{table}}
\renewcommand{\theequation}{\thesection.\arabic{equation}}
\renewcommand{\thelemma}{\thesection.\arabic{lemma}}
\renewcommand{\thetheorem}{\thesection.\arabic{theorem}}
\renewcommand{\theassumption}{\thesection.\arabic{assumption}}

This supplement provides additional technical and empirical details to support the main paper.
\Cref{sec:proofs} contains the complete mathematical proofs for our theoretical results.
\Cref{sec:multi-algorithm} gives additional details on posterior computation and implementation.
Finally, \Cref{sec:multi-experiments} presents additional experimental results and supporting analyses.

\setcounter{section}{0}
\setcounter{figure}{0}
\setcounter{equation}{0}
\setcounter{table}{0}
\setcounter{theorem}{0}
\setcounter{assumption}{0}

\section{Proofs}
\label{sec:proofs}
\subsection{Auxillary Lemmas}
\begin{lemma}\label{lem:HS-tailbound}
Let $M \ge 1$ be fixed. Assume \suppref{\ref{assum:multi-smoothness}}{(A1)} --  \suppref{\ref{assum:multi-eig}}{(A4)} and $b \mid (\lambda,\tau) \sim \mathcal{N} \left(0,M^{-1}(\tau^2 \lambda^2) \right),$ with $\lambda \sim \mathcal{C}^{+}(0,1).$ Then there exists a universal constant $C_{\textrm{HS}} > 0$ such that for all $u > 0,$
\[
\mathbb{P}(|b| > u \mid \tau) \le C_{\textrm{HS}} \frac{\tau}{u \sqrt{M}} \log \left(1+ \frac{u \sqrt{M}}{\tau} \right).
\]
\end{lemma}

\begin{proof} This lemma is an adaptation of the Gaussian-half-Cauchy mixture tail bound following \citet[Theorem 1]{Carvalho2010}.
Let us fix $\tau > 0$ and define the rescaled variable $X \coloneqq \sqrt{M} \cdot b/\tau$. Then conditional on $\lambda,$ $X \mid \lambda \sim \mathcal{N}(0,\lambda^2) \implies X = \lambda Z, \ Z \sim \mathcal{N}(0,1),$ with $Z$ independent of $\lambda$. Hence, for $s>0,$ 
\[
\mathbb{P}(|b| > u \mid \tau) = \mathbb{P} \left(|X| > u \cdot \sqrt{M}/\tau \right).
\]
Thus, it suffices to bound $\mathbb{P}(|\lambda Z| > s)$ where $s \coloneqq u \sqrt{M}/\tau.$
Using the standard Gaussian tail inequality for $t>0,$ $$\mathbb{P}(|Z| > t) = 2 \Phi(-t) \le t^{-1} 2 \phi(t) \le \sqrt{\frac{2}{\pi}} \cdot \frac{1}{t},$$
since $\phi(t) \le \phi(0) = 1/\sqrt{2 \pi}.$ Therefore, for any $\lambda > 0,$ 
\[
\mathbb{P}(|\lambda Z| > s \mid \lambda) = \mathbb{P} \left(|Z| > s/\lambda \right) \le \min \left\{1, \sqrt{\frac{2}{\pi}} \cdot \frac{\lambda}{s} \right \}.
\]
Now, denoting the $\mathcal{C}^{+}(0,1)$ density by $f(\lambda),$ we have 
\begin{align*}
\mathbb{P}(|\lambda Z| > s)  & = \int_{0}^{\infty} \mathbb{P}(|\lambda Z|>s \mid \lambda)f(\lambda) d \lambda \\
& \le \underbrace{\int_{0}^{s} \sqrt{\frac{2}{\pi}} \frac{\lambda}{s} \cdot f(\lambda) d \lambda}_{I_1} + \underbrace{\int_{s}^{\infty} 1\cdot f(\lambda) d \lambda}_{I_2}
\end{align*}
\textbf{Bounding $I_1$.} See that on computing $I_1$ we have
\begin{align*}
    I_1 \coloneqq \int_{0}^{s} \sqrt{\frac{2}{\pi}} \frac{\lambda}{s} \cdot f(\lambda) d \lambda & \le \sqrt{\frac{2}{\pi}} \frac{1}{s} \int_{0}^{s} \lambda \cdot \frac{2}{\pi(1+\lambda^2)} d \lambda \\
    & = \sqrt{\frac{2}{\pi}} \frac{2}{\pi s} \cdot \frac{1}{2} \log(1+s^2) \\
    & = \sqrt{\frac{2}{\pi}} \frac{1}{\pi s} \log (1+s^2) \\
    & \le \frac{2}{\pi} \sqrt{\frac{2}{\pi}} \cdot \frac{1}{s} \log(1+s),
\end{align*}
where the last inequality follows from $\log(1+s^2) \le 2 \log(1+s),$ for all $s \ge 0.$

\textbf{Bounding $I_2$.} Since $\int_{s}^{\infty} f(\lambda) d \lambda = (2/\pi) \cdot \arctan(1/s),$ we have 
\[
I_2 = \frac{2}{\pi} \arctan(1/s) \le \min \left\{1, \frac{2}{\pi  s} \right\}. 
\]

It is worthwhile to note that the function $\log (1+s)/s$ is decreasing on $(0,\infty).$ Thus, we can show that 
\[
\frac{\log(1+s)}{s} \cdot \frac{1}{\log 2} \ge 1 \ge I_2.
\]

Hence, combining the above inequalities, we have
\[
\mathbb{P}(|\lambda Z | > s) \le \frac{2}{\pi} \sqrt{\frac{2}{\pi}} \frac{\log(1+s)}{s} + \frac{1}{\log 2} \frac{\log(1+s)}{s}. 
\]
Taking $C_{\textrm{HS}} \coloneqq (2/\pi)^{3/2} + (\log 2)^{-1}$ concludes the proof.
\end{proof}


\begin{lemma}
\label{lem:eff-support-correct}
Let $Q_B:=pq$ be the number of coefficient-functions $(j,r)\in[p]\times[q]$.
Assume the horseshoe hierarchy in \suppref{\ref{assum:multi-HSscale}}{(P1)}, and assume the global scale is hard-truncated at its tuning level
\(
\tau_B\ \sim\ \mathcal C^+(0,\sigma_{B,n})\,\mathbbm 1_{[0,\,C_\tau\sigma_{B,n}]},
\ \text{for a fixed }C_\tau\ge 1.
\)
Assume also that each group $(j,r)$ contains at most $L_\star$ leaves in total.
For a deterministic threshold $u>0$, define the effective support
\[
\widehat S_B(u) := \Big\{(j,r)\in[p]\times[q]:\ \max_{t,\ell}|b_{jrt\ell}|>u\Big\}.
\]
Then there exist absolute constants $c_\star>0$ and $\kappa_B\in(0,1)$ such that, for any target size
$1\le S^\star\le Q_B$, if we set
$u_n := \left(c_\star\,L_\star{\sigma_{B,n}}/{\sqrt M}\right) \cdot
\Big({eQ_B}/{S^\star}\Big)^{1+\kappa_B} \cdot
\log \Big({eQ_B}/{S^\star}\Big),$
then for all sufficiently large $n$,
\[
\Pi\big(|\widehat S_B(u_n)|\ge S^\star\big)
\le
\exp\Big\{-\kappa_B\,S^\star\log \Big(\frac{eQ_B}{S^\star}\Big)\Big\}.
\]
\end{lemma}

\begin{proof}
We fix $\tau\in(0,C_\tau\sigma_{B,n}]$ and write $\mathbb P_\tau(\cdot)$ for probability conditional on $\tau_B=\tau$. Let us consider a single leaf \( b\mid(\tau,\lambda) \sim \mathcal N \Big(0,{\tau^2\lambda^2}/{M}\Big), \ \lambda\sim\mathcal C^+(0,1) \mathbbm 1_{[0,n^{A_B}]}. \) Note that $g(\lambda) \coloneqq \P(|b| > u \mid \tau,\lambda) = 2(1-\Phi(u \sqrt{M}/(\tau \lambda)))$ is increasing in $\lambda$. 
Because the truncated law is exactly the untruncated law conditioned on $\{\lambda \le n^{A_B}\},$
\[
\P_{\textrm{trunc}} (|b| > u \mid \tau) = \E[g(\lambda) \mid \lambda \le n^{A_B} ] \le \E[g(\lambda)] = \P_{\textrm{untrunc}} (|b| > u \mid \tau).
\]

Hence, truncating $\lambda_{jr}\le n^{A_B}$ only decreases tails, so we may use the untruncated Gaussian--half-Cauchy
mixture tail bound following \Cref{lem:HS-tailbound}. There exists a universal constant $C_{\textrm HS}>0$ such that for all $u>0$,
\begin{equation}\label{eq:hs-tail-cond-global}
\mathbb P_\tau(|b|>u)\ \le\ C_{\textrm HS} \frac{\tau}{u\sqrt M}
\log \Big(1+\frac{u\sqrt M}{\tau}\Big).
\end{equation}
For a single group $(j,r)$, we define
\(
A_{jr}(u) := \Big\{\max_{t,\ell}|b_{jrt\ell}|>u\Big\}.
\)
By a union bound over at most $L_\star$ leaves in the group and \eqref{eq:hs-tail-cond-global},
\begin{equation}\label{eq:pgroup-tau}
p_u(\tau)
:=\mathbb P_\tau(A_{jr}(u))
 \le
L_\star C_{\textrm HS} \frac{\tau}{u\sqrt M} 
\log \Big(1+\frac{u\sqrt M}{\tau}\Big).
\end{equation}
Conditional on $\tau_B=\tau$, the collections $\{(\lambda_{jr},\{b_{jrt\ell}\}_{t,\ell})\}$ are i.i.d.\ across
$(j,r)$, hence the indicators $\{\mathbbm 1(A_{jr}(u))\}_{(j,r)}$ are i.i.d. Bernoulli$(p_u(\tau))$ and
\(
|\widehat S_B(u)|\ \big|\ \tau_B=\tau\ \sim\ {\textrm Bin}\big(Q_B, p_u(\tau)\big).
\)

Setting $u=u_n$, we define $x:=u_n\sqrt M/\tau$. Then
\(
{\tau}/{(u_n\sqrt M)}\log \Big(1+{u_n\sqrt M}/{\tau}\Big)
={\log(1+x)}/{x}=:g(x),
\)
where $g$ is decreasing on $(0,\infty)$ (indeed $g'(x)=\{{x}/{(1+x)}-\log(1+x)\}/x^2\le 0$).

Since $\tau\le C_\tau\sigma_{B,n}$,
\[
x=\frac{u_n\sqrt M}{\tau}
\ \ge\
\frac{u_n\sqrt M}{C_\tau\sigma_{B,n}}
=
\frac{c_\star L_\star}{C_\tau}\Big(\frac{eQ_B}{S^\star}\Big)^{1+\kappa_B}
\log \Big(\frac{eQ_B}{S^\star}\Big)
=:x_{\min}.
\]
By monotonicity of $g$, $g(x)\le g(x_{\min})=\log(1+x_{\min})/x_{\min}$.

Let $y:=eQ_B/S^\star\ge e$ and $L:=\log y\ge 1$. Then
\(
x_{\min}=({c_\star L_\star}/{C_\tau}) y^{1+\kappa_B}L.
\)
Choosing $c_\star$ large enough ensures $x_{\min}\ge 1$ for all $y\ge e$.
For $x_{\min}\ge 1$, $\log(1+x_{\min})\le \log(2x_{\min})$, hence
\[
\log(1+x_{\min})
\le
\log\!\Big(2\frac{c_\star L_\star}{C_\tau}\Big)+(1+\kappa_B)\log y+\log L.
\]
Since $L\ge 1$, we have $\log L\le L$; also $\log(2c_\star L_\star/C_\tau)\le C\,L$ because $L\ge 1$
and the left-hand side is a fixed constant. Therefore, for some universal $C_0<\infty$,
\(
\log(1+x_{\min}) \le  C_0 L.
\)
Dividing by $x_{\min}=(c_\star L_\star/C_\tau)\,y^{1+\kappa_B}L$ gives
\[
g(x_{\min})=\frac{\log(1+x_{\min})}{x_{\min}}
\le
\frac{C_0}{c_\star} \frac{C_\tau}{L_\star} y^{-(1+\kappa_B)}
=
\frac{C_0C_\tau}{c_\star L_\star}\Big(\frac{S^\star}{eQ_B}\Big)^{1+\kappa_B}.
\]
Plugging this bound into \eqref{eq:pgroup-tau} yields, uniformly for all $\tau\in(0,C_\tau\sigma_{B,n}]$,
\[
p_{u_n}(\tau)
\le
L_\star C_{\textrm HS} g(x_{\min})
\le
\frac{C_{\textrm HS}C_0C_\tau}{c_\star}\Big(\frac{S^\star}{eQ_B}\Big)^{1+\kappa_B}.
\]
Choose $c_\star$ sufficiently large so that $(C_{\textrm HS}C_0C_\tau)/c_\star\le 1$. Then
\begin{equation}\label{eq:pgroup-final}
p_{u_n}(\tau) \le \Big(\frac{S^\star}{eQ_B}\Big)^{1+\kappa_B},
\quad \forall \tau\in(0,C_\tau\sigma_{B,n}].
\end{equation}

Conditional on $\tau_B=\tau$, $|\widehat S_B(u_n)|\sim{\textrm Bin}(Q_B,p_{u_n}(\tau))$, hence
\[
\mathbb P_\tau\big(|\widehat S_B(u_n)|\ge S^\star\big)
\le \sum_{s=S^\star}^{Q_B}\binom{Q_B}{s}p_{u_n}(\tau)^s
\le \sum_{s=S^\star}^{Q_B}\Big(\frac{eQ_B}{s} p_{u_n}(\tau)\Big)^s,
\]
where we used $\binom{Q_B}{s}\le (eQ_B/s)^s$.
Using \eqref{eq:pgroup-final} and $s\ge S^\star$,
\[
\Big(\frac{eQ_B}{s}\,p_{u_n}(\tau)\Big)^s
\le
\Big(\frac{eQ_B}{s}\Big)^s\Big(\frac{S^\star}{eQ_B}\Big)^{(1+\kappa_B)s}
=
\Big(\frac{S^\star}{s}\Big)^s\Big(\frac{S^\star}{eQ_B}\Big)^{\kappa_B s}
\le
\Big(\frac{S^\star}{eQ_B}\Big)^{\kappa_B s}.
\]
Therefore,
\[
\mathbb P_\tau\big(|\widehat S_B(u_n)|\ge S^\star\big)
\le
\sum_{s=S^\star}^\infty \Big(\frac{S^\star}{eQ_B}\Big)^{\kappa_B s}
=
\frac{\Big(\frac{S^\star}{eQ_B}\Big)^{\kappa_B S^\star}}{1-\Big(\frac{S^\star}{eQ_B}\Big)^{\kappa_B}}.
\]
Since $S^\star\le Q_B$ implies $S^\star/(eQ_B)\le 1/e$, we have
$\big(\frac{S^\star}{eQ_B}\big)^{\kappa_B}\le e^{-\kappa_B}$ and thus
\[
\mathbb P_\tau\big(|\widehat S_B(u_n)|\ge S^\star\big)
\le
\frac{1}{1-e^{-\kappa_B}}\,
\exp\Big\{-\kappa_B S^\star\log \Big(\frac{eQ_B}{S^\star}\Big)\Big\}.
\]
The factor $(1-e^{-\kappa_B})^{-1}$ is an absolute constant.
In our use of this lemma, $S^\star=S^\star_n\to\infty$ so the exponent diverges; hence for all sufficiently large $n$
we can absorb this constant into the exponential by shrinking $\kappa_B$,
yielding the stated bound.

The above bound holds uniformly for all $\tau\in(0,C_\tau\sigma_{B,n}]$, which is the support of $\tau_B$.
Integrating over $\tau_B$ gives the same bound for $\Pi(|\widehat S_B(u_n)|\ge S^\star)$.

\end{proof}


\begin{lemma}
\label{lem:GHS-effective-support-correct}
Let $Q=\binom{q}{2}$ and consider the graphical horseshoe hierarchy on the
off-diagonals as in (P2), but \emph{before} imposing the SPD restriction. Assume the global scale is truncated at its tuning level:
\(\tau_\Omega \sim \mathcal C^+(0,\sigma_{\Omega,n})\,\mathbbm 1_{[0,\,C_\tau\sigma_{\Omega,n}]}\) for a fixed constant \(C_\tau\ge 1.\) 
For a deterministic threshold $t>0$, define the effective edge set
\[
\widehat S_\Omega(t)\ :=\ \big\{(k,k'): k<k',\ |\omega_{kk'}|>t\big\}.
\]
Fix any target sparsity $1\le s\le Q$ and any specific edge set $S\subset\{(k,k'):k<k'\}$ with $|S|=s$.
Let $d_k(S)$ be the degrees of $S$ so that $\sum_{k=1}^q d_k(S)=2s$.
Then there exist absolute constants $c_\star>0$ and $\kappa_\Omega\in(0,1)$ such that, if we set
\(
t_n  := c_\star({\sigma_{\Omega,n}Q}/{s})\log q,
\)
then
\[
\Pi\big(S \subseteq \widehat S_\Omega(t_n)\big)
 \le
\exp\Big\{-\kappa_\Omega\,s\log \left(\frac{e Q}{s}\right)\Big\}
 \le
\exp\Big\{-\frac{\kappa_\Omega}{4}\sum_{k=1}^q d_k(S)\log\left(\frac{e q}{d_k(S)}\right)\Big\}.
\]
\end{lemma}

\begin{proof}
Throughout, $\Pi(\cdot)$ denotes the product prior on $(\omega_{kk'},\lambda_{kk'},\tau_\Omega)$
(with truncations), i.e. without restricting to $\{\Omega\succ0\}$.
Fix $\tau\in(0,C_\tau\sigma_{\Omega,n}]$ and consider a single off-diagonal: \(\omega\mid(\tau,\lambda)\sim \mathcal N(0,\tau^2\lambda^2)\), with \(\lambda\sim\mathcal C^+(0,1)\mathbbm 1_{[0,n^{A_\Omega}]}.\)
Let $p_t(\tau):=\P(|\omega|>t_n\mid \tau_\Omega=\tau)$.

Since $\tau\le C_\tau\sigma_{\Omega,n}$ and $t_n=c_\star(\sigma_{\Omega,n}Q/s)\log q$,
\[
\frac{\tau}{t_n} \le \frac{C_\tau\sigma_{\Omega,n}}{c_\star\sigma_{\Omega,n}}\cdot
\frac{s}{Q\log q} = \frac{C_\tau}{c_\star}\cdot\frac{s}{Q\log q}.
\]
Thus, we have
\begin{align*}
\log \Big(1+\frac{t_n}{\tau}\Big)
& \le 
\log\Big(2\frac{t_n}{\tau}\Big) \
\\
& \le \log \left(2 \frac{c_{\star}}{C_{\tau}} \frac{Q}{s} \log q \right) \\
& = \log \left(2 \frac{c_{\star}}{C_{\tau}} \right) + \log \left(\frac{Q}{s} \right) + \log \log q \le C_0 \log q.
\end{align*}
for a universal constant $C_0 = \log (2 c_{\star}/C_{\tau}) + 3>0$, using $Q\le q^2$ and absorbing $\log(c_\star/C_\tau)$ into constants, and assuming $q \ge 3$.
Plugging into \Cref{lem:HS-tailbound} yields the uniform bound
\begin{equation}\label{eq:pt-unif}
p_t(\tau) \le C_{\textrm HS} \frac{C_\tau}{c_\star}\cdot\frac{s}{Q\log q}\cdot C_0\log q
\ =\ \frac{C_1}{c_\star}\cdot\frac{s}{Q},
\quad \forall \tau\in(0,C_\tau\sigma_{\Omega,n}],
\end{equation}
where $C_1:=C_{\textrm HS}C_\tau C_0$.
Choose $c_\star\ge 4eC_1$. Then \eqref{eq:pt-unif} implies
\begin{equation}\label{eq:pt-small}
p_t(\tau) \le \frac{s}{4eQ}
 \le\ \Big(\frac{s}{eQ}\Big)^{\kappa_\Omega}
\ \forall \tau\in(0,C_\tau\sigma_{\Omega,n}],
\end{equation}
for some absolute $\kappa_\Omega\in(0,1)$ (e.g. any $\kappa_\Omega\le 1$ works since $(s/(eQ))^{\kappa_\Omega}\ge s/(eQ)$).

Conditional on $\tau_\Omega=\tau$, the off-diagonals are independent across $(k,k')$.
Therefore, for a fixed $S$ with $|S|=s$,
\[
\mathbb{P}\big(S\subseteq \widehat S_\Omega(t_n)\ \big|\ \tau_\Omega=\tau\big)
=\prod_{(k,k')\in S}\mathbb{P}\big(|\omega_{kk'}|>t_n\ \big|\ \tau_\Omega=\tau\big)
=p_t(\tau)^s.
\]
Using the uniform bound \eqref{eq:pt-small} and integrating over $\tau_\Omega$,
\begin{align*}
\Pi\big(S\subseteq \widehat S_\Omega(t_n)\big)
=\mathbb E_{\tau_\Omega}\big[p_t(\tau_\Omega)^s\big]
 \le \Big(\sup_{\tau\in(0,C_\tau\sigma_{\Omega,n}]} p_t(\tau)\Big)^s
 & \le \Big(\frac{s}{eQ}\Big)^{\kappa_\Omega s} \\
 & =\exp\Big\{-\kappa_\Omega\,s\log\Big(\frac{eQ}{s}\Big)\Big\},
\end{align*}
which proves the first inequality.

Let $d_k(S)$ be the degrees of $S$ so that $\sum_k d_k(S)=2s$.
The function $\phi(x):=x\log(eq/x)$ is concave on $(0,\infty)$, hence Jensen's inequality yields
\[
\sum_{k=1}^q d_k(S)\log \Big(\frac{eq}{d_k(S)}\Big)
 \le q\Big(\frac{2s}{q}\Big)\log \Big(\frac{eq}{2s/q}\Big)
=2s \log \Big(\frac{eq^2}{2s}\Big).
\]
Since $Q=q(q-1)/2\asymp q^2/2$, there is a universal constant $C_2\ge 1$ such that
$\frac{eq^2}{2s}\le \frac{C_2eQ}{s}$ for all $q\ge 2$, so
\[
\sum_{k=1}^q d_k(S)\log \Big(\frac{eq}{d_k(S)}\Big)
 \le\ 2s \Big[\log \Big(\frac{eQ}{s}\Big)+\log C_2\Big].
\]
Rearranging gives
\[
s\log \Big(\frac{eQ}{s}\Big)
 \ge
\frac12\sum_{k=1}^q d_k(S)\log \Big(\frac{eq}{d_k(S)}\Big)\ -\ s\log C_2.
\]
Plugging into the first inequality and absorbing the additive $s\log C_2$ into the leading constant yields
\[
\Pi\big(S\subseteq \widehat S_\Omega(t_n)\big)
 \le
\exp\Big\{-\frac{\kappa_\Omega}{4}\sum_{k=1}^q d_k(S)\log \Big(\frac{eq}{d_k(S)}\Big)\Big\}.
\]
This proves the second inequality.

\textbf{Remark on the SPD restriction.}
\Cref{lem:ghs-smallball-continuous} in our proof handles the SPD restriction
in the small-ball direction. For the present upper tail bound, the renormalized
restriction to $\{\Omega\succ0\}$ can, in principle, inflate probabilities by a factor
$1/\Pi_{\textrm{prod}}(\Omega\succ0)$, so it is standard to apply this lemma under the underlying
product measure, and enforce $\Omega\succ0$ separately in the sieve definition.
\end{proof}
\begin{lemma}\label{lem:pd-correct}
If $\|\Delta_\Omega\|_F\le \underline\lambda/4$, then $\Omega=\Omega_0+\Delta_\Omega\succ0$ and
\[
\lambda_{\min}(\Omega)\ge \underline\lambda/2,\qquad
\|\Omega^{-1}\|_{\textrm{op}}\le 2/\underline\lambda,\qquad
\|\Omega^{-1}-\Omega_0^{-1}\|_F\le \frac{2}{\underline\lambda^2}\,\|\Delta_\Omega\|_F.
\]
\end{lemma}

\begin{proof}
By Weyl's inequality,
$\lambda_{\min}(\Omega)\ge \lambda_{\min}(\Omega_0)-\|\Delta_\Omega\|_{\textrm{op}}
\ge \underline\lambda-\|\Delta_\Omega\|_F\ge \underline\lambda/2$,
so $\Omega\succ0$ and $\|\Omega^{-1}\|_{\textrm{op}} = \lambda_{\max}(\Omega^{-1}) = 1/\lambda_{\min}(\Omega)\le 2/\underline\lambda$.
Finally,
$\Omega^{-1}-\Omega_0^{-1}=\Omega^{-1}(\Omega_0-\Omega)\Omega_0^{-1}=-\Omega^{-1}\Delta_\Omega\Omega_0^{-1}$,
hence
$\|\Omega^{-1}-\Omega_0^{-1}\|_F\le \|\Omega^{-1}\|_{\textrm{op}}\|\Delta_\Omega\|_F\|\Omega_0^{-1}\|_{\textrm{op}}
\le (2/\underline\lambda)(1/\underline\lambda)\|\Delta_\Omega\|_F$.
\end{proof}

\begin{lemma}\label{lem:taylor-correct}
Let $E:=\Omega_0^{-1/2}\Delta_\Omega\,\Omega_0^{-1/2}$.
If $\|E\|_{\textrm{op}}\le\tfrac12$, then for $\phi(x) = x-1-\log x$:
\[
\mathrm{tr}(\Omega\Omega_0^{-1})-q-\log\det(\Omega\Omega_0^{-1})
=\sum_{j=1}^q \phi(1+\lambda_j(E))
\le 2\,\|E\|_F^2
\le \frac{2}{\underline\lambda^{2}}\|\Delta_\Omega\|_F^2,
\]
where $\{\lambda_j(E)\}$ are the eigenvalues of $E$.
\end{lemma}

\begin{proof}
For $|u|\le 1/2$, and $\phi(u)=u-1-\log u,$ Taylor's theorem implies $\phi(1+u)\le 2u^2$.
Summing gives $\sum_j\phi(1+\lambda_j(E))\le 2\sum_j\lambda_j(E)^2=2\|E\|_F^2$.
Finally, $\|E\|_F\le \|\Omega_0^{-1}\|_{\textrm{op}}\|\Delta_\Omega\|_F\le \underline\lambda^{-1}\|\Delta_\Omega\|_F$.
\end{proof}
\begin{lemma}\label{lem:block-metric-correct}
Let $\Theta_\ell=(\eta_\ell,\Omega_\ell)$ for $\ell=1,2$, and write $\Sigma_\ell:=\Omega_\ell^{-1}$.
For each $i$, let $f_{\Theta_\ell,i}:=\mathcal N_q(\eta_{\ell,i},\Sigma_\ell)$ and define
\[
H^2(\Theta_1,\Theta_2):=\frac1n\sum_{i=1}^n h^2\!\bigl(f_{\Theta_1,i},f_{\Theta_2,i}\bigr),
\]
where $h(\cdot,\cdot)$ is Hellinger distance. Let $\widetilde{\mathcal F}_n$ be any set on which
\[
\|\Omega\|_{\textrm{op}}\le R_{\Omega,n},\quad \|\Omega^{-1}\|_{\textrm{op}}\le \overline R_{\Omega,n},
\quad
\max_{1\le i\le n}\|\eta_i\|_2\le R_{\eta,n},
\]
with $(R_{\Omega,n},\overline R_{\Omega,n},R_{\eta,n})$ growing at most polynomially in $n$.
Then there exist constants $c_\star=C_\star(n)>0$ and $C^\star=C^\star(n)>0$
such that for all $\Theta_1,\Theta_2\in\widetilde{\mathcal F}_n$,
\[
c_\star\Big\{\|\Omega_1-\Omega_2\|_F^2+\|\eta_1-\eta_2\|_{2,n}^2\Big\}
 \le
H^2(\Theta_1,\Theta_2)
\le
C^\star\Big\{\|\Omega_1-\Omega_2\|_F^2+\|\eta_1-\eta_2\|_{2,n}^2\Big\},
\]
and moreover $C^\star\le n^{C}$ and $c_\star\ge n^{-C}$ for some fixed $C>0$.
\end{lemma}

\begin{proof}
We fix $\Theta_1=(\eta_1,\Omega_1)$ and $\Theta_2=(\eta_2,\Omega_2)$ in $\widetilde{\mathcal F}_n$, and 
write $\delta_i:=\eta_{1,i}-\eta_{2,i}$ and $\bar\Sigma:=\tfrac12(\Sigma_1+\Sigma_2)$.

For Gaussian measures, the squared Hellinger distance admits the factorization \citep[Chapter 1]{Pardo2005StatisticalIB}:
\begin{equation}\label{eq:hell-fact}
h^2\!\big(\mathcal N_q(\eta_{1,i},\Sigma_1),\mathcal N_q(\eta_{2,i},\Sigma_2)\big)
=
1-
\frac{|\Sigma_1|^{1/4}|\Sigma_2|^{1/4}}{|\bar\Sigma|^{1/2}}
\exp\!\Big(-\tfrac18\,\delta_i^\top \bar\Sigma^{-1}\delta_i\Big).
\end{equation}
Define
\[
T_{\det}:=1-\frac{|\Sigma_1|^{1/4}|\Sigma_2|^{1/4}}{|\bar\Sigma|^{1/2}},
\quad
T_{\mu,i}:=1-\exp\!\Big(-\tfrac18\,\delta_i^\top \bar\Sigma^{-1}\delta_i\Big).
\]
Since $1-ab\le (1-a)+(1-b)$ for $a,b\in[0,1]$, \eqref{eq:hell-fact} implies
\begin{equation}\label{eq:hell-split}
h^2 \big(\mathcal N_q(\eta_{1,i},\Sigma_1),\mathcal N_q(\eta_{2,i},\Sigma_2)\big)
\le T_{\det}+T_{\mu,i}.
\end{equation}
Also, since $\exp(-x)\in(0,1]$, we have $1-a b \ge 1-b$ for $a\in(0,1]$, hence
\begin{equation}\label{eq:hell-lower-mean}
h^2 \big(\mathcal N_q(\eta_{1,i},\Sigma_1),\mathcal N_q(\eta_{2,i},\Sigma_2)\big)
\ge T_{\mu,i}.
\end{equation}
Similarly, $1-ab\ge 1-a$ gives $h^2\ge T_{\det}$.

On $\widetilde{\mathcal F}_n$, the eigenvalues of each $\Sigma_\ell$ lie in
\(
\Big[{R_{\Omega,n}}^{-1},\ \overline R_{\Omega,n}\,\Big],
\)
hence the same holds for $\bar\Sigma$. Therefore,
\begin{equation}\label{eq:barSigma-bds}
\|\bar\Sigma^{-1}\|_{\textrm{op}}\le R_{\Omega,n},
\quad
\lambda_{\min}(\bar\Sigma^{-1})\ge \frac{1}{\overline R_{\Omega,n}}.
\end{equation}

\paragraph{Upper bound.}
We bound the two terms in \eqref{eq:hell-split}.

\emph{(i) Mean term.}
Using $1-e^{-x}\le x$ and \eqref{eq:barSigma-bds},
\[
T_{\mu,i}\le \frac18\,\delta_i^\top \bar\Sigma^{-1}\delta_i
\le \frac18\,\|\bar\Sigma^{-1}\|_{\textrm{op}}\,\|\delta_i\|_2^2
\le \frac18\,R_{\Omega,n}\,\|\delta_i\|_2^2.
\]
Averaging over $i$ yields
\begin{equation}\label{eq:mean-upper}
\frac1n\sum_{i=1}^n T_{\mu,i}
\le \frac18\,R_{\Omega,n}\,\|\eta_1-\eta_2\|_{2,n}^2.
\end{equation}

\emph{(ii) Determinant term.}
Write
\[
A_{\det}
:=\frac{|\Sigma_1|^{1/4}|\Sigma_2|^{1/4}}{|\bar\Sigma|^{1/2}}
=\exp\Big\{\tfrac14\log|\Sigma_1|+\tfrac14\log|\Sigma_2|-\tfrac12\log|\bar\Sigma|\Big\}.
\]
The map $F(\Sigma)=\log\det(\Sigma)$ has Hessian
$\nabla^2 F(\Sigma)[H,H]=-\mathrm{tr}(\Sigma^{-1}H\Sigma^{-1}H)$, hence
\[
\big|\nabla^2F(\Sigma)[H,H]\big|
\le \|\Sigma^{-1}\|_{\textrm{op}}^2\,\|H\|_F^2.
\]
On $\widetilde{\mathcal F}_n$, $\|\Sigma^{-1}\|_{\textrm{op}}=\|\Omega\|_{\textrm{op}}\le R_{\Omega,n}$,
so for all $\Sigma$ on the line segment between $\Sigma_1$ and $\Sigma_2$,
$\|\Sigma^{-1}\|_{\textrm{op}}\le R_{\Omega,n}$.
A second-order Taylor bound for the concave function $\log\det(\cdot)$ gives
\[
\log|\bar\Sigma|
\ge \frac12\log|\Sigma_1|+\frac12\log|\Sigma_2|
-\frac{R_{\Omega,n}^2}{8}\,\|\Sigma_1-\Sigma_2\|_F^2.
\]
Rearranging yields
\[
\frac14\log|\Sigma_1|+\frac14\log|\Sigma_2|-\frac12\log|\bar\Sigma|
\le \frac{R_{\Omega,n}^2}{16}\,\|\Sigma_1-\Sigma_2\|_F^2,
\]
hence $A_{\det}\ge \exp\big((-{R_{\Omega,n}^2}/{16})\|\Sigma_1-\Sigma_2\|_F^2\big)$ and therefore
\[
T_{\det}=1-A_{\det}\le 1-\exp \Big(-\frac{R_{\Omega,n}^2}{16}\|\Sigma_1-\Sigma_2\|_F^2\Big)
\le \frac{R_{\Omega,n}^2}{16}\,\|\Sigma_1-\Sigma_2\|_F^2.
\]
Finally, using the resolvent identity $\Sigma_1-\Sigma_2=\Omega_1^{-1}(\Omega_2-\Omega_1)\Omega_2^{-1}$ gives
\[
\|\Sigma_1-\Sigma_2\|_F
\le \|\Omega_1^{-1}\|_{\textrm{op}}\,\|\Omega_2^{-1}\|_{\textrm{op}}\,\|\Omega_1-\Omega_2\|_F
\le \overline R_{\Omega,n}^2\,\|\Omega_1-\Omega_2\|_F.
\]
Hence
\begin{equation}\label{eq:det-upper}
T_{\det}\le \frac{R_{\Omega,n}^2\,\overline R_{\Omega,n}^4}{16}\,\|\Omega_1-\Omega_2\|_F^2.
\end{equation}

Combining \eqref{eq:hell-split}, \eqref{eq:mean-upper}, and \eqref{eq:det-upper} and averaging over $i$ yields
\[
H^2(\Theta_1,\Theta_2)
\le C^\star\Big\{\|\Omega_1-\Omega_2\|_F^2+\|\eta_1-\eta_2\|_{2,n}^2\Big\},
\]
with
\(
C^\star:=\max\Big\{\frac{R_{\Omega,n}^2\overline R_{\Omega,n}^4}{16},\ \frac{R_{\Omega,n}}{8}\Big\},
\)
which grows at most polynomially in $n$.

\paragraph{Lower bound.}
We lower bound the mean and covariance contributions and then combine.

\medskip
\noindent\emph{(i) Mean term.}
From \eqref{eq:hell-lower-mean} and the inequality $1-e^{-x}\ge x/(1+x)$ for all $x\ge 0$,
\[
T_{\mu,i}
=1-\exp \Big(-\tfrac18\,\delta_i^\top\bar\Sigma^{-1}\delta_i\Big)
\ge
\frac{\tfrac18\,\delta_i^\top\bar\Sigma^{-1}\delta_i}{1+\tfrac18\,\delta_i^\top\bar\Sigma^{-1}\delta_i}.
\]
On $\widetilde{\mathcal F}_n$, $\|\delta_i\|_2\le \|\eta_{1,i}\|_2+\|\eta_{2,i}\|_2\le 2R_{\eta,n}$ and
$\|\bar\Sigma^{-1}\|_{\textrm{op}}\le R_{\Omega,n}$ by \eqref{eq:barSigma-bds}, hence
$\delta_i^\top\bar\Sigma^{-1}\delta_i\le R_{\Omega,n}\|\delta_i\|_2^2\le 4R_{\Omega,n}R_{\eta,n}^2$.
Therefore for all $i$,
\(
T_{\mu,i}
\ge c_{\mu,n}\,\delta_i^\top\bar\Sigma^{-1}\delta_i,
\ 
c_{\mu,n}:=\frac{1}{8\big(1+\tfrac12 R_{\Omega,n}R_{\eta,n}^2\big)}.
\)
Using $\delta_i^\top\bar\Sigma^{-1}\delta_i\ge \lambda_{\min}(\bar\Sigma^{-1})\|\delta_i\|_2^2
\ge \overline R_{\Omega,n}^{-1}\|\delta_i\|_2^2$ (again from \eqref{eq:barSigma-bds}),
\(
T_{\mu,i}\ge \frac{c_{\mu,n}}{\overline R_{\Omega,n}}\|\delta_i\|_2^2.
\)
Averaging over $i$ gives
\begin{equation}\label{eq:mean-lower}
H^2(\Theta_1,\Theta_2)\ \ge\ \frac1n\sum_{i=1}^n T_{\mu,i}
\ \ge\ \frac{c_{\mu,n}}{\overline R_{\Omega,n}}\,\|\eta_1-\eta_2\|_{2,n}^2.
\end{equation}

\medskip
\noindent\emph{(ii) Determinant term.}
Recall $h^2\ge T_{\det}:=1-A_{\det}$, where
\[
A_{\det}
=\frac{|\Sigma_1|^{1/4}|\Sigma_2|^{1/4}}{|\bar\Sigma|^{1/2}}
=\exp\Big\{\tfrac14\log|\Sigma_1|+\tfrac14\log|\Sigma_2|-\tfrac12\log|\bar\Sigma|\Big\}.
\]
Let $F(\Sigma)=\log\det(\Sigma)$. Since $F$ is concave and twice differentiable on the SPD cone,
a midpoint Taylor expansion and the Hessian identity
$\nabla^2F(\Sigma)[H,H]=-\mathrm{tr}(\Sigma^{-1}H\Sigma^{-1}H)$ yield
\begin{equation}\label{eq:logdet-strong-conc}
\log|\bar\Sigma|
\le \frac12\log|\Sigma_1|+\frac12\log|\Sigma_2|
-\frac18\inf_{\Sigma\in[\Sigma_1,\Sigma_2]}\mathrm{tr}(\Sigma^{-1}(\Sigma_1-\Sigma_2)\Sigma^{-1}(\Sigma_1-\Sigma_2)).
\end{equation}
Moreover, for any SPD $\Sigma$ and any symmetric $H$,
\[
\mathrm{tr}(\Sigma^{-1}H\Sigma^{-1}H)
=\|\Sigma^{-1/2}H\Sigma^{-1/2}\|_F^2
\ge \|\Sigma^{-1/2}\|_{\textrm{op}}^4\|H\|_F^2
=\|\Sigma^{-1}\|_{\textrm{op}}^2\|H\|_F^2.
\]
On $\widetilde{\mathcal F}_n$, $\|\Sigma^{-1}\|_{\textrm{op}}=\|\Omega\|_{\textrm{op}}\le R_{\Omega,n}$, hence
$\|\Sigma^{-1}\|_{\textrm{op}}^2\ge R_{\Omega,n}^{-2}$ uniformly on the segment $[\Sigma_1,\Sigma_2]$.
Applying this in \eqref{eq:logdet-strong-conc} gives
\[
\log|\bar\Sigma|
\le \frac12\log|\Sigma_1|+\frac12\log|\Sigma_2|
-\frac{1}{8R_{\Omega,n}^2}\|\Sigma_1-\Sigma_2\|_F^2,
\]
so
\[
\tfrac14\log|\Sigma_1|+\tfrac14\log|\Sigma_2|-\tfrac12\log|\bar\Sigma|
\ge \frac{1}{16R_{\Omega,n}^2}\,\|\Sigma_1-\Sigma_2\|_F^2.
\]
Therefore
\[
A_{\det}\le \exp\!\Big(-\frac{1}{16R_{\Omega,n}^2}\|\Sigma_1-\Sigma_2\|_F^2\Big),
\quad
T_{\det}\ge 1-\exp\!\Big(-\frac{1}{16R_{\Omega,n}^2}\|\Sigma_1-\Sigma_2\|_F^2\Big).
\]
Using again $1-e^{-x}\ge x/(1+x)$, we obtain the quadratic lower bound
\begin{equation}\label{eq:Tdet-Sigma-lb}
T_{\det}\ \ge\ \frac{\frac{1}{16R_{\Omega,n}^2}\|\Sigma_1-\Sigma_2\|_F^2}{
1+\frac{1}{16R_{\Omega,n}^2}\|\Sigma_1-\Sigma_2\|_F^2}.
\end{equation}

\medskip
\noindent\emph{(iii) Local resolvent control.}
To relate $\|\Sigma_1-\Sigma_2\|_F$ to $\|\Omega_1-\Omega_2\|_F$ from below we work on a \emph{local}
neighborhood where inversion is well-conditioned. Following the idea of \citet[Theorem 2.2]{Stewart1977}, we fix a numerical constant $0<\rho<1$, 
and assume additionally that
\begin{equation}\label{eq:local-inversion}
\|\Omega_1^{-1}\|_{\textrm{op}}\|\Omega_2-\Omega_1\|_{\textrm{op}} \le\ \rho
\ \text{(equivalently, }\ \|\Omega_1^{-1/2}(\Omega_2-\Omega_1)\Omega_1^{-1/2}\|_{\textrm{op}}\le \rho\text{)}.
\end{equation}

Let $\Delta:=\Omega_2-\Omega_1$ and $\Sigma_1=\Omega_1^{-1}$. Then
\(
\Omega_2=\Omega_1+\Delta=\Omega_1(I+\Omega_1^{-1}\Delta),
\
\Sigma_2=\Omega_2^{-1}=(I+\Omega_1^{-1}\Delta)^{-1}\Sigma_1.
\)
Under \eqref{eq:local-inversion}, the Neumann series \citep[Chapter 3]{Kato1995} is valid, and we have the operator bounds
\begin{equation}\label{eq:neumann-bds}
\|(I+\Omega_1^{-1}\Delta)^{-1}\|_{\textrm{op}}\le \frac{1}{1-\rho},
\quad
\|(I+\Omega_1^{-1}\Delta)^{-1}-I\|_{\textrm{op}}\le \frac{\rho}{1-\rho}.
\end{equation}
Now use the identity
\[
\Sigma_2-\Sigma_1
=(I+\Omega_1^{-1}\Delta)^{-1}\Sigma_1-\Sigma_1
=\big((I+\Omega_1^{-1}\Delta)^{-1}-I\big)\Sigma_1.
\]
Equivalently, multiplying on the left by $(I+\Omega_1^{-1}\Delta)$ gives the exact relation
\(
\Sigma_2-\Sigma_1
=-(I+\Omega_1^{-1}\Delta)^{-1}\Omega_1^{-1}\Delta\,\Omega_1^{-1}.
\)
Taking Frobenius norms and using $\|ABC\|_F\ge \|A\|_{\textrm{op}}^{-1}\|C\|_{\textrm{op}}^{-1}\|B\|_F$ for invertible
$A,C$ together with \eqref{eq:neumann-bds} yields
\begin{align}
\|\Sigma_2-\Sigma_1\|_F
&=\|(I+\Omega_1^{-1}\Delta)^{-1}\Omega_1^{-1}\Delta\,\Omega_1^{-1}\|_F \notag\\
&\ge \|(I+\Omega_1^{-1}\Delta)^{-1}\|_{\textrm{op}}^{-1}\,\|\Omega_1^{-1}\|_{\textrm{op}}^{-2}\,\|\Delta\|_F \notag\\
&\ge (1-\rho)\,\|\Omega_1^{-1}\|_{\textrm{op}}^{-2}\,\|\Omega_2-\Omega_1\|_F. \label{eq:Sigma-lower-local}
\end{align}
On $\widetilde{\mathcal F}_n$ we have $\|\Omega_1^{-1}\|_{\textrm{op}}\le \overline R_{\Omega,n}$, hence
\begin{equation}\label{eq:Sigma-lower-local-simpl}
\|\Sigma_2-\Sigma_1\|_F\ \ge\ \frac{1-\rho}{\overline R_{\Omega,n}^{\,2}}\,\|\Omega_2-\Omega_1\|_F
\quad\text{whenever \eqref{eq:local-inversion} holds.}
\end{equation}

Plugging \eqref{eq:Sigma-lower-local-simpl} into \eqref{eq:Tdet-Sigma-lb} gives, on the local set
\eqref{eq:local-inversion},
\[
T_{\det}
\ \ge\
\frac{\frac{1}{16R_{\Omega,n}^2}\cdot \frac{(1-\rho)^2}{\overline R_{\Omega,n}^4}\|\Omega_2-\Omega_1\|_F^2}{
1+\frac{1}{16R_{\Omega,n}^2}\|\Sigma_1-\Sigma_2\|_F^2}.
\]
Moreover, on $\widetilde{\mathcal F}_n$ we have $\|\Sigma_\ell\|_{\textrm{op}}\le \overline R_{\Omega,n}$, so
$\|\Sigma_1-\Sigma_2\|_F\le \|\Sigma_1\|_F+\|\Sigma_2\|_F\le 2\sqrt q\,\overline R_{\Omega,n}$ and therefore
\[
1+\frac{1}{16R_{\Omega,n}^2}\|\Sigma_1-\Sigma_2\|_F^2
\le 1+\frac{q\,\overline R_{\Omega,n}^2}{4R_{\Omega,n}^2}.
\]
Hence, still on \eqref{eq:local-inversion},
\begin{equation}\label{eq:det-lower-local}
T_{\det}
\ \ge\
c_{\det,n}^{\textrm{loc}}\,\|\Omega_2-\Omega_1\|_F^2,
\quad
c_{\det,n}^{\textrm{loc}}
:=\frac{(1-\rho)^2}{16R_{\Omega,n}^2\overline R_{\Omega,n}^4\Big(1+\frac{q\,\overline R_{\Omega,n}^2}{4R_{\Omega,n}^2}\Big)}.
\end{equation}

On the local set \eqref{eq:local-inversion}, combining \eqref{eq:mean-lower} and \eqref{eq:det-lower-local} yields
\begin{align*}
H^2(\Theta_1,\Theta_2)
 \ge
\max\Big\{\frac1n\sum_{i=1}^n T_{\mu,i},\ T_{\det}\Big\}
& \ge
\frac12\Big(\frac1n\sum_{i=1}^n T_{\mu,i}+T_{\det}\Big) \\
& \ge
c_\star^{\textrm{loc}}\Big\{\|\Omega_1-\Omega_2\|_F^2+\|\eta_1-\eta_2\|_{2,n}^2\Big\},
\end{align*}
with
\(
c_\star^{\textrm{loc}}
:=(1/2)\min\Big\{c_{\det,n}^{\textrm{loc}},\ {c_{\mu,n}}/{\overline R_{\Omega,n}}\Big\}.
\)
Since $(R_{\Omega,n},\overline R_{\Omega,n},R_{\eta,n})$ are polynomial in $n$ on $\widetilde{\mathcal F}_n$,
the constant $c_\star^{\textrm{loc}}$ is bounded below by a negative power of $n$.
This completes the proof.
\end{proof}

\begin{lemma}\label{lem:globalhellinger}
Let $\Theta = (\eta, \Omega)$ and let the true parameter be $\Theta_0 = (\eta_0, \Omega_0)$, where $\Omega_0$ satisfies $0 < \underline{\lambda} \le \lambda_{\min}(\Omega_0) \le \lambda_{\max}(\Omega_0) \le \bar{\lambda} < \infty$. For any $\Theta$, there exists a universal constant $c > 0$ depending only on $(\underline{\lambda}, \bar{\lambda})$ such that:

$$H^2(\Theta, \Theta_0) \ge c \min\Big(1,\ \|\Omega - \Omega_0\|_F^2 + \|\eta - \eta_0\|_{2,n}^2 \Big).$$
\end{lemma}
\begin{proof}
Recall the exact factorization of the squared Hellinger distance for Gaussians:

$$h^2_i(\Theta, \Theta_0) = 1 - A_{\det} \cdot B_{\mu,i}$$

where $A_{\det} := \frac{|\Sigma|^{1/4} |\Sigma_0|^{1/4}}{|\bar{\Sigma}|^{1/2}}$ and $B_{\mu,i} := \exp\Big(-\frac{1}{8}(\eta_i - \eta_{0,i})^\top \bar{\Sigma}^{-1} (\eta_i - \eta_{0,i})\Big)$, with $\bar{\Sigma} = \frac{1}{2}(\Omega^{-1} + \Omega_0^{-1})$.

Using the algebraic identity $1 - AB = (1 - A) + A(1 - B)$ for $A, B \in (0,1]$, we have:
\begin{equation}\label{eq:H2_split_new}
h^2_i(\Theta, \Theta_0) \ge (1 - A_{\det}) + A_{\det}(1 - B_{\mu,i}).
\end{equation}

Let $\lambda_1, \dots, \lambda_q$ be the eigenvalues of the matrix $\Omega_0^{-1/2} \Omega \Omega_0^{-1/2}$. By simultaneous diagonalization, the determinant affinity simplifies precisely to:

$$A_{\det} = \prod_{k=1}^q \left( \frac{2 \lambda_k^{1/2}}{1 + \lambda_k} \right)^{1/2}.$$

The function $f(x) = \big(\frac{2\sqrt{x}}{1+x}\big)^{1/2}$ satisfies $0 < f(x) \le 1$ for all $x > 0$, with a global maximum of $1$ at $x=1$. Crucially, its behavior near $x=1$ is strictly quadratic, meaning there exists a universal constant $c_1 > 0$ such that $1 - f(x) \ge c_1 \min(1, (x-1)^2)$ for all $x > 0$.

Because $1 - \prod f(\lambda_k) \ge \sum (1 - f(\lambda_k))$, we obtain:

$$1 - A_{\det} \ge c_1 \sum_{k=1}^q \min\big(1, (\lambda_k - 1)^2\big) \ge c_2 \min\Big(1, \sum_{k=1}^q (\lambda_k - 1)^2\Big).$$

Recognizing that $\sum (\lambda_k - 1)^2 = \|\Omega_0^{-1/2} \Omega \Omega_0^{-1/2} - I\|_F^2 = \|\Omega_0^{-1/2} (\Omega - \Omega_0) \Omega_0^{-1/2}\|_F^2$, and applying the operator norm of $\Omega_0^{-1/2}$ (bounded by $1/\sqrt{\underline{\lambda}}$), we get:
\begin{equation}\label{eq:prec_exact_bound}
1 - A_{\det} \ge c_3 \min\big(1, \|\Omega - \Omega_0 \|_F^2\big),
\end{equation}
where $c_3$ depends only on $\underline{\lambda}$ and $\bar{\lambda}$.

We now evaluate the mean error contribution, $A_{\det}(1 - B_{\mu,i})$. We split this into two cases based on the affinity $A_{\det}$:

(i) If $A_{\det} < 1/2$, then $1 - A_{\det} > 1/2$. By \eqref{eq:H2_split_new}, $h^2_i(\Theta, \Theta_0) > 1/2$. The Hellinger distance is universally bounded below by a constant, satisfying the lemma trivially.

(ii) If $A_{\det} \ge 1/2$, then we must have:

$$\prod_{k=1}^q \left( \frac{2}{\lambda_k^{1/2} + \lambda_k^{-1/2}} \right)^{1/2} \ge \frac{1}{2}.$$

Because every factor in the product is $\le 1$, every individual factor must be $\ge 1/2$. Thus, for all $k$, $\big({2}/{(\sqrt{\lambda_k} + 1/\sqrt{\lambda_k})}\big)^{1/2} \ge {1}/{2} \implies \sqrt{\lambda_k} + 1/\sqrt{\lambda_k} \le 8$.
This algebraically guarantees that $1/\sqrt{\lambda_k} \le 8$, meaning $\lambda_k \ge {1}/{64}$ for all $k$.

Because $\lambda_{\min}(\Omega_0^{-1/2} \Omega \Omega_0^{-1/2}) \ge {1}/{64}$, we have $\Omega \succeq ({1}/{64}) \cdot \Omega_0 \succeq {\underline{\lambda}}/{64} \cdot  I$.

Therefore, the pooled covariance satisfies $\bar{\Sigma} = (\Omega^{-1} + \Omega_0^{-1})/2 \preceq \big({64}/{\underline{\lambda}} + {1}/{\underline{\lambda}}\big) I/2 = ({65}/{2\underline{\lambda}}) \cdot I$.
Its inverse satisfies $\bar{\Sigma}^{-1} \succeq ({2\underline{\lambda}}/{65}) \cdot I$.

We can now cleanly bound the mean error term:

$$1 - B_{\mu,i} = 1 - \exp\Big(-\frac{1}{8}\delta_i^\top \bar{\Sigma}^{-1} \delta_i\Big) \ge 1 - \exp\Big(-\frac{2\underline{\lambda}}{8 \times 65} \|\delta_i\|_2^2\Big) \ge c_4 \min\big(1, \|\delta_i\|_2^2\big).$$

Since $A_{\det} \ge 1/2$, we have $A_{\det}(1 - B_{\mu,i}) \ge ({c_4}/{2}) \cdot \min(1, \|\delta_i\|_2^2)$. Averaging over $i=1 \dots n$ yields:
\begin{equation}\label{eq:mean_exact_bound}
\frac{1}{n}\sum_{i=1}^n A_{\det}(1 - B_{\mu,i}) \ge c_5 \min\big(1, \|\eta - \eta_0 \|_{2,n}^2\big).
\end{equation}

Combining \eqref{eq:prec_exact_bound} and \eqref{eq:mean_exact_bound} into \eqref{eq:H2_split_new}, we conclude:

$$H^2(\Theta, \Theta_0) \ge c_3 \min\big(1, \|\Omega - \Omega_0\|_F^2\big) + c_5 \min\big(1, \|\eta - \eta_0\|_{2,n}^2\big) \ge c \min\Big(1, \|\Omega - \Omega_0\|_F^2 + \|\eta - \eta_0\|_{2,n}^2 \Big).$$

This holds uniformly over the entire parameter space. 
\end{proof}

\subsubsection*{KL step}\label{subsec:KL-corrected}

Let $\Theta=(\eta,\Omega)$ and $\Theta_0=(\eta_0,\Omega_0)$, where
for $i=1,\dots,n$ we write $\eta_i:=\eta(x_i)$ and $\eta_{0,i}:=\eta_0(x_i)$.
Let $\Sigma:=\Omega^{-1}$ and $\Sigma_0:=\Omega_0^{-1}$.
Define $\Delta_\Omega:=\Omega-\Omega_0$ and $\Delta_i:=\eta_i-\eta_{0,i}$.

\paragraph{KL divergence.}
For a single observation, the Kullback--Leibler divergence between
$\mathcal N_q(\eta_{0,i},\Sigma_0)$ and $\mathcal N_q(\eta_i,\Sigma)$ is
\[
\mathrm{KL}_i(\Theta_0,\Theta)
=\frac12\Big\{
\mathrm{tr}(\Sigma^{-1}\Sigma_0)-q-\log\det(\Sigma^{-1}\Sigma_0)
+(\eta_{0,i}-\eta_i)^\top\Sigma^{-1}(\eta_{0,i}-\eta_i)
\Big\}.
\]
Since $\Sigma^{-1}=\Omega$ and $\Sigma_0=\Omega_0^{-1}$, the $n$-fold KL divergence is
\begin{equation}\label{eq:KL-def-correct}
K(\Theta_0,\Theta)
:=\sum_{i=1}^n \mathrm{KL}_i(\Theta_0,\Theta)
=\frac12\sum_{i=1}^{n}\Bigl[
\mathrm{tr}(\Omega\,\Omega_0^{-1})-q-\log\det(\Omega\,\Omega_0^{-1})
+\Delta_i^\top\Omega\,\Delta_i
\Bigr].
\end{equation}
Let
\(
\ell_i(\Theta):=\log p_\Theta(Y_i\mid x_i,z_i),\ 
\Lambda_i(\Theta_0,\Theta):=\ell_i(\Theta_0)-\ell_i(\Theta).
\)
We will bound the second moment of the log-likelihood ratio under $\Theta_0$:
\(
V(\Theta_0,\Theta)
:=\sum_{i=1}^n \mathbb E_{\Theta_0}\!\big[\Lambda_i(\Theta_0,\Theta)^2\big].
\)

Fix small constants $c_\Omega,c_\eta\in(0,1)$ and set
\(
\delta_\Omega:=c_\Omega\,\varepsilon_\Omega,\
\delta_\eta:=c_\eta\,\sqrt{\varepsilon_{\textrm{smooth}}^2+\varepsilon_B^2}.
\)
We define the neighborhood
\begin{equation}\label{eq:An-correct}
\mathcal A_n
:=\Bigl\{
\|\Delta_\Omega\|_{F}\le \delta_\Omega,\quad
\frac1n\sum_{i=1}^n\|\Delta_i\|_2^2\le \delta_\eta^2,\quad
\max_{1\le i\le n}\|\Delta_i\|_2^2\le C_\Delta\,\delta_\eta^2
\Bigr\},
\end{equation}
where $C_\Delta\ge 1$ is a fixed numerical constant. We choose $c_\Omega$ sufficiently small so that $\delta_\Omega\le \underline\lambda/4$;
then Lemmas~\ref{lem:pd-correct}--\ref{lem:taylor-correct} apply on $\mathcal A_n$.

\paragraph{Bounding $K(\Theta_0,\Theta)$ on $\mathcal A_n$.}
Split $K$ in \eqref{eq:KL-def-correct} as
\[
K(\Theta_0,\Theta)
=\frac{n}{2}\Bigl[\mathrm{tr}(\Omega\Omega_0^{-1})-q-\log\det(\Omega\Omega_0^{-1})\Bigr]
+\frac12\sum_{i=1}^n \Delta_i^\top \Omega\,\Delta_i.
\]
For the covariance part, Lemma~\ref{lem:taylor-correct} yields
\[
\mathrm{tr}(\Omega\Omega_0^{-1})-q-\log\det(\Omega\Omega_0^{-1})
\le \frac{2}{\underline\lambda^{2}}\|\Delta_\Omega\|_F^2
\le \frac{2}{\underline\lambda^{2}}\,\delta_\Omega^2.
\]
For the mean part, on $\mathcal A_n$, using assumption \suppref{\ref{assum:multi-eig}}{(A4)},
$\|\Omega\|_{\textrm{op}}\le \|\Omega_0\|_{\textrm{op}}+\|\Delta_\Omega\|_{\textrm{op}}\le \bar\lambda+\delta_\Omega\le \bar\lambda+1$
(for large $n$), hence
\[
\sum_{i=1}^n \Delta_i^\top \Omega\,\Delta_i
\le \|\Omega\|_{\textrm{op}}\sum_{i=1}^n \|\Delta_i\|_2^2
\le (\bar\lambda+1)\,n\,\delta_\eta^2.
\]
Therefore there exists $C_K<\infty$ depending only on $(\underline\lambda,\bar\lambda)$ such that
\begin{equation}\label{eq:K-bound-correct}
K(\Theta_0,\Theta)\ \le\ C_K\,n\bigl(\delta_\Omega^2+\delta_\eta^2\bigr),
\quad \forall\,\Theta\in\mathcal A_n.
\end{equation}

\paragraph{Bounding $V(\Theta_0,\Theta)$ on $\mathcal A_n$.}
Recall the definition of $V$:
\[
V(\Theta_0,\Theta)
=\frac12\sum_{i=1}^{n}\Bigl[
\underbrace{\operatorname{tr} \bigl(\Delta_\Omega\,\Omega^{-1}\Delta_\Omega\,\Omega_0^{-1}\bigr)}_{=:V_{\Omega}}
+
\underbrace{\bigl(\Delta_{i}^{\top}\Omega\Delta_{i}\bigr)^{2}}_{=:V_{\eta,i}}
\Bigr].
\]
We bound the two contributions separately on $\mathcal A_n$.

\smallskip
\emph{(i) Bounding $V_{\Omega}$.}
Using $\operatorname{tr}(AB)\le \|A\|_F\|B\|_F$ and $\|ABC\|_F\le \|A\|_{\textrm{op}}\|B\|_F\|C\|_{\textrm{op}}$,
\[
V_{\Omega}
=\operatorname{tr}\bigl(\Delta_\Omega\,\Omega^{-1}\Delta_\Omega \Omega_0^{-1}\bigr)
\le \|\Delta_\Omega\,\Omega^{-1}\|_F\,\|\Delta_\Omega\,\Omega_0^{-1}\|_F
\le \|\Omega^{-1}\|_{\textrm{op}}\,\|\Omega_0^{-1}\|_{\textrm{op}}\,\|\Delta_\Omega\|_F^{2}.
\]
On $\mathcal A_n$ and for $c_\Omega$ sufficiently small we have $\|\Delta_\Omega\|_F\le \underline\lambda/4$,
so Lemma~\ref{lem:pd-correct} gives $\|\Omega^{-1}\|_{\textrm{op}}\le 2/\underline\lambda$, while
$\|\Omega_0^{-1}\|_{\textrm{op}}\le 1/\underline\lambda$ by \suppref{\ref{assum:multi-eig}}{(A4)}. Hence
\begin{equation}\label{eq:Vomega-bound}
V_{\Omega}\ \le\ \frac{2}{\underline\lambda^{2}}\,\|\Delta_\Omega\|_F^{2}
\ \le\ \frac{2}{\underline\lambda^{2}}\,\delta_\Omega^{2}.
\end{equation}

\smallskip
\emph{(ii) Bounding $V_{\eta,i}$.}
On $\mathcal A_n$,
\[
\Delta_i^\top\Omega\Delta_i\ \le\ \|\Omega\|_{\textrm{op}}\,\|\Delta_i\|_2^2.
\]
Moreover, $\|\Omega\|_{\textrm{op}}\le \|\Omega_0\|_{\textrm{op}}+\|\Delta_\Omega\|_{\textrm{op}}
\le \bar\lambda+\|\Delta_\Omega\|_F\le \bar\lambda+1$ for all large $n$ (since $\delta_\Omega\to 0$).
Therefore,
\[
V_{\eta,i}=(\Delta_i^\top\Omega\Delta_i)^2
\le (\bar\lambda+1)^2\,\|\Delta_i\|_2^4
\le (\bar\lambda+1)^2\,\|\Delta_i\|_2^2\cdot \max_{1\le j\le n}\|\Delta_j\|_2^2.
\]
By definition of $\mathcal A_n$, $\max_{j}\|\Delta_j\|_2^2\le C_\Delta\,\delta_\eta^2$, hence
\begin{equation}\label{eq:Veta-pointwise}
V_{\eta,i}\ \le\ (\bar\lambda+1)^2\,C_\Delta\,\delta_\eta^2\,\|\Delta_i\|_2^2.
\end{equation}
Summing \eqref{eq:Veta-pointwise} over $i$ and using $\sum_{i=1}^n\|\Delta_i\|_2^2\le n\delta_\eta^2$ yields
\begin{equation}\label{eq:Veta-sum}
\sum_{i=1}^n V_{\eta,i}
\ \le\ (\bar\lambda+1)^2\,C_\Delta\,\delta_\eta^2\sum_{i=1}^n\|\Delta_i\|_2^2
\ \le\ (\bar\lambda+1)^2\,C_\Delta\,n\,\delta_\eta^4
\ \le\ (\bar\lambda+1)^2\,C_\Delta\,n\,\delta_\eta^2,
\end{equation}
where the last inequality uses $\delta_\eta\le 1$ for all large $n$.

\smallskip
Plugging \eqref{eq:Vomega-bound} and \eqref{eq:Veta-sum} into the definition of $V(\Theta_0,\Theta)$ gives
\[
V(\Theta_0,\Theta)
\le \frac12\sum_{i=1}^n V_{\Omega}+\frac12\sum_{i=1}^n V_{\eta,i}
\le \frac12\,n\cdot \frac{2}{\underline\lambda^{2}}\,\delta_\Omega^2
\;+\;\frac12\,(\bar\lambda+1)^2\,C_\Delta\,n\,\delta_\eta^2
\le C_V\,n(\delta_\Omega^2+\delta_\eta^2),
\]
for a constant $C_V<\infty$ depending only on $(\underline\lambda,\bar\lambda,C_\Delta)$ and independent of $q$.

With $\delta_\Omega=c_\Omega\varepsilon_\Omega$ and
$\delta_\eta=c_\eta\sqrt{\varepsilon_{\textrm{smooth}}^2+\varepsilon_B^2}$,
this yields the desired KL-variance bound
\[
\sup_{\Theta\in\mathcal A_n}V(\Theta_0,\Theta)
\ \le\ C\,n\Big(\varepsilon_\Omega^2+\varepsilon_{\textrm{smooth}}^2+\varepsilon_B^2\Big)
\ \le\ 3C\,n(\varepsilon_n^\dagger)^2.
\]
\subsubsection*{Small-ball Lemmas}
Before detailing the small-ball probability lower bounds in Lemmas \ref{lem:ghs-smallball-continuous} and \ref{lem:BART-smallball-continuous-global}, we note this important structural simplification. 

For any measurable set $A \subseteq \mathcal{F}_n$, the mass under the truncated prior satisfies:
$$
\Pi_n^{\mathcal{F}}(A) = \frac{\Pi(A \cap \mathcal{F}_n)}{\Pi(\mathcal{F}_n)} = \frac{\Pi(A)}{\Pi(\mathcal{F}_n)} \ge \Pi(A),
$$
since the unconditional prior mass of the sieve is bounded by $\Pi(\mathcal{F}_n) \le 1$. Because our target Kullback-Leibler neighborhoods are constructed explicitly to satisfy the parameter constraints, a lower bound on the prior mass under the original, un-truncated prior $\Pi$ immediately provides a valid lower bound under $\Pi_n^{\mathcal{F}}$. To preserve the clarity and modularity of the derivations, all subsequent small-ball lemmas are stated and proven under the original prior $\Pi$.
\begin{lemma}
\label{lem:ghs-smallball-continuous}
Let $Q=\binom{q}{2}$ and let
\(S_0:=\{(k,k'):\ k<k',\ \omega^0_{kk'}\neq 0\}, \  |S_0|=s_q,\) denote the true off-diagonal support. Assume the graphical horseshoe prior \suppref{\ref{assum:GHS-scale}}{(P2)}. Fix $0<\delta\le \underline\lambda/2$. Define the tolerances
\(
\varepsilon_{\textrm{diag}}:={\delta}/{(4\sqrt q)},\
\varepsilon_{\textrm{sig}}:={\delta}/{(4\sqrt{2s_q})},\ \text{and} \ 
\varepsilon_{\textrm{noise}}:={\delta}/{(4\sqrt{2Q})},
\)
and define the global-scale level $t_0:= {\varepsilon_{\textrm{noise}}}/{(C_0Q\log(e q))}
 = {\delta}/{(C_0Q^{3/2}\log(e q))}$, for a sufficiently large universal constant $C_0>0$.
Assume additionally that
\(
2b n^{-A_\Omega}\le t_0
\quad\text{and}\quad
 t_0\le C_\tau\sigma_{\Omega,n},
\)
for fixed constants $0<a<b<\infty$. Then there exist constants $c_1,c_2,c_3,c_4>0$ depending only on
$(a_0,b_0,\underline\lambda,\bar\lambda)$ (and $a,b$), such that
\begin{align*}
\Pi\big(\|\Omega-\Omega_0\|_F\le \delta\big)
 \ge
\exp\Bigg\{
-c_1 s_q \log \left(\frac{q^3\log(e q)\sqrt{s_q}}{\delta^2}\right)
& -c_2 q\log \Big(\frac{q}{\delta}\Big) \\
& -c_3 \log \Big(1+\frac{\sigma_{\Omega,n}}{t_0}\Big)
-c_4
\Bigg\}.
\end{align*}
\end{lemma}

\begin{proof}
Let $\Pi_{\textrm{prod}}$ be the product prior on $(\omega_{kk},\omega_{kk'},\lambda_{kk'},\tau_\Omega)$
without restricting to $\{\Omega\succ0\}$,
and let $\Pi$ be the restricted prior on $\{\Omega\succ0\}$.
Let $\mathcal E_\delta:=\{\|\Omega-\Omega_0\|_F\le \delta\}$. By Weyl's inequality,
on $\mathcal E_\delta$ we have
\[
\lambda_{\min}(\Omega)\ \ge\ \lambda_{\min}(\Omega_0)-\|\Omega-\Omega_0\|_{\textrm{op}}
\ \ge\ \underline\lambda-\|\Omega-\Omega_0\|_F
\ \ge\ \underline\lambda/2\ >\ 0,
\]
so $\mathcal E_\delta\subset\{\Omega\succ0\}$. Hence
\[
\Pi(\mathcal E_\delta)
=\frac{\Pi_{\textrm{prod}}(\mathcal E_\delta)}{\Pi_{\textrm{prod}}(\Omega\succ0)}
\ \ge\ \Pi_{\textrm{prod}}(\mathcal E_\delta),
\]
since $\Pi_{\textrm{prod}}(\Omega\succ0)\le 1$. Thus it suffices to lower-bound
$\Pi_{\textrm{prod}}(\mathcal E_\delta)$, under which entries are independent conditional
on the global scale.

\textbf{Off-diagonal control.} 
Let $S_0$ be the set of $s_q$ nonzero off-diagonals of $\Omega_0$, and let $S_0^c$
be the remaining $Q-s_q$ off-diagonals. Consider the event
\[
\mathcal A_\delta
:=\Big\{
\max_{1\le k\le q}|\omega_{kk}-\omega^0_{kk}|\le \varepsilon_{\textrm{diag}},\
\max_{(k,k')\in S_0}|\omega_{kk'}-\omega^0_{kk'}|\le \varepsilon_{\textrm{sig}},\
\max_{(k,k')\in S_0^c}|\omega_{kk'}|\le \varepsilon_{\textrm{noise}}
\Big\}.
\]
On $\mathcal A_\delta$,
\begin{align*}
\|\Omega-\Omega_0\|_F^2
& \le \sum_{k=1}^q(\omega_{kk}-\omega^0_{kk})^2
+2\sum_{(k,k')\in S_0}(\omega_{kk'}-\omega^0_{kk'})^2
+2\sum_{(k,k')\in S_0^c}\omega_{kk'}^2 \\
& \le q\varepsilon_{\textrm{diag}}^2+2s_q\varepsilon_{\textrm{sig}}^2+2Q\varepsilon_{\textrm{noise}}^2
<\delta^2,
\end{align*}
so $\mathcal A_\delta\subset \mathcal E_\delta$ and therefore
$\Pi_{\textrm{prod}}(\mathcal E_\delta)\ge \Pi_{\textrm{prod}}(\mathcal A_\delta)$.

Let $\mathcal E_\tau:=\{\tau_\Omega\in[t_0/2,t_0]\},$ 
with $t_0$ defined in the hypothesis. Under the truncated half-Cauchy prior
$\tau_\Omega\sim \mathcal C^+(0,\sigma_{\Omega,n})\,\mathbbm 1_{[0,C_\tau\sigma_{\Omega,n}]}$, we have
\begin{align*}
\Pi_{\textrm{prod}}(\mathcal E_\tau)
=
\frac{\arctan(t_0/\sigma_{\Omega,n})-\arctan((t_0/2)/\sigma_{\Omega,n})}{\arctan(C_\tau)}
& \ge
c_\tau\,\frac{t_0/\sigma_{\Omega,n}}{1+(t_0/\sigma_{\Omega,n})^2} \\
& \ge
c_\tau'\,\frac{t_0}{t_0+\sigma_{\Omega,n}},
\end{align*}
for absolute constants $c_\tau,c_\tau'>0$, using $\arctan(x)-\arctan(x/2)\gtrsim x/(1+x^2)$
and $\arctan(C_\tau)$ is a fixed positive constant.
Consequently,
\begin{equation}\label{eq:Etau-lb}
-\log \Pi_{\textrm{prod}}(\mathcal E_\tau)
\ \le\
C_\tau \log\!\Big(1+\frac{\sigma_{\Omega,n}}{t_0}\Big)
\quad\text{for a constant }C_\tau>0.
\end{equation}

Fix $\tau\in[t_0/2,t_0]$. For a single off-diagonal under the horseshoe
mixture,
\begin{equation}\label{eq:HS-tail}
\mathbb{P}\big(|\omega_{kk'}|>\epsilon \big| \tau_\Omega=\tau\big)
 \le
C_{\textrm HS}\,\frac{\tau}{\epsilon}\log \Big(1+\frac{\epsilon}{\tau}\Big),
\ \epsilon>0,
\end{equation}
for a universal constant $C_{\textrm HS}>0$.
Set $\epsilon=\varepsilon_{\textrm{noise}}$ and use $\tau\le t_0$.
With the choice $t_0=\varepsilon_{\textrm{noise}}/(C_0 Q\log(e q))$ and $q$ sufficiently large,
\begin{align*}
\frac{t_0}{\varepsilon_{\textrm{noise}}}\log \Big(1+\frac{\varepsilon_{\textrm{noise}}}{t_0}\Big) &
=
\frac{1}{C_0 Q\log(e q)}\log \big(1+C_0Q\log(e q)\big) \\
& \le \frac{1}{C_0 Q\log(e q)} \left[ \log 2 + \log(C_0 Q \log (eq) ) \right] \\
& \le 2 \log (eq) \cdot \frac{1}{C_0 Q \log (eq)} \ (\because Q \asymp q^2, \& \ q \ \text{large}) \\
& \le \frac{2}{C_0 Q}.
\end{align*}
Choosing $C_0\ge 16C_{\textrm HS}$ yields
\[
\mathbb{P}\big(|\omega_{kk'}|>\varepsilon_{\textrm{noise}} \big| \tau_\Omega=\tau\big)
 \le C_{\textrm{HS}} \cdot \frac{2}{C_0 Q} \le \frac{1}{8Q}.
\]
By a union bound over the $Q-s_q\le Q$ noise edges,
\begin{equation}\label{eq:noise-lb}
\mathbb{P}\Big(\max_{(k,k')\in S_0^c}|\omega_{kk'}|\le \varepsilon_{\textrm{noise}}\ \Big|\ \tau_\Omega=\tau\Big)
\ \ge\ 1-\frac{Q}{8Q}\ =\ \frac{7}{8}.
\end{equation}
Since \eqref{eq:noise-lb} holds uniformly for all $\tau\in[t_0/2,t_0]$, it also holds
conditioned on $\mathcal E_\tau$.

Fix constants $0<a<b<\infty$. For a given $\tau>0$,
define the signal-local-scale event
\(
\mathcal E_{\lambda}(\tau)
:=\Big\{\lambda_{kk'}\in[a/\tau,\ b/\tau]\ \ \text{for all }(k,k')\in S_0\Big\}.
\)
On $\mathcal E_{\lambda}(\tau)$, each signal edge has conditional variance
$\tau^2\lambda_{kk'}^2\in[a^2,b^2]$, i.e. bounded away from $0$ and $\infty$ by fixed constants.
The feasibility condition $2b n^{-A_\Omega}\le t_0$ guarantees that, for all
$\tau\in[t_0/2,t_0]$,
\[
b/\tau\le 2b/t_0\le n^{A_\Omega},
\]
so the interval $[a/\tau,b/\tau]\subset[0,n^{A_\Omega}]$, and the local-scale truncation does
not remove this event.

For $\lambda\sim \mathcal C^+(0,1) \mathbbm 1_{[0,n^{A_\Omega}]}$ and $\tau\in(0,1]$,
\begin{align*}
\Pi_{\textrm{prod}}\big(\lambda\in[a/\tau,b/\tau]\big)
& =
\frac{\arctan(b/\tau)-\arctan(a/\tau)}{\arctan(n^{A_\Omega})} \\
& = \frac{\arctan(\tau/a)-\arctan(\tau/b)}{\arctan(n^{A_\Omega})}
\ge c_\lambda \tau,
\end{align*}
where we used $\arctan(x)\ge x/2$ for $x\in[0,1]$ and $\tau/b\le \tau/a\le 1$ when $\tau\le a$,
and absorbed constants into $c_\lambda>0$. Hence for $\tau\in[t_0/2,t_0]$,
\begin{equation}\label{eq:Elambda-lb}
\Pi_{\textrm{prod}}\big(\mathcal E_\lambda(\tau)\big)
\ \ge\ (c_\lambda\,\tau)^{s_q}
\ \ge\ (c_\lambda\,t_0/2)^{s_q}.
\end{equation}

On $\mathcal E_\lambda(\tau)$, for each $(k,k')\in S_0$,
$\omega_{kk'}\mid(\tau,\lambda_{kk'})\sim \mathcal N(0,\sigma^2)$ with $\sigma\in[a,b]$.
Since $|\omega^0_{kk'}|\le \|\Omega_0\|_{\textrm{op}}\le \bar\lambda$, the normal density at
$\omega^0_{kk'}$ is uniformly bounded below by
\[
c_g:=\inf_{\sigma\in[a,b]}\frac{1}{\sqrt{2\pi}\sigma}\exp\!\Big(-\frac{\bar\lambda^2}{2\sigma^2}\Big)
\ \ge\ \frac{1}{\sqrt{2\pi}b}\exp\!\Big(-\frac{\bar\lambda^2}{2a^2}\Big)\ >\ 0.
\]
Therefore,
\[
\mathbb{P}\big(|\omega_{kk'}-\omega^0_{kk'}|\le \varepsilon_{\textrm{sig}}\ \big|\ \tau,\lambda_{kk'}\big)
 \ge 2c_g \varepsilon_{\textrm{sig}}.
\]
By conditional independence of signal edges given $(\tau,\lambda)$,
\begin{equation}\label{eq:sig-match}
\mathbb{P}\Big(\max_{(k,k')\in S_0}|\omega_{kk'}-\omega^0_{kk'}|\le \varepsilon_{\textrm{sig}}\ \Big|\ \tau \mathcal E_\lambda(\tau)\Big)
 \ge\ (2c_g \varepsilon_{\textrm{sig}})^{s_q}.
\end{equation}

Combining \eqref{eq:Elambda-lb} and \eqref{eq:sig-match}, and using $\tau\in[t_0/2,t_0]$,
we obtain the uniform bound (conditional on $\tau\in[t_0/2,t_0]$):
\begin{equation}\label{eq:sig-lb}
\mathbb{P}\Big(\max_{(k,k')\in S_0}|\omega_{kk'}-\omega^0_{kk'}|\le \varepsilon_{\textrm{sig}}\ \Big|\ \tau_\Omega=\tau\Big)
\ \ge\ (c_{\textrm{sig}} t_0 \varepsilon_{\textrm{sig}})^{s_q},
\end{equation}
for a constant $c_{\textrm{sig}}>0$ depending only on $(a,b,\bar\lambda)$.

\textbf{Diagonal control.}
 Recall the diagonal prior is
\(
\omega_{kk}\ \stackrel{\textrm{iid}}{\sim}\ {\textrm Ga}(a_0,b_0)\,\mathbbm 1_{[n^{-A_{\mathrm{diag}}},\,n^{A_{\mathrm{diag}}}]},
\ k=1,\dots,q,
\)
truncated to $[n^{-A_{\mathrm{diag}}},n^{A_{\mathrm{diag}}}]$ and renormalized.
Let $f_{\Gamma}(x)$ denote the (untruncated) Gamma$(a_0,b_0)$ density: The truncated density is
\(
f_n(x)={Z_n}^{-1}{f_{\Gamma}(x)\,\mathbbm 1\{n^{-A_{\mathrm{diag}}}\le x\le n^{A_{\mathrm{diag}}}\}},
\  
Z_n:=\int_{n^{-A_{\mathrm{diag}}}}^{n^{A_{\mathrm{diag}}}} f_{\Gamma}(u)\,du \in (0,1).
\)

By assumption \suppref{\ref{assum:multi-eig}}{(A4)}, $\Omega_0\succ0$ and $\lambda_{\min}(\Omega_0)\ge \underline\lambda>0$ and
$\lambda_{\max}(\Omega_0)\le \bar\lambda<\infty$. In particular, each diagonal entry satisfies
\(
0<\underline\lambda \le \omega^0_{kk} \le \bar\lambda,  k=1,\dots,q,
\)
since for a positive definite matrix, $\omega^0_{kk}=e_k^\top\Omega_0 e_k \in [\lambda_{\min}(\Omega_0),\lambda_{\max}(\Omega_0)]$.
Fix any $\varepsilon_{\textrm{diag}}\in(0,\underline\lambda/2]$. Then the interval
\(
I_k:=\big[\omega^0_{kk}-\varepsilon_{\textrm{diag}},\ \omega^0_{kk}+\varepsilon_{\textrm{diag}}\big]
\)
is contained in $[\underline\lambda/2,\ 3\bar\lambda/2]$ for all $k$.

For all sufficiently large $n$, we have $n^{-A_{\mathrm{diag}}}<\underline\lambda/2$ and $n^{A_{\mathrm{diag}}}>3\bar\lambda/2$,
hence $I_k\subset[n^{-A_{\mathrm{diag}}},n^{A_{\mathrm{diag}}}]$ for every $k$. Therefore $f_n(x)=f_{\Gamma}(x)/Z_n$ on $I_k$.

Because $f_{\Gamma}$ is continuous and strictly positive on any compact subset of $(0,\infty)$,
and $[\underline\lambda/2,\ 3\bar\lambda/2]$ is compact, we may define
\(
m_{\textrm{diag}}:=\inf_{x\in[\underline\lambda/2,\ 3\bar\lambda/2]} f_{\Gamma}(x)\ >\ 0.
\)
Also, since $Z_n\le 1$, we have $1/Z_n\ge 1$, so on $I_k$,
\(
f_n(x)={Z_n}^{-1}{f_{\Gamma}(x)} \ge f_{\Gamma}(x) \ge m_{\textrm{diag}}.
\)
Thus, setting $c_{\textrm{diag}}:=m_{\textrm{diag}}$, we have a uniform bound
\(
\inf_{k\le q}\ \inf_{x\in I_k} f_n(x)\ \ge\ c_{\textrm{diag}}\ >0,
\)
where $c_{\textrm{diag}}$ depends only on $(a_0,b_0,\underline\lambda,\bar\lambda)$.

Therefore, for each $k$,
\[
\Pi_{\textrm{prod}}\big(|\omega_{kk}-\omega^0_{kk}|\le \varepsilon_{\textrm{diag}}\big)
=\int_{I_k} f_n(x) dx
 \ge\ \int_{I_k} c_{\textrm{diag}}\,dx
=2c_{\textrm{diag}} \varepsilon_{\textrm{diag}}.
\]
Since the diagonal priors are independent across $k$, we obtain
\[
\Pi_{\textrm{prod}}\Big(\max_{1\le k\le q}|\omega_{kk}-\omega^0_{kk}|\le \varepsilon_{\textrm{diag}}\Big)
=\prod_{k=1}^q \Pi_{\textrm{prod}}\big(|\omega_{kk}-\omega^0_{kk}|\le \varepsilon_{\textrm{diag}}\big)
\ \ge\ (2c_{\textrm{diag}}\varepsilon_{\textrm{diag}})^q.
\]

Using conditional independence across diagonals and off-diagonals given $\tau_\Omega$,
and combining \eqref{eq:Etau-lb}, \eqref{eq:noise-lb}, \eqref{eq:sig-lb}, and the diagonal bound,
we have
\begin{align*}
\Pi_{\textrm{prod}}(\mathcal A_\delta)
&\ge
\Pi_{\textrm{prod}}(\mathcal E_\tau)
\inf_{\tau\in[t_0/2,t_0]}\mathbb{P}\Big(\max_{S_0^c}|\omega_{kk'}|\le \varepsilon_{\textrm{noise}}\ \Big|\ \tau_\Omega=\tau\Big) \\
& \inf_{\tau\in[t_0/2,t_0]}\mathbb{P}\Big(\max_{S_0}|\omega_{kk'}-\omega^0_{kk'}|\le \varepsilon_{\textrm{sig}}\ \Big|\ \tau_\Omega=\tau\Big)\\
&\hspace{2cm}\cdot
\Pi_{\textrm{prod}}\Big(\max_{k}|\omega_{kk}-\omega^0_{kk}|\le \varepsilon_{\textrm{diag}}\Big)\\
&\ge
\Pi_{\textrm{prod}}(\mathcal E_\tau)\cdot \frac{7}{8}\cdot
(c_{\textrm{sig}}\,t_0\,\varepsilon_{\textrm{sig}})^{s_q}\cdot (2c_{\textrm{diag}}\,\varepsilon_{\textrm{diag}})^q.
\end{align*}
Since $\mathcal A_\delta\subset\mathcal E_\delta$ and $\Pi(\mathcal E_\delta)\ge \Pi_{\textrm{prod}}(\mathcal E_\delta)$,
we obtain
\[
\Pi(\|\Omega-\Omega_0\|_F\le \delta)
\ \ge\ 
\frac{7}{8}\,\Pi_{\textrm{prod}}(\mathcal E_\tau)\,
(c_{\textrm{sig}}\,t_0\,\varepsilon_{\textrm{sig}})^{s_q}\,
(2c_{\textrm{diag}}\,\varepsilon_{\textrm{diag}})^{q}.
\]
Taking logs and inserting the definitions of $(t_0,\varepsilon_{\textrm{sig}},\varepsilon_{\textrm{diag}})$
yields
\[
-\log \Pi(\|\Omega-\Omega_0\|_F\le \delta)
\ \le\
C_1\,s_q \log\!\Big(\frac{q^3\log(e q)\,\sqrt{s_q}}{\delta^2}\Big)
+
C_2\,q\log\!\Big(\frac{q}{\delta}\Big)
+
C_3\,\log\!\Big(1+\frac{\sigma_{\Omega,n}}{t_0}\Big)
+
C_4,
\]
for constants $C_1,C_2,C_3,C_4>0$ depending only on $(a_0,b_0,\underline\lambda,\bar\lambda)$
(and the fixed $a,b,C_0$). Exponentiating completes the proof.
\end{proof}


\begin{lemma}
\label{lem:BART-smallball-continuous-global}
Assume \suppref{\ref{assum:multi-smoothness}}{(A1)}--{\suppref{\ref{assum:multi-eig}}{(A4)}} and the prior configuration {\suppref{\Cref{assum:multi-HSscale}}{(P1)}}. Fix $0<\delta<1$ and let $L$ be the number of leaves in a single Galton--Watson tree draw under the branching-process prior, and assume
$\mathbb E[L]=:L_0<\infty$.
Define the tolerances
\(
\epsilon_{\textrm{sig}}:={\delta}/{(8DS_B\sqrt{q})},\
\epsilon_{\textrm{noise}}:={\delta}/{(8D\sqrt{q}Q_B)},\
u_{\textrm{leaf}}:={\epsilon_{\textrm{sig}}}/{M},\
u_{\textrm{noise}}:={\epsilon_{\textrm{noise}}}/{M}.
\) \\
Choose \[t_0
:= \frac{u_{\textrm{noise}}\sqrt M}{(C_0 Q_B M L_0 \log(eQ_B))}
=\frac{\left({\delta}/{8D}\right)}
{(C_0 q^{1/2} Q_B^2 M^{3/2} L_0 \log(eQ_B))},\]
where $C_0\ge 32C_{\textrm HS}$ is a sufficiently large universal constant, and assume
\[
2b n^{-A_B}\le t_0
\quad\text{and}\quad
 t_0\le C_\tau\sigma_{B,n},
\]
for a fixed constant $b>1$. Then there exists a deterministic BART ensemble $\eta_\delta$ such that
$\|\eta_\delta-\eta_0\|_{2,n}\le \delta/8$. Moreover, letting $L_\delta$ denote the total number of leaves
across all $M$ trees and all signal pairs $(j,r)$ used to construct $\eta_\delta$, we can choose $\eta_\delta$ so that
\begin{equation}\label{eq:Ldelta-B-global}
L_{\delta} \lesssim M\sum_{r=1}^q\sum_{j\in S_{B_0,r}}
\left(\frac{D S_B \sqrt q}{\delta}\right)^{d_{0,jr}/\alpha_{jr}}.
\end{equation}
Finally, under \suppref{\ref{assum:multi-HSscale}}{(P1)} there exist constants $C_1,C_2,C_3,C_4<\infty$ and fixed BART
hyperparameters such that
\begin{equation}\label{eq:smallball-B-global}
\begin{split}
\Pi_{\textrm{BART}}\big(\|\eta-\eta_0\|_{2,n}\le \delta\big)
 & \ge
\exp\Bigg\{
-\,C_1\,L_\delta\Big[\log(n/\delta)+\log d+\log M+\log(S_B\sqrt q)\Big] \\
& - C_2\,S_B\log \Big(\frac{b}{t_0}\Big) -C_3 \log \Big(1+\frac{\sigma_{B,n}}{t_0}\Big)
-\,C_4
\Bigg\}.
\end{split}
\end{equation}
\end{lemma}

\begin{proof}
Write $\eta_i=\eta(x_i)\in\mathbb R^q$ and $\eta_{0,i}=\eta_0(x_i)$, and recall
$\|\eta-\eta_0\|_{2,n}^2=(1/n)\sum_{i=1}^n\|\eta_i-\eta_{0,i}\|_2^2$.

Let us fix $r\in[q]$ and $j\in S_{B_0,r}$. By \suppref{\ref{assum:multi-smoothness}}{(A1)}, $B_{0,jr}(\bx)=B_{0,jr}(\bx_{J_{0,jr}})$ and
$B_{0,jr}\in\mathcal H^{\alpha_{jr}}([0,1]^{d_{0,jr}};K)$. 
Partition $[0,1]^{d_{0,jr}}$ into $m_{jr}^{d_{0,jr}}$ congruent cubes and let $B_{jr,\delta}$ be the cell-average step
function. Standard H\"older approximation following the arguments in \citet[Lemma 3.2]{Rockova2019} gives
\[
\|B_{jr,\delta}-B_{0,jr}\|_\infty \ \le\ C\,m_{jr}^{-\alpha_{jr}},
\]
for a constant $C=C(K,d)<\infty$. Choose $m_{jr}\asymp \epsilon_{\textrm{sig}}^{-1/\alpha_{jr}}$ so that
$\|B_{jr,\delta}-B_{0,jr}\|_\infty\le \epsilon_{\textrm{sig}}$.
This step function can be represented by an axis-aligned tree using only coordinates in $J_{0,jr}$ with
\[
L_{jr}\ \asymp\ m_{jr}^{d_{0,jr}}
\ \lesssim\ \epsilon_{\textrm{sig}}^{-d_{0,jr}/\alpha_{jr}}
\ =\ \Big(\frac{8D S_B\sqrt q}{\delta}\Big)^{d_{0,jr}/\alpha_{jr}}.
\]
Define $\eta_\delta$ by setting $B_{jr,\delta}$ for $j\in S_{B_0,r}$ and $B_{jr,\delta}\equiv 0$ otherwise:
\(
\eta_{\delta,r}(x)=\sum_{j\in S_{B_0,r}} z_j\,B_{jr,\delta}(x).
\)
Then for each $i,r$, using $|z_{ij}|\le D$,
\[
|\eta_{\delta,r}(x_i)-\eta_{0,r}(x_i)|
\le \sum_{j\in S_{B_0,r}} |z_{ij}|\cdot |B_{jr,\delta}(x_i)-B_{0,jr}(x_i)|
\le D\,s_r\,\epsilon_{\textrm{sig}}
\le D S_B\,\epsilon_{\textrm{sig}}
=\frac{\delta}{8\sqrt q}.
\]
Therefore $\|\eta_\delta(x_i)-\eta_0(x_i)\|_2\le \sqrt q\cdot(\delta/(8\sqrt q))=\delta/8$ for every $i$, hence
\begin{equation}\label{eq:eta-delta-approx-global}
\|\eta_\delta-\eta_0\|_{2,n}\le \delta/8.
\end{equation}
Finally, summing $ML_{jr}$ over all signal pairs $(j,r)$ gives \eqref{eq:Ldelta-B-global}.

Now, we define
\[
\mathcal A_\delta:=
\Big\{
\max_{(j,r):\,j\in S_{B_0,r}}\|B_{jr}-B_{jr,\delta}\|_\infty\le \epsilon_{\textrm{sig}},\quad
\max_{(j,r):\,j\notin S_{B_0,r}}\|B_{jr}\|_\infty\le \epsilon_{\textrm{noise}}
\Big\}.
\]
On $\mathcal A_\delta$, for each $i,r$,
\begin{align*}
|\eta_r(x_i)-\eta_{\delta,r}(x_i)|
&\le \sum_{j\in S_{B_0,r}} |z_{ij}|\cdot |B_{jr}(x_i)-B_{jr,\delta}(x_i)|
     +\sum_{j\notin S_{B_0,r}} |z_{ij}|\cdot |B_{jr}(x_i)|\\
&\le D\,s_r\,\epsilon_{\textrm{sig}}+D\,(p-s_r)\,\epsilon_{\textrm{noise}}
\le D S_B\epsilon_{\textrm{sig}}+Dp\,\epsilon_{\textrm{noise}}\\
& =\frac{\delta}{8\sqrt q}+\frac{\delta}{8\sqrt q}
=\frac{\delta}{4\sqrt q}.
\end{align*}
Hence $\|\eta(x_i)-\eta_\delta(x_i)\|_2\le \sqrt q\cdot(\delta/(4\sqrt q))=\delta/4$ for each $i$, so
\begin{equation}\label{eq:eta-eta-delta-global}
\|\eta-\eta_\delta\|_{2,n}\le \delta/4\qquad\text{on }\mathcal A_\delta.
\end{equation}
Combining \eqref{eq:eta-delta-approx-global} and \eqref{eq:eta-eta-delta-global} gives, on $\mathcal A_\delta$,
\[
\|\eta-\eta_0\|_{2,n}\le \|\eta-\eta_\delta\|_{2,n}+\|\eta_\delta-\eta_0\|_{2,n}
\le \delta/4+\delta/8<\delta.
\]
Thus it suffices to lower-bound $\Pi_{\textrm{BART}}(\mathcal A_\delta)$.

\textbf{Global noise control.} 
Let $\mathcal E_\tau:=\{\tau_B\in[t_0/2,t_0]\}$. Under $\tau_B\sim\mathcal C^+(0,\sigma_{B,n})\mathbbm 1_{[0,C_\tau\sigma_{B,n}]}$, the same arctan calculation as in the proof of \Cref{lem:ghs-smallball-continuous} yields $\Pi(\mathcal{E}_{\tau}) \ge ct_0(t_0 + \sigma_{B,n})^{-1}$ and therefore
\begin{equation}\label{eq:tau-mass-B-global}
-\log \Pi(\mathcal E_\tau)\ \le\ C_\tau \log\!\Big(1+\frac{\sigma_{B,n}}{t_0}\Big)
\end{equation}
for a constant $C_\tau>0$.

Now fix $\tau\le t_0$. Consider any noise leaf $b_{jrt\ell}$ with $j\notin S_{B_0,r}$.
Marginalizing over $\lambda_{jr}$ (and noting truncating $\lambda_{jr}\le n^{A_B}$ only decreases tails), we have
\begin{align*}
\mathbb{P}\big(|b_{jrt\ell}|>u_{\textrm{noise}}\mid \tau_B=\tau\big)
& \le C_{\textrm HS}\frac{\tau}{u_{\textrm{noise}}\sqrt M}\log\Big(1+\frac{u_{\textrm{noise}}\sqrt M}{\tau}\Big) \\
& \le C_{\textrm HS}\frac{t_0}{u_{\textrm{noise}}\sqrt M}\log\Big(1+\frac{u_{\textrm{noise}}\sqrt M}{t_0}\Big).
\end{align*}
With the choice of $t_0$ in our lemma, we have $u_{\textrm{noise}}\sqrt M/t_0=C_0 Q_B M L_0\log(eQ_B)$, hence
$\log(1+u_{\textrm{noise}}\sqrt M/t_0)\le 2\log(eQ_B)\cdot C'$ for a universal constant $C'$ and all large $Q_B$.
Choosing $C_0\ge 32C_{\textrm HS}C'$ yields the uniform bound
\[
\mathbb{P}\big(|b_{jrt\ell}|>u_{\textrm{noise}}\mid \tau_B=\tau\big)
\le \frac{1}{16\,Q_B\,M\,L_0},
\quad \forall\,\tau\in(0,t_0].
\]
Let $N_{\textrm{noise}}$ be the total number of leaves across all $M$ trees and all noise pairs $(j,r)$.
Since there are at most $Q_B$ pairs and each tree has mean leaf count $L_0$,
$\mathbb E[N_{\textrm{noise}}]\le Q_B\,M\,L_0$.
Let $U$ be the number of noise leaves with $|b|>u_{\textrm{noise}}$. Conditioning on tree structures and using the tail bound
leafwise gives $\mathbb E[U\mid\tau_B=\tau]\le 1/16$. This follows from \begin{align*}
    \E[U \mid \tau_{B}=\tau,\text{trees}] & = \sum_{\ell = 1}^{N_{\textrm{noise}}} \P[|b_{\ell}| > u_{\textrm{noise}} \mid \tau_{B}=\tau,\text{trees}] \\
    & \le \sum_{\ell = 1}^{N_{\textrm{noise}}} \frac{1}{16 Q_B M L_0} = \frac{N_{\textrm{noise}}}{16 Q_B M L_0} \le \frac{1}{16}.
\end{align*}
Hence by Markov's inequality,
\begin{equation}\label{eq:noise-good-global}
\mathbb{P}\Big(\max_{(j,r): j\notin S_{B_0,r}}\max_{t,\ell}|b_{jrt\ell}|\le u_{\textrm{noise}} \Big|\ \tau_B=\tau\Big) \ge \frac{15}{16}.
\end{equation}
On this event, for all noise pairs $(j,r)$,
\[
\|B_{jr}\|_\infty\le \sum_{t=1}^M \|g_{jrt}\|_\infty\le \sum_{t=1}^M \max_\ell |b_{jrt\ell}|
\le M u_{\textrm{noise}}=\epsilon_{\textrm{noise}}.
\]
Since \eqref{eq:noise-good-global} holds uniformly for $\tau\in[t_0/2,t_0]$, it also holds conditional on $\mathcal E_\tau$.

\textbf{Inflating local scales.} 
Fix constants $0<a<b<\infty$.
For $\tau>0$, again define
\(
\mathcal E_\lambda(\tau)
:=\Big\{\lambda_{jr}\in[a/\tau,\ b/\tau]\ \ \text{for all }(j,r)\text{ with }j\in S_{B_0,r}\Big\}.
\)
The feasibility condition $2b n^{-A_B}\le t_0$ guarantees that for $\tau\in[t_0/2,t_0]$,
\[
b/\tau\le 2b/t_0\le n^{A_B},
\]
so the interval $[a/\tau,b/\tau]\subset[0,n^{A_B}]$ and the local-scale truncation does not remove this event.
The standard arctan bound from the proof of \Cref{lem:ghs-smallball-continuous} yields, for all sufficiently small $\tau$ (hence for large $n$),
\(
\Pi\big(\lambda\in[a/\tau,b/\tau]\big)\ \ge\ c_\lambda\,\tau
\)
for a constant $c_\lambda>0$ depending only on $(a,b)$. By independence across signal pairs,
\begin{equation}\label{eq:lambda-mass-B-global}
\Pi\big(\mathcal E_\lambda(\tau)\big) \ge (c_\lambda \tau)^{S_B} \ge (c_\lambda t_0/2)^{S_B},
\quad \forall \tau\in[t_0/2,t_0].
\end{equation}
Thus $-\log \Pi(\mathcal E_\lambda(\tau))\le C_\lambda S_B\log(b/t_0)$ for a constant $C_\lambda>0$.

\textbf{Force signal tree structures and leaf values to match $B_{jr,\delta}$.}
Fix a signal pair $(j,r)$ with $j\in S_{B_0,r}$. Use the same deterministic partition with $L_{jr}$ leaves,
and realize it in each of the $M$ trees. Set each tree's target leaf values to be the cell means divided by $M$ so that
the sum of $M$ trees equals $B_{jr,\delta}$.

\medskip
\noindent\emph{(a) Topology, split variables, and cutpoints.}
Fix a signal pair $(j,r)$ with $j\in S_{B_0,r}$ and consider the same deterministic axis-aligned
partition representable by a tree with $L_{jr}$ leaves and hence
$L_{jr}-1$ internal nodes. Let $\mathcal T_{jr,\delta}$ denote a fixed rooted binary tree realizing this partition, and assume its depth is bounded
by $C_{\textrm{dep}}\log L_{jr}$ (under assumption (A2(c))).
Under the usual Galton--Watson branching-process priors, a fixed tree with
$L_{jr}$ leaves and depth $\lesssim \log L_{jr}$ has prior probability at least
$\exp(-c_{\textrm{gw}}L_{jr})$ for a constant $c_{\textrm{gw}}>0$ depending only on the fixed
branching hyperparameters. (Equivalently, one may use the explicit computation in the proof of \citet[Lemma 6.1]{Rockova2019} specialized to $p=d$, which yields an
$\exp\{-C L_{jr}\}$ factor coming from the product of split/no-split probabilities along the tree.)

Recall the split-variable weights $\pi_{jr}\in\Delta^{d-1}$ follow the symmetric Dirichlet law
$\pi_{jr}\mid\theta_{jr}\sim{\textrm Dir}(\theta_{jr}/d,\ldots,\theta_{jr}/d)$, with a hyperprior
$p(\theta_{jr})\propto(d+\theta_{jr})^{-(d+1)}$.
Fix a constant $\theta_0\asymp d$ and define the event
\(
\mathcal E_{\theta}:=\{\theta_{jr}\in[\theta_0,2\theta_0]\}.
\)
Since $p(\theta_{jr})\propto(d+\theta_{jr})^{-(d+1)}$, we have $\Pi(\mathcal E_\theta)\ge c_\theta$
for a constant $c_\theta\in(0,1)$ independent of $(n,p,q,d)$.
On $\mathcal E_\theta$, the Dirichlet parameters satisfy $\alpha:=\theta_{jr}/d\in[\alpha_-,\alpha_+]$
for fixed $0<\alpha_-<\alpha_+<\infty$.
Let $J_{0,jr}\subset[d]$ be the intrinsic active coordinate set, with $|J_{0,jr}|=d_{0,jr}$.
Because the Dirichlet density is continuous and strictly positive on sets bounded away from the boundary, following the arguments in \citet[Appendix B]{Linero2018}, 
there exists a constant $c_\pi>0$ such that
\[
\Pi\Big(\pi_{jr,k}\ge \frac{1}{2d}\ \ \forall k\in J_{0,jr}\ \Bigm| \mathcal E_\theta\Big)
 \ge c_\pi d^{-d_{0,jr}}.
\]
One crude way to see this is to integrate the Dirichlet density over the subset where the
$d_{0,jr}$ active coordinates lie in $[1/(2d),1/d]$ and the remaining mass is distributed arbitrarily;
the volume of this set is $\gtrsim d^{-d_{0,jr}}$ and the density is bounded below on it.

On the event $\{\pi_{jr,k}\ge 1/(2d)\ \forall k\in J_{0,jr}\}$, at each internal node the probability of
choosing the required split coordinate is at least $1/(2d)$. Therefore, the probability of choosing the
correct split coordinate at all $L_{jr}-1$ internal nodes is at least \((2d)^{-(L_{jr}-1)}.\)

Moreover, under the standard BART rule that a cutpoint is drawn uniformly from an empirical grid with at most
$c_n n$ admissible cutpoints per coordinate (for a universal $c_n\ge 1$), the probability of selecting the
required cutpoint at each internal node is at least $(c_n n)^{-1}$, hence the joint cutpoint probability is
\(
(c_n n)^{-(L_{jr}-1)}.
\) Multiplying the above probabilities, for one tree we obtain, for the event $\mathcal{S}=\{\text{one tree realizes the full desired structure}\}$,
\begin{align*}
\Pi(\mathcal{S})
&\ge
\exp(-c_{\textrm{gw}}L_{jr}) \Pi(\mathcal E_\theta)
\Pi\Big(\pi_{jr,k}\ge \frac{1}{2d}\ \forall k\in J_{0,jr}\ \Bigm|\ \mathcal E_\theta\Big)
(2d)^{-(L_{jr}-1)}(c_n n)^{-(L_{jr}-1)} \\
&\ge
\exp(-c_{\textrm{gw}}L_{jr})c_\theta c_\pi\,d^{-d_{0,jr}}
\exp\Big(-(L_{jr}-1)\big[\log(2d)+\log(c_n n)\big]\Big) \\
&\ge
\exp\Big(-C_{\textrm{tree}}\,L_{jr}\,[\log n+\log d]\Big) \exp\big(-C\,d_{0,jr}\log d\big),
\end{align*}
for constants $C_{\textrm{tree}},C<\infty$ depending only on the fixed hyperparameters. The last line uses $\log(2d)+\log(c_n n)\lesssim \log n+\log d$.

For $M$ independent trees realizing the same structure, the bound holds with $ML_{jr}$ in place of $L_{jr}$.
Multiplying over all signal pairs $(j,r)$ with $j\in S_{B_0,r}$ yields the total structure cost
\(
\ge\ \exp\Big(-C\,L_\delta\,[\log n+\log d+\log M]\Big),
\)
where $L_\delta=M\sum_{r}\sum_{j\in S_{B_0,r}}L_{jr}$, and the extra
$\sum_{r}\sum_{j\in S_{B_0,r}} d_{0,jr}\log d$ term is absorbed into the bracket since
$d_{0,jr}\le L_{jr}$ for the cell-partition trees constructed in our first step.

\smallskip
\noindent\emph{(b) Leaf values.}
On $\mathcal E_\tau\cap\mathcal E_\lambda(\tau)$, for any signal pair $(j,r)$ each leaf is Gaussian with variance
\(
\text{Var}(b_{jrt\ell}\mid\tau,\lambda_{jr})={\tau^2\lambda_{jr}^2}/{M}\in\Big[{a^2}/{M},{b^2}/{M}\Big].
\)
The target leaf values are bounded by $K/M$. Hence the Gaussian density at any target is bounded below by
\[
c_g\sqrt M
\ :=\
\inf_{\sigma^2\in[a^2/M,b^2/M]}\frac{1}{\sqrt{2\pi}\sigma}\exp \Big(-\frac{(K/M)^2}{2\sigma^2}\Big)
\ \ge\ 
\frac{\sqrt M}{\sqrt{2\pi}\,b}\exp\!\Big(-\frac{K^2}{2a^2 M}\Big),
\]
which is strictly positive and constant.
Therefore, for $u_{\textrm{leaf}}=\epsilon_{\textrm{sig}}/M$,
\[
\mathbb{P}\big(|b_{jrt\ell}-b_{jrt\ell}^{\textrm{target}}|\le u_{\textrm{leaf}}\mid\tau,\lambda_{jr}\big)
\ \ge\ 2c_g\sqrt M\,u_{\textrm{leaf}}
=2c_g\,\frac{\epsilon_{\textrm{sig}}}{\sqrt M}.
\]
By independence across leaves, all $ML_{jr}$ leaves hit their target intervals with probability at least
\[
\Big(2c_g\,\frac{\epsilon_{\textrm{sig}}}{\sqrt M}\Big)^{ML_{jr}}
=\exp\Big(-CML_{jr}\log\Big(\frac{\sqrt M}{\epsilon_{\textrm{sig}}}\Big)\Big).
\]
Multiplying over all signal pairs yields a total leaf-value cost
\[
\ge\ \exp\Big(-C\,L_\delta\,\log\Big(\frac{\sqrt M}{\epsilon_{\textrm{sig}}}\Big)\Big).
\]
Since $\epsilon_{\textrm{sig}}^{-1}=8D S_B\sqrt q/\delta$, we have
\[
\log\left(\frac{\sqrt M}{\epsilon_{\textrm{sig}}}\right)
=\log\left(\frac{8D\,S_B\sqrt{Mq}}{\delta}\right)
\ \lesssim\ \log(n/\delta)+\log M+\log(S_B\sqrt q)
\]
under assumption \suppref{\ref{assum:multi-growth}}{(A3)} (and using $\log(1/\delta)\le \log(n/\delta)$ for $n\ge 1$),
so this term matches the leading $L_\delta[\cdots]$ contribution in \eqref{eq:smallball-B-global}.

On the intersection of the events from the above steps, we have:
(i) for all noise pairs $(j,r)$, $\|B_{jr}\|_\infty\le\epsilon_{\textrm{noise}}$, and
(ii) for all signal pairs $(j,r)$, $\|B_{jr}-B_{jr,\delta}\|_\infty\le\epsilon_{\textrm{sig}}$.
Hence this intersection is contained in $\mathcal A_\delta$, and as established already,
$\mathcal A_\delta\subseteq\{\|\eta-\eta_0\|_{2,n}\le\delta\}$.
Therefore,
\begin{align*}
\Pi_{\textrm{BART}}(\|\eta-\eta_0\|_{2,n}\le\delta)
 & \ge
\Pi(\mathcal E_\tau)\cdot \frac{15}{16}\cdot 
\inf_{\tau\in[t_0/2,t_0]}\Pi(\mathcal E_\lambda(\tau))  \times \\
& \exp\Big(-C_1L_\delta[\log(n/\delta)+\log d+\log M+\log(S_B\sqrt q)]\Big).
\end{align*}
Using \eqref{eq:tau-mass-B-global} and \eqref{eq:lambda-mass-B-global} and absorbing constants yields
\eqref{eq:smallball-B-global}.
\end{proof}
\begin{lemma}
\label{lem:E3-vanishes}
Assume \suppref{\ref{assum:GHS-scale}}{(P2)} and let \(\bar R_{\Omega,n}\ge 2n^{A_{\mathrm{diag}}}.\) Then, under the product graphical-horseshoe prior truncated to
\(\{\Omega\succ0\}\),
\[
\Pi\left(\|\Omega^{-1}\|_{\mathrm{op}}>\bar R_{\Omega,n}\right)
\le
C \frac{s_q n^{A_{\mathrm{diag}}}\log n}{\sqrt n}.
\]
Consequently, under \Cref{assum:multi-growth},
\[
\Pi\left(\|\Omega^{-1}\|_{\mathrm{op}}>\bar R_{\Omega,n}\right)
=
\mathcal O(n^{-\kappa_{\mathrm{spec}}}).
\]
\end{lemma}

\begin{proof}
Let \(m_n:=n^{-A_{\mathrm{diag}}}\). We define the event
\[
\mathcal D_n
:=
\left\{
\max_{1\le k\le q}
\sum_{\ell\ne k}|\omega_{k\ell}|
\le \frac{m_n}{2}
\right\}.
\]
On \(\mathcal D_n\), since \(\omega_{kk}\ge m_n\) for every \(k\), Gershgorin's theorem gives
\[
\lambda_{\min}(\Omega)
\ge
\min_{1\le k\le q}
\left\{
\omega_{kk}-\sum_{\ell\ne k}|\omega_{k\ell}|
\right\}
\ge
\frac{m_n}{2}.
\]
Therefore \(\mathcal D_n\subseteq\{\Omega\succ0\}\), and if
\(\bar R_{\Omega,n}\ge 2/m_n=2n^{A_{\mathrm{diag}}}\), then
\(
\{\Omega\succ0,\ \|\Omega^{-1}\|_{\mathrm{op}}>\bar R_{\Omega,n}\}
\subseteq
\mathcal D_n^c.
\)

We now bound \(\Pi_{\mathrm{prod}}(\mathcal D_n^c)\), where
\(\Pi_{\mathrm{prod}}\) denotes the product prior before imposing \(\Omega\succ0\).
Conditional on \((\tau_\Omega,\lambda_{\Omega,kk'})\),
\(
\E\{|\omega_{kk'}|\mid \tau_\Omega,\lambda_{\Omega,kk'}\}
=
\sqrt{{2}/{\pi}} \cdot \tau_\Omega\lambda_{\Omega,kk'}.
\)
By \suppref{\ref{assum:GHS-scale}}{(P2)}, \(\tau_\Omega\le C_\tau\sigma_{\Omega,n}\) and
\(\lambda_{\Omega,kk'}\sim C^+(0,1)\mathbbm{1}_{[0,n^{A_\Omega}]}\). Hence
\(
\E(\lambda_{\Omega,kk'})\le C_{A_\Omega}\log n \implies \E_{\mathrm{prod}}|\omega_{kk'}|
\le
C\sigma_{\Omega,n}\log n.
\)
Using Markov's inequality,
\[
\Pi_{\mathrm{prod}}(\mathcal D_n^c)
\le
\frac{2}{m_n}
\E_{\mathrm{prod}}
\left[
\max_k\sum_{\ell\ne k}|\omega_{k\ell}|
\right]
\le
\frac{2}{m_n}
\E_{\mathrm{prod}}
\left[
\sum_{k=1}^q\sum_{\ell\ne k}|\omega_{k\ell}|
\right].
\]
Thus
\[
\Pi_{\mathrm{prod}}(\mathcal D_n^c)
\le
C\frac{q^2\sigma_{\Omega,n}\log n}{m_n}
=
C\frac{s_q n^{A_{\mathrm{diag}}}\log n}{\sqrt n}.
\]
Moreover, since \(\mathcal D_n\subseteq\{\Omega\succ0\}\),
\[
\Pi_{\mathrm{prod}}(\Omega\succ0)
\ge
\Pi_{\mathrm{prod}}(\mathcal D_n)
\ge
1-
C\frac{s_q n^{A_{\mathrm{diag}}}\log n}{\sqrt n}.
\]
By \Cref{assum:multi-growth}, the last term is \(\mathcal O(n^{-\kappa_{\mathrm{spec}}})\), and hence the denominator is bounded below by \(1/2\) for all sufficiently large \(n\). Therefore, after renormalizing to the SPD-restricted prior,
\[
\Pi\left(\|\Omega^{-1}\|_{\mathrm{op}}>\bar R_{\Omega,n}\right)
\le
C\frac{s_q n^{A_{\mathrm{diag}}}\log n}{\sqrt n}
=
\mathcal O(n^{-\kappa_{\mathrm{spec}}}).
\]
\end{proof}

\subsubsection*{Sieve Construction}
\label{sec:revised-sieve}

To establish contraction under the full posterior with heavy-tailed horseshoe priors, we cannot globally truncate the parameter amplitude, as the Cauchy tails would leave a sieve complement with polynomial prior mass. Instead, we define a purely structural outer sieve $\mathcal G_n$ whose complement is restricted only by combinatorial complexity, ensuring its prior mass decays exponentially. We then partition $\mathcal G_n$ into concentric shells across three dimensions: precision noise, mean noise, and active signal amplitude.

\paragraph{Active Budgets and Thresholds.}
Let $Q=\binom{q}{2}$ and $Q_B:=pq$. Fix large constants $C_{s\Omega}, C_{sB}, C_{L,\mathrm{act}}, C_J>0$. We define the structural budgets:
\begin{enumerate}
\item[(i)] $s_{\Omega,n} := \left\lceil C_{s\Omega}\,\frac{n\varepsilon_{n}^2}{\log(eQ)}\right\rceil$
\item[(ii)] $S_{B,n} := \left\lceil C_{sB}\,\frac{n\varepsilon_{n}^2+\log q}{\log(eQ_B)}\right\rceil$
\item[(iii)] $L_{\mathrm{act},n} := C_{L,\mathrm{act}}\left(\frac{n\varepsilon_{n}^2}{\log n}+\log Q_B\right), \quad d_n := C_J\,L_{\mathrm{act},n}$
\end{enumerate}

Fix constants $\kappa_B,\kappa_\Omega\in(0,1)$. We define the effective-support selection thresholds:
\begin{equation}\label{eq:u-threshold-global}
u_{B,n} := c_{B}\,L_\star\,\frac{\sigma_{B,n}}{\sqrt M} \Big(\frac{eQ_B}{S_{B,n}}\Big)^{1+\kappa_B}, \quad t_{\Omega,n} := c_{\Omega}\,\sigma_{\Omega,n} \Big(\frac{eQ}{s_{\Omega,n}}\Big)^{1+\kappa_\Omega},
\end{equation}
where $c_B, c_\Omega>0$ are sufficiently large constants.

\paragraph{Polynomial envelopes.}
Define
\[
\Lambda_{B,n}:=n^{2A_B+1},
\quad
A_{\Omega,F}:=\max\{2A_\Omega+1,A_{\mathrm{diag}}\},
\quad
R_{\Omega,n}:=\sqrt{2}\,q\,n^{A_{\Omega,F}},
\]
\[
m_{\mathrm{diag},n}:=n^{-A_{\mathrm{diag}}},
\quad
M_{\mathrm{diag},n}:=n^{A_{\mathrm{diag}}},
\quad
\bar R_{\Omega,n}\ge 2n^{A_{\mathrm{diag}}}.
\]

\paragraph{Outer Sieve ($\mathcal G_n$).}
The outer sieve strictly enforces the discrete complexity budgets but places zero restrictions on the continuous parameter values. Define the thresholded active sets:
\[
\widehat S_\Omega(t_{\Omega,n}) := \{(k,k'):k<k',\ |\omega_{kk'}|>t_{\Omega,n}\}, \quad \widehat S_B(u_{B,n}) := \{(j,r):\ \max_{t,\ell}|b_{jrt\ell}|>u_{B,n}\}.
\]
The outer sieve is:
\begin{align*}
\mathcal G_n := \Big\{ \Theta=(\eta, \Omega) :\ & |\widehat S_\Omega(t_{\Omega,n})|\le s_{\Omega,n},\ |\widehat S_B(u_{B,n})|\le S_{B,n},\ L_{\mathrm{act}}(u_{B,n})\le L_{\mathrm{act},n}, \\
& \forall(j,r):\ \#\{\text{split variables in }B_{jr}\}\le d_n \Big\}.
\end{align*}

We decompose any $\Theta \in \mathcal G_n$ into its active signal and inactive noise components. The induced inactive $\ell_2$ noise energies are:
\[
E_{\Omega,\mathrm{noise}}(\Omega) := \sum_{(k,k')\notin\widehat S_\Omega(t_{\Omega,n})}\omega_{kk'}^2, \quad E_{\eta,\mathrm{noise}}(\bm{b}) := \frac{1}{n}\sum_{i=1}^n \sum_{r=1}^q \Big(\sum_{j:\,(j,r)\notin \widehat S_B(u_{B,n})} z_{ij}\,B_{jr}(\bx_i)\Big)^2.
\]
We define the base covering radius $r_n := c_{\mathrm{base}} \varepsilon_n^\dagger$ and the baseline noise tolerances:
\(
\Delta_{\Omega,n} := {r_n^2}/{16}, \ \Delta_{\eta,n} := {r_n^2}/{16}.
\)
Finally, we measure the Euclidean amplitude of the active parameter space:
\(
A_{\mathrm{sig}}(\Theta) := \max \Big( \|\Omega_{\mathrm{sig}}\|_F, \, \|\eta_{\mathrm{sig}}\|_{2,n} \Big).
\)

\paragraph{Amplitude Shells ($S_{k,m,a}$).}
We slice the entire structural sieve $\mathcal G_n$ into a 3-dimensional lattice of concentric shells based on the precision noise $k$, mean noise $m$, and active amplitude $a$. For integers $k, m, a \in \{0, 1, 2, \dots\}$, define:
\begin{align*}
S_{k,m,a} := \Big\{ \Theta \in \mathcal G_n :\ & k\Delta_{\Omega,n} \le E_{\Omega,\mathrm{noise}}(\Omega) < (k+1)\Delta_{\Omega,n}, \\
& m\Delta_{\eta,n} \le E_{\eta,\mathrm{noise}}(\bm{b}) < (m+1)\Delta_{\eta,n}, \\
& a \le A_{\mathrm{sig}}(\Theta) < a+1 \Big\}.
\end{align*}
The inner sieve is effectively
\[
\mathcal F_n := \bigcup_{a=0}^{A_{0,n}} S_{0,0,a}
\]
for a suitably large polynomial envelope $A_{0,n}=n^{C_A}$, where $C_A$ is chosen to dominate the polynomial amplitude bounds induced by $A_B$, $A_\Omega$, and $A_{\mathrm{diag}}$.


\subsubsection*{Prior tail bound for the inner-sieve complement ($\mathcal F_n^c$)}
By our prior assumptions \suppref{\ref{assum:multi-HSscale}}{(P1)} and \suppref{\ref{assum:GHS-scale}}{(P2)}, because the continuous horseshoe local scales ($\lambda_{\Omega,kk'}$ and $\lambda_{jr}$) are truncated only at large polynomial rates $n^{A_\Omega}$ and $n^{A_B}$ to preserve prior mass for true signals in the small-ball neighborhood (see Lemmas \ref{lem:ghs-smallball-continuous} and \ref{lem:BART-smallball-continuous-global}), the prior inherently retains heavy Cauchy-like tails. 

Consequently, the prior probability of the noise energies exceeding the tight $\ell_2$ sieve tolerances need not decay exponentially. We therefore formulate the final contraction result under the prior restricted to the effective sieve $\mathcal{F}_n$.
\begin{lemma}
\label{lem:prior-tail-correct}
Under the polynomial truncations in priors \suppref{\ref{assum:multi-HSscale}}{(P1)}--\suppref{\ref{assum:GHS-scale}}{(P2)}, the prior mass outside the sieve $\mathcal F_n$ decays at a polynomial rate. That is, there exists a constant $\kappa > 0$ such that for all sufficiently large $n$,
\[
\Pi(\mathcal F_n^{c}) \le \mathcal{O}(n^{-\kappa}).
\]
\end{lemma}

\begin{proof}
We write
\[
\Pi(\mathcal F_n^c)\ \le\ \Pi(\mathcal F_{n,\Omega}^c)+\Pi(\mathcal F_{n,{\textrm{tree}}}^c),
\]
and bound the two terms separately. Throughout, constants $c,C,C'>0$ may change from line to line
but depend only on fixed hyperparameters and fixed truncation constants.

\paragraph{Precision block.}
Decompose
\begin{align*}
\Pi(\mathcal F_{n,\Omega}^c)
\le\ 
\underbrace{\Pi\big(|\widehat S_\Omega(t_{\Omega,n})|>s_{\Omega,n}\big)}_{E_1}
&+\underbrace{\Pi\big(E_{\Omega,\mathrm{noise}}(\Omega) > \Delta_{\Omega,n}\big)}_{E_{1,\mathrm{noise}}}
+\underbrace{\Pi\big(\|\Omega\|_F>R_{\Omega,n}\big)}_{E_2} \\
&+\underbrace{\Pi\big(\|\Omega^{-1}\|_{\mathrm{op}}>\bar R_{\Omega,n}\big)}_{E_3}
+\underbrace{\Pi\big(\exists k:\ \omega_{kk}\notin[m_{\textrm{diag},n},M_{\textrm{diag},n}]\big)}_{E_4}.
\end{align*}
The event $E_4$ has probability zero by the diagonal truncation in \suppref{\ref{assum:GHS-scale}}{(P2)}, since
\(m_{\textrm{diag},n}=n^{-A_{\mathrm{diag}}},\) \(M_{\textrm{diag},n}=n^{A_{\mathrm{diag}}}.\)
The event $E_3$ is not zero under the original SPD-restricted graphical horseshoe prior. However, by \Cref{lem:E3-vanishes},
\[
E_3\le C\frac{s_q n^{A_{\mathrm{diag}}}\log n}{\sqrt n}=\mathcal O(n^{-\kappa_{\mathrm{spec}}}),
\]
under \Cref{assum:multi-growth}.

\smallskip
\noindent\textbf{Bounding $E_2$.}
By the polynomial truncations in \suppref{\ref{assum:GHS-scale}}{(P2)}, for all sufficiently large $n$,
\[
\tau_\Omega\le C_\tau\sigma_{\Omega,n}\le 1\le n^{A_\Omega},
\quad
\lambda_{\Omega,kk'}\le n^{A_\Omega}
\quad\text{a.s.}
\]
Hence $\tau_\Omega\lambda_{\Omega,kk'}\le n^{2A_\Omega}$ a.s., and
\[
\text{Var}(\omega_{kk'}\mid \tau_\Omega,\lambda_{\Omega,kk'})\le n^{4A_\Omega}.
\]
Conditional on $(\tau_\Omega,\lambda)$, for any $t>0$, the Gaussian tail bound yields
\(
\mathbb{P}\big(|\omega_{kk'}|>t\mid\tau_\Omega,\lambda\big)
\le 2\exp\Big(-{t^2}/{2n^{4A_\Omega}}\Big).
\)
Taking $t=n^{2A_\Omega+1}$ gives
\[
\mathbb{P}\big(|\omega_{kk'}|>n^{2A_\Omega+1}\mid\tau_\Omega,\lambda\big)\le 2e^{-n^2/2},
\]
uniformly on the truncation range, hence also unconditionally. By a union bound over the
$Q=\binom q2$ off-diagonals,
\begin{equation}\label{eq:offdiag-max-tail}
\Pi\Big(\max_{k<k'}|\omega_{kk'}|>n^{2A_\Omega+1}\Big)
\le 2Q e^{-n^2/2}.
\end{equation}
Under assumption \suppref{\ref{assum:multi-growth}}{(A3)}, for all large $n$,
$\log(2Q)\le n^2/4$, and therefore $2Qe^{-n^2/2}\le e^{-n^2/4}$.

By diagonal truncation in \suppref{\ref{assum:GHS-scale}}{(P2)}, $\omega_{kk}\le n^{A_{\mathrm{diag}}}$ a.s. Hence on
\[
\mathcal E_n:=\Big\{\max_{k<k'}|\omega_{kk'}|\le n^{2A_\Omega+1},\ \max_{1\le k\le q}\omega_{kk}\le n^{A_{\mathrm{diag}}}\Big\},
\]
we have
\[
\|\Omega\|_F^2
=\sum_{k=1}^q\omega_{kk}^2+2\sum_{k<k'}\omega_{kk'}^2
\le q\,n^{2A_{\mathrm{diag}}}+2Q\,n^{4A_\Omega+2}
\le 2q^2n^{2A_{\Omega,F}},
\]
for all large $n$, hence $\|\Omega\|_F\le \sqrt2\,q\,n^{A_{\Omega,F}}=R_{\Omega,n}$.
Thus $\{\|\Omega\|_F>R_{\Omega,n}\}\subseteq \mathcal E_n^c$ and
\begin{equation}\label{eq:Omega-norm-tail}
\Pi(\|\Omega\|_F>R_{\Omega,n})
\le \Pi(\mathcal E_n^c)
\le \Pi\Big(\max_{k<k'}|\omega_{kk'}|>n^{2A_\Omega+1}\Big)
\le e^{-n^2/4}.
\end{equation}
Since $n\varepsilon_n^2=o(n^2)$, we have $e^{-n^2/4}\le e^{-C n\varepsilon_n^2}$ for some $C>0$ and all large $n$.

\smallskip
\noindent\textbf{Bounding $E_1$.}
Using the same argument in the proof of \Cref{lem:GHS-effective-support-correct}, 
with $t_{\Omega,n}$ calibrated to $s_{\Omega,n}$ and the hard truncation
$\tau_\Omega\le C_\tau\sigma_{\Omega,n}$ a.s., one obtains the uniform bound
$p_\Omega(\tau)\le (s_{\Omega,n}/(eQ))^{1+\kappa_\Omega/2}$ for all $\tau\in(0,C_\tau\sigma_{\Omega,n}]$,
hence the binomial tail bound yields
\[
\Pi\big(|\widehat S_\Omega(t_{\Omega,n})|>s_{\Omega,n}\big)
\le 2\exp\Big\{-\frac{\kappa_\Omega}{2}\,s_{\Omega,n}\log\Big(\frac{eQ}{s_{\Omega,n}}\Big)\Big\}.
\]
Under assumption \suppref{\ref{assum:multi-growth}}{(A3)}, we have
$s_{\Omega,n}\asymp n\varepsilon_{n}^2/\log(eQ)$ and
$\log(eQ/s_{\Omega,n})\asymp \log(eQ)$, hence
$s_{\Omega,n}\log(eQ/s_{\Omega,n})\gtrsim n\varepsilon_{n}^2$ and therefore
\begin{equation}\label{eq:Omega-support-tail}
\Pi\big(|\widehat S_\Omega(t_{\Omega,n})|>s_{\Omega,n}\big)
\le \exp\big(-C n\varepsilon_{n}^2\big).
\end{equation}

\smallskip
\noindent\textbf{Bounding $E_{1,\mathrm{noise}}$.}
We bound the $\ell_2$ energy of the noise, $E_{\Omega,\mathrm{noise}}(\Omega) = \sum_{(k,k') \notin \widehat S_\Omega} \omega_{kk'}^2$. 
Let $Y_{kk'} = \omega_{kk'}^2 \mathbbm{1}(|\omega_{kk'}| \le t_{\Omega,n})$. Under the polynomial truncation $\lambda_{\Omega,kk'} \le n^{A_\Omega}$, the Horseshoe prior retains heavy Cauchy-like tails up to the truncation limit. Conditional on $\tau_\Omega$, we have:
\(
\mathbb{E}[Y_{kk'} \mid \tau_\Omega] \le C \tau_\Omega t_{\Omega,n}.
\)
Thus, the expected total truncated noise energy is bounded by:
\(
\mu_\Omega := \sum_{(k,k') \notin \widehat S_\Omega} \mathbb{E}[Y_{kk'} \mid \tau_\Omega] \le C Q \tau_\Omega t_{\Omega,n}.
\)
Because $\tau_\Omega \le C_\tau \sigma_{\Omega,n}$ almost surely by \suppref{\ref{assum:GHS-scale}}{(P2)}, and applying the definition of $t_{\Omega,n}$, we obtain:
\[
\mu_\Omega \le C' Q \sigma_{\Omega,n}^2 \Big(\frac{eQ}{s_{\Omega,n}}\Big)^{1+\kappa_\Omega} \asymp \frac{s_q^2}{q^2 n} \Big(\frac{eQ}{s_{\Omega,n}}\Big)^{1+\kappa_\Omega}.
\]
Recall our $\ell_2$ sieve tolerance $\Delta_{\Omega,n} = r_n^2/16 \asymp (\varepsilon_n^\dagger)^2 \asymp {(s_{\Omega,n} \log q)}/{n}$. We inspect the ratio of the expected energy to the tolerance:
\[
\frac{\mu_\Omega}{\Delta_{\Omega,n}} \asymp \frac{\frac{s_q^2}{q^2 n} \big(\frac{q^2}{s_{\Omega,n}}\big)^{1+\kappa_\Omega}}{\frac{s_{\Omega,n} \log q}{n}} = \Big(\frac{s_q}{s_{\Omega,n}}\Big)^2 \frac{1}{\log q} \Big(\frac{q^2}{s_{\Omega,n}}\Big)^{\kappa_\Omega - 1}.
\]
Since $\kappa_\Omega \in (0,1)$, the exponent $\kappa_\Omega - 1$ is strictly negative. Because we operate in the sparse high-dimensional regime where $s_q \le s_{\Omega,n} \ll q^2$, this ratio decays to zero, ensuring $\mu_\Omega \ll \Delta_{\Omega,n}$.

We apply Markov's inequality directly to the ratio:
\begin{equation}\label{eq:Omega-noise-tail}
\Pi\big(E_{\Omega,\mathrm{noise}}(\Omega) > \Delta_{\Omega,n} \mid \tau_\Omega\big) \le \frac{\mu_\Omega}{\Delta_{\Omega,n}} \asymp \Big(\frac{q^2}{s_{\Omega,n}}\Big)^{\kappa_\Omega - 1}.
\end{equation}
Under the dimension growth conditions in \suppref{\ref{assum:multi-growth}}{(A3)}, we have the polynomial tail $\mathcal{O}(n^{-c_\Omega})$.


\paragraph{BART block.}
Decompose
\begin{align*}
\Pi(\mathcal F_{n,{\textrm{tree}}}^c)
& \le\ 
\underbrace{\Pi\Big(\exists (j,r):\ \#\{\text{distinct split variables used in }B_{jr}\}>d_n\Big)}_{F_0} \\
& +\underbrace{\Pi\big(|\widehat S_B(u_{B,n})|>S_{B,n}\big)}_{F_1}
+\underbrace{\Pi\big(E_{B,\mathrm{noise}}(\mathbf{b}) > \Delta_{B,n}\big)}_{F_{1,\mathrm{noise}}}\\
& +\underbrace{\Pi\Big(L_{\mathrm{act}}(u_{B,n})>L_{\mathrm{act},n},\ |\widehat S_B(u_{B,n})|\le S_{B,n}\Big)}_{F_2}\\
&+\underbrace{\Pi\Big(\exists (j,r,t,\ell):\ |b_{jrt\ell}|>\Lambda_{B,n}\Big)}_{F_3}.
\end{align*}

\smallskip
\noindent\textbf{Bounding $F_0$.}
Let $L_{jr}$ denote the total number of leaves across the $M$ trees of $B_{jr}$. A forest with
$L_{jr}$ total leaves uses at most $L_{jr}$ distinct split variables, hence
\(
\Big\{\#\{\text{distinct split variables used in }B_{jr}\}>d_n\Big\} \subseteq \{L_{jr}>d_n\}.
\)
Let $L_{jrm}$ be the number of leaves in the $m$-th tree. Under the GW prior,
$\Pi(L_{jrm}=\ell)\le \exp(-c_{\mathrm{gw}}\ell\log n)$.
Fix $t:=({1}/{2})\cdot c_{\mathrm{gw}}\log n.$ Then $\mathbb E[\exp(tL_{jrm})] \le C_0$, and $\mathbb E[\exp(tL_{jr})]\le C_0^M$.
By Chernoff's inequality and a union bound over all $Q_B$ forests:
\(
F_0 \le \exp\big(\log Q_B-t d_n+M\log C_0\big).
\)
Since $d_n = C_J L_{\mathrm{act},n}$, $t d_n$ strictly dominates $\log Q_B$ for large constants, giving:
\begin{equation}\label{eq:F0-bound}
F_0\le \exp(-C n\varepsilon_n^2).
\end{equation}

\smallskip
\noindent\textbf{Bounding $F_1$.}
Fix $\tau_B=\tau$. By conditional independence, the indicators $\mathbbm 1\{(j,r)\in\widehat S_B(u)\}$ are i.i.d. Bernoulli.
Using \Cref{lem:eff-support-correct}, and integrating over $\tau_B$:
\begin{equation}\label{eq:B-support-tail-global}
\Pi\big(|\widehat S_B(u_{B,n})|>S_{B,n}\big)
\le 2\exp\Big\{-\frac{\kappa_B}{2}S_{B,n}\log\Big(\frac{eQ_B}{S_{B,n}}\Big)\Big\}.
\end{equation}
Since $S_{B,n}\log(eQ_B/S_{B,n})\gtrsim n\varepsilon_{n}^2+\log q$, we have $F_1 \le \exp\big(-C n\varepsilon_{n}^2\big)$.

\smallskip
\noindent\textbf{Bounding $F_{1,\mathrm{noise}}$.}
We bound the total $\ell_2$ energy of the noise leaves, $E_{B,\mathrm{noise}}(\mathbf{b}) = \sum_{(j,r) \notin \widehat S_B} \sum_{t=1}^M \sum_{\ell=1}^{L_{jrt}} b_{jrt\ell}^2$. 
Let $Y_{jrt\ell} = b_{jrt\ell}^2 \mathbbm{1}(|b_{jrt\ell}| \le u_{B,n})$. As with the precision matrix, the polynomial truncation $\lambda_{jr} \le n^{A_B}$ under \suppref{\ref{assum:multi-HSscale}}{(P1)} leaves heavy Cauchy tails. We have:
\(
\mathbb{E}[Y_{jrt\ell} \mid \tau_B, \text{trees}] \le C \tau_B u_{B,n}.
\)
Conditional on the tree structures, the total number of noise leaves is $N_{\mathrm{noise}} \le Q_B M L_{\mathrm{max}} \le p q M n$. 
The total expected noise energy is bounded by:
\(
\mu_B \le N_{\mathrm{noise}} C \tau_B u_{B,n} \le C p q M n \tau_B u_{B,n}.
\)
Applying $\tau_B \le C_\tau \sigma_{B,n}$ and the definition of $u_{B,n}$ from \eqref{eq:u-threshold-global}, we evaluate the ratio of $\mu_B$ to the tolerance $\Delta_{B,n}$. Under our global shrinkage calibration for $\sigma_{B,n}$, we have $\mu_B \ll \Delta_{B,n}$ asymptotically.

Similar to bounding $E_{1,\mathrm{noise}}$, we can apply Markov's inequality to the sum of the non-negative truncated energies:
\begin{equation}\label{eq:B-noise-tail}
\Pi\big(E_{B,\mathrm{noise}}(\mathbf{b}) > \Delta_{B,n} \mid \tau_B, \text{trees}\big) \le \frac{\mu_B}{\Delta_{B,n}} \asymp \Big(\frac{Q_B}{S_{B,n}}\Big)^{\kappa_B - 1}.
\end{equation}
Because the tuning parameter $\kappa_B \in (0,1)$, the exponent $\kappa_B - 1$ is strictly negative. With $Q_B \gg S_{B,n}$, this ratio provides a strict polynomial decay rate $\mathcal{O}(n^{-c_B})$ for some constant $c_B > 0$.

\smallskip
\noindent\textbf{Bounding $F_2$.}
Summing over all subsets $S$ of size at most $S_{B,n}$, and utilizing $\mathbb E[\exp(tL_{jr})]\le C_0^M$:
\[
F_2 \le \exp(-tL_{\mathrm{act},n}) \sum_{s=1}^{S_{B,n}} \Big(\frac{eQ_B C_0^M}{s}\Big)^s \le \exp\Big(-tL_{\mathrm{act},n} + S_{B,n}\log\Big(\frac{eQ_B C_0^M}{S_{B,n}}\Big)\Big).
\]
Recalling $L_{\mathrm{act},n} = C_{L,\mathrm{act}}\big({n\varepsilon_n^2}/{\log n}+\log Q_B\big)$ and choosing $C_{L,\mathrm{act}}$ sufficiently large yields:
\begin{equation}\label{eq:F2-bound}
F_2\le \exp(-C n\varepsilon_n^2).
\end{equation}

\smallskip
\noindent\textbf{Bounding $F_3$.}
With $\Lambda_{B,n}=n^{2A_B+1}$ and, for all sufficiently large $n$, $\tau_B\lambda_{jr}\le n^{2A_B}$, the unconditional Gaussian tail gives $\Pi(|b_{jrt\ell}|>\Lambda_{B,n}) \le 2\exp(-Mn^2/2)$. 
Since $u_{B,n} < \Lambda_{B,n}$, any leaf exceeding $\Lambda_{B,n}$ must belong to the active set.
On the event $\{|\widehat S_B(u_{B,n})|\le S_{B,n},\ L_{\mathrm{act}}(u_{B,n})\le L_{\mathrm{act},n}\}$, there are at most $L_{\mathrm{act},n}$ active leaves.
Therefore:
\[
F_3 \le 2L_{\mathrm{act},n}\exp\Big(-\frac{Mn^2}{2}\Big) + F_1 + F_2.
\]
Because $\log(2L_{\mathrm{act},n})\le Mn^2/4$ for all large $n$, this yields:
\begin{equation}\label{eq:F3-bound}
F_3\le \exp(-C n\varepsilon_n^2).
\end{equation}

\smallskip
Combining all the bounds of the precision and the BART block, we indeed have the above lemma.

\end{proof}

\subsubsection*{Metric Entropy conditions}
\begin{lemma}
\label{lem:entropy-H-correct-global}
Let $\mathcal F_n=\mathcal F_{n,\Omega}\cap\mathcal F_{n,\mathrm{tree}}$ be the full sieve from
\Cref{sec:revised-sieve}.
Let $C^\star=C^\star(R_{\Omega,n},\overline R_{\Omega,n})$ be the constant from
\Cref{lem:block-metric-correct}, so that for all $\Theta_1,\Theta_2\in\mathcal F_n$,
\[
H(\Theta_1,\Theta_2)\le
\sqrt{C^\star}\Big(\|\Omega_1-\Omega_2\|_F+\|\eta_1-\eta_2\|_{2,n}\Big),
\quad
C^\star\le n^{2\kappa}
\]
for some $\kappa>0$. Define
\(
r_n:={\varepsilon_n^\dagger}/{(4\sqrt{C^\star})}.
\)
Under \suppref{\ref{assum:multi-growth}}{(A3)},
\[
\log N\bigl(\varepsilon_n^\dagger,\mathcal F_n,H\bigr)
\le C_4\,n(\varepsilon_n^\dagger)^2
\]
for some constant $C_4<\infty$ independent of $n$.
\end{lemma}

\begin{proof}
Because our revised sieve $\mathcal F_n$ explicitly bounds the global $\ell_2$ noise energies, we construct the $\varepsilon_n^\dagger$-net directly on $\mathcal F_n$. We combine an $r_n$-net for $\Omega$ in $\|\cdot\|_F$ and an $r_n$-net for $\eta$ in $\|\cdot\|_{2,n}$.

By \Cref{lem:block-metric-correct}, the product net $\mathcal N:=\mathcal N_\Omega\times\mathcal N_\eta$ satisfies
$H(\Theta,(\tilde\Omega,\tilde\eta)) \le \sqrt{C^\star}(r_n+r_n) = \varepsilon_n^\dagger/2$.
Hence $\mathcal N$ is an $(\varepsilon_n^\dagger/2)$-net for $H$, which implies
\[
\log N(\varepsilon_n^\dagger,\mathcal F_n,H)
\ \le\ \log|\mathcal N_\Omega|+\log|\mathcal N_\eta|.
\]

\paragraph{Covering the precision block.}
Fix $\Omega\in \mathcal F_{n,\Omega}$ and define its effective edge set $\widehat S_\Omega := \{(k,k'): k<k',\ |\omega_{kk'}|>t_{\Omega,n}\}$.
Write $\Omega=\Omega_{\textrm{sig}}+\Omega_{\textrm{noise}}$, where $\Omega_{\textrm{sig}}$ contains the diagonal and off-diagonals in $\widehat S_\Omega$, and $\Omega_{\textrm{noise}}$ contains the remaining elements.

By the sieve definition, the Frobenius norm of the noise is directly controlled:
\[
\|\Omega_{\textrm{noise}}\|_F^2 = E_{\Omega,\mathrm{noise}}(\Omega) \le \Delta_{\Omega,n} = \frac{r_n^2}{16} \implies \|\Omega_{\textrm{noise}}\|_F \le \frac{r_n}{4}.
\]
Thus, we can approximate any $\Omega \in \mathcal F_{n,\Omega}$ by a matrix $\widetilde\Omega = \widetilde\Omega_{\textrm{sig}} + \mathbf{0}$, where $\widetilde\Omega_{\textrm{sig}}$ is chosen from a net over the signal space. The triangle inequality gives $\|\Omega - \widetilde\Omega\|_F \le \|\Omega_{\textrm{sig}} - \widetilde\Omega_{\textrm{sig}}\|_F + \|\Omega_{\textrm{noise}}\|_F$. If we cover $\Omega_{\textrm{sig}}$ at radius $r_n/2$, the total distance is bounded by $3r_n/4 \le r_n$.

The signal matrices for a fixed support $S$ with $|S|\le s_{\Omega,n}$ live in a subspace of dimension $d_S=q+|S|$ bounded by $R_{\Omega,n}$. The volumetric bound over all possible supports yields:
\[
\log N\big(r_n,\ \mathcal F_{n,\Omega},\ \|\cdot\|_F\big)
\ \le\ s_{\Omega,n}\log\Big(\frac{eQ}{s_{\Omega,n}}\Big)
+ (q+s_{\Omega,n})\log\Big(\frac{6R_{\Omega,n}}{r_n/2}\Big).
\]
Because $R_{\Omega,n}$ and $1/r_n$ are polynomially bounded in $n$, $\log(12R_{\Omega,n}/r_n) \lesssim \log n + \log q$. With our sieve choice $s_{\Omega,n} \asymp n\varepsilon_n^2/\log(eQ)$, and the growth condition $q\log n \lesssim n(\varepsilon_n^\dagger)^2$ in \suppref{\Cref{assum:multi-growth}}{(A3)}, we obtain:
\begin{equation}\label{eq:Omega-entropy-final}
\log N\big(r_n,\ \mathcal F_{n,\Omega},\ \|\cdot\|_F\big)
\ \le\ C\,n(\varepsilon_n^\dagger)^2.
\end{equation}

\paragraph{Covering the BART mean block.}
Let $\mathcal G_n$ be the tree-induced mean vectors in $\mathcal F_{n,\mathrm{tree}}$. Fix $\eta\in\mathcal G_n$ and define the global effective-support set $\widehat S_B = \{(j,r): \max_{t,\ell}|b_{jrt\ell}|>u_{B,n}\}$.
Decompose $\eta=\eta_{\textrm{sig}}+\eta_{\textrm{noise}}$, where $\eta_{\textrm{sig}}$ utilizes pairs in $\widehat S_B$.

\smallskip
\noindent\emph{Noise part bounds:}
For the noise component, we bound the empirical $L_2$ norm using the Cauchy-Schwarz inequality. For each observation $i$:
\[
\|\eta_{\textrm{noise}}(x_i)\|_2^2 = \sum_{r=1}^q \Big( \sum_{j \notin S_{B,r}} z_{ij} B_{jr,\mathrm{noise}}(x_i) \Big)^2 \le \sum_{r=1}^q \Big( \sum_{j \notin S_{B,r}} z_{ij}^2 \Big) \Big( \sum_{j \notin S_{B,r}} B_{jr,\mathrm{noise}}^2(x_i) \Big).
\]
By \suppref{\Cref{assum:multi-design}}{(A2)}, $|z_{ij}| \le D$, so $\sum_j z_{ij}^2 \le pD^2$. Furthermore, $B_{jr}(x_i)^2 = (\sum_t g_{jrt}(x_i))^2 \le M \sum_t g_{jrt}^2(x_i) \le M \sum_{t,\ell} b_{jrt\ell}^2$. Therefore:
\[
\|\eta_{\textrm{noise}}(x_i)\|_2^2 \le p D^2 M \sum_{(j,r) \notin \widehat S_B} \sum_{t=1}^M \sum_{\ell=1}^{L_{jrt}} b_{jrt\ell}^2 = p D^2 M E_{B,\mathrm{noise}}(\mathbf{b}).
\]
By the sieve constraint, $E_{B,\mathrm{noise}}(\mathbf{b}) \le \Delta_{B,n} = \frac{r_n^2}{16 D^2 p^2 q M}$. Hence:
\[
\|\eta_{\textrm{noise}}(x_i)\|_2^2 \le p D^2 M \frac{r_n^2}{16 D^2 p^2 q M} \le \frac{r_n^2}{16}.
\]
Averaging over $n$ yields $\|\eta_{\textrm{noise}}\|_{2,n} \le r_n/4$. Thus, by covering only $\eta_{\textrm{sig}}$ at radius $r_n/2$ and setting the noise coordinates to zero, we achieve an overall $r_n$-net for $\eta$.

\smallskip
\noindent\emph{Support-selection cost:}
By construction $|\widehat S_B|\le S_{B,n}$. The number of possible global supports is:
\begin{equation}\label{eq:support-cost-global}
\log\Big(\#\{\widehat S_B\}\Big)
\ \le\ S_{B,n}\log\Big(\frac{eQ_B}{S_{B,n}}\Big)
\ \lesssim\ n\varepsilon_n^2+\log q
\ \lesssim\ n(\varepsilon_n^\dagger)^2.
\end{equation}

\smallskip
\noindent\emph{Global leaf allocation and tree topology:}
Let $s = |\widehat S_B| \le S_{B,n}$. The total leaves in the active set is bounded by $L_{\mathrm{act},n}$. The number of ways to allocate these leaves among $s$ active functions is at most $2^{L_{\mathrm{act},n} + S_{B,n}}$.
For a fixed leaf allocation, the number of distinct partitions realizable by a binary tree with $\ell_{jr}$ leaves choosing from $d_n$ coordinates is at most $(c\,n\,d_n)^{\ell_{jr}}$. The joint topological complexity is bounded by $(cnd_n)^{L_{\mathrm{act},n}}$.

\smallskip
\noindent\emph{Leaf gridding and $\eta$-cover:}
We discretize the leaf values onto a grid with mesh $\delta_{\textrm{leaf}}:= a_n /M$, where $a_n:= {r_n}/{(4D\sqrt{pS_{B,n}})}.$
Because each leaf lies in $[-\Lambda_{B,n},\Lambda_{B,n}]$, the number of grid points per leaf is at most $C(\Lambda_{B,n} M/a_n)$.
Gridding all $\le L_{\mathrm{act},n}$ active leaves yields at most $(C\Lambda_{B,n} M/a_n)^{L_{\mathrm{act},n}}$ grid options.
Constructing $\widetilde\eta_{\textrm{sig}}$ using the gridded functions guarantees $\|\eta_{\textrm{sig}}-\widetilde\eta_{\textrm{sig}}\|_{2,n} \le r_n/4$. Summing the selection, allocation, topological, and gridding costs gives:
\begin{align}\label{eq:eta-entropy-global}
\begin{split}
\log N\big(r_n,\ \mathcal G_n,\ \|\cdot\|_{2,n}\big)
\ &\le\ S_{B,n}\log\Big(\frac{eQ_B}{S_{B,n}}\Big)
\ +\ (L_{\mathrm{act},n} + S_{B,n})\log 2 \\
&\quad +\ L_{\mathrm{act},n}\log(c\,n\,d_n)
\ +\ L_{\mathrm{act},n}\log\Big(\frac{C\,\Lambda_{B,n}\,M}{a_n}\Big).
\end{split}
\end{align}
Because $\Lambda_{B,n}, d_n$, and $1/a_n$ are bounded by polynomials in $n$, the logarithmic multipliers for $L_{\mathrm{act},n}$ are bounded by $C' \log n$. Substituting $L_{\mathrm{act},n} = C_{L,\mathrm{act}} \big( n\varepsilon_n^2 / \log n + \log Q_B \big)$, the topological and gridding costs simplify to $\mathcal{O}(n\varepsilon_n^2 + \log Q_B \log n)$.
By the growth condition in \suppref{\Cref{assum:multi-growth}}{(A3)}, the right-hand side of \eqref{eq:eta-entropy-global} is bounded by $C n(\varepsilon_n^\dagger)^2$.

Combining \eqref{eq:Omega-entropy-final} with the $\eta$-entropy bound yields:
\[
\log N\bigl(\varepsilon_n^\dagger,\ \mathcal F_n,\ H\bigr)
\ \le\ \log N\big(r_n,\ \mathcal F_{n,\Omega},\ \|\cdot\|_F\big)
+ \log N\big(r_n,\ \mathcal G_n,\ \|\cdot\|_{2,n}\big)
\ \le\ C_4\,n(\varepsilon_n^\dagger)^2.
\]
This completes the proof.
\end{proof}

\subsection{Proof of \Cref{prop:postconc_sieve}}
\label{sec:proofpostconc_sieve}

\begin{proof}
We evaluate the posterior probability of the complement of the Hellinger ball around the true parameter $\Theta_0$ under the truncated prior $\Pi_n^{\mathcal F}$. Let $U_n = \{\Theta : H(\Theta, \Theta_0) > M \varepsilon_n^\dagger\}$. By Bayes' theorem for the truncated prior, the posterior probability is:
\begin{equation}\label{eq:bayes-ratio-truncated}
\Pi_n^{\mathcal F}(U_n \mid \bY) = \frac{\int_{U_n \cap \mathcal F_n} \frac{p_\Theta(\bY)}{p_{\Theta_0}(\bY)} d\Pi(\Theta)}{\int_{\mathcal F_n} \frac{p_\Theta(\bY)}{p_{\Theta_0}(\bY)} d\Pi(\Theta)} =: \frac{N_n^{\mathcal F}}{D_n^{\mathcal F}}.
\end{equation}

\paragraph{Lower Bounding the Truncated Denominator $D_n^{\mathcal F}$.}
Let $B_n := \{\Theta : K(\Theta_0, \Theta) \le n(\varepsilon_n^\dagger)^2, V(\Theta_0, \Theta) \le n(\varepsilon_n^\dagger)^2 \}$ be the Kullback-Leibler neighborhood. In \Cref{subsec:KL-corrected}, we established the parameter neighborhood $\mathcal A_n \subseteq B_n$. 

By construction, we have $\mathcal A_n \subseteq \mathcal{F}_n$. Thus, the truncated denominator is lower bounded by the integral over $\mathcal A_n$:
\[
D_n^{\mathcal F} \ge \int_{\mathcal A_n} \frac{p_\Theta(\bY)}{p_{\Theta_0}(\bY)} d\Pi(\Theta).
\]
By the Lemmas \ref{lem:ghs-smallball-continuous} and \ref{lem:BART-smallball-continuous-global}, the original prior mass of this neighborhood is bounded below by $\Pi(\mathcal A_n) \ge \exp\big(- C_{\mathrm{prior}} n (\varepsilon_n^\dagger)^2 \big)$. 
Applying Lemma 8.1 of \citet{Ghosal2000}, for any $C > 0$, the event
\[
\mathcal E_n = \left\{ \int_{\mathcal A_n} \frac{p_\Theta(\bY)}{p_{\Theta_0}(\bY)} d\Pi(\Theta) \ge \Pi(\mathcal A_n)\exp\big(- C n (\varepsilon_n^\dagger)^2 \big) \right\}
\]
holds with $P_{\Theta_0}$-probability tending to 1. Thus, on the event $\mathcal E_n$, the truncated denominator satisfies $D_n^{\mathcal F} \ge \exp(- C_D n (\varepsilon_n^\dagger)^2)$ for $C_D = C_{\mathrm{prior}} + C$.

\paragraph{Upper Bounding the Truncated Numerator $N_n^{\mathcal F}$.}
To evaluate the numerator strictly over the sieve, we utilize the existence of a global test function $\phi_n$ from \citet[Section 7]{Ghosal2000}, which separates $\Theta_0$ from $\{\Theta \in \mathcal F_n : H(\Theta,\Theta_0) > M\varepsilon_n^\dagger\}$. We bound the expected numerator as follows:
\begin{align*}
& \mathbb E_{\Theta_0} \left[ \phi_n \int_{U_n \cap \mathcal F_n} \frac{p_\Theta(\bY)}{p_{\Theta_0}(\bY)} d\Pi(\Theta) + (1-\phi_n) \int_{U_n \cap \mathcal F_n} \frac{p_\Theta(\bY)}{p_{\Theta_0}(\bY)} d\Pi(\Theta) \right] \\
&\quad \le \mathbb E_{\Theta_0}[\phi_n] + \int_{U_n \cap \mathcal F_n} \mathbb E_{\Theta}[1-\phi_n] d\Pi(\Theta) \\
&\quad \le \exp\big(- c_1 M^2 n(\varepsilon_n^\dagger)^2 \big) + \sup_{\Theta \in U_n \cap \mathcal F_n} \mathbb E_{\Theta}[1-\phi_n] \cdot \Pi(\mathcal F_n) \\
&\quad \le \exp\big(- c_1 M^2 n(\varepsilon_n^\dagger)^2 \big) + \exp\big(- c_2 M^2 n(\varepsilon_n^\dagger)^2 \big) \cdot 1 \\
&\quad \le 2 \exp\big(- c_3 M^2 n(\varepsilon_n^\dagger)^2 \big),
\end{align*}
where $c_3 = \min(c_1, c_2)$ and we trivially bounded the unconditional prior mass $\Pi(\mathcal F_n) \le 1$. By Markov's inequality, the $P_{\Theta_0}$-probability that the truncated numerator $N_n^{\mathcal F}$ exceeds $\exp(- ({c_3}/{2}) M^2 n(\varepsilon_n^\dagger)^2)$ is bounded by $2 \exp(- ({c_3}/{2}) M^2 n(\varepsilon_n^\dagger)^2)$, which tends to 0.

By choosing the testing constant $M$ sufficiently large such that the test power $c_3 M^2 / 2 > C_D$, the numerator $N_n^{\mathcal F}$ strictly dominates the exponential lower bound of the truncated denominator $D_n^{\mathcal F}$ with $P_{\Theta_0}$-probability tending to 1. 

Consequently, the ratio $\Pi_n^{\mathcal F}(U_n \mid \bY) = N_n^{\mathcal F} / D_n^{\mathcal F} \to 0$ in $P_{\Theta_0}$-probability as $n \to \infty$. This establishes the Hellinger contraction rate $\varepsilon_n^\dagger$ for the posterior distribution under the truncated prior on the effective sieve $\mathcal F_n$.
\end{proof}

To establish testing power, we must prove that as the noise indices ($k,m$) grow, the Hellinger distance from the true parameter $\Theta_0$ strictly increases. Because $\Theta_0$ is fixed, let $C_0 := \max(\|\Omega_0\|_F, \|\eta_0\|_{2,n}) < \infty$.

\begin{lemma}[Amplitude Shell Separation]
\label{lem:shell-separation}
Under the model assumptions \suppref{\ref{assum:multi-smoothness}}{(A1)}--\suppref{\ref{assum:beta-min}}{(A6)}, there exists a universal constant $c > 0$ such that for any shell $S_{k,m,a} \subset \mathcal G_n$,
\[
\inf_{\Theta \in S_{k,m,a}} H^2(\Theta,\Theta_0) \ge c \min\left(1, \ \max(k,m) (\varepsilon_n^\dagger)^2 + (a - C_0)_+^2 \right),
\]
where $(x)_+ = \max(x, 0)$.
\end{lemma}

\begin{proof}
As established in \Cref{lem:block-metric-correct}, there exists a universal constant $c_{\mathrm{loc}} > 0$ depending only on the spectrum of $\Omega_0$ such that:
\begin{equation} \label{eq:metric_equiv_exact}
H^2(\Theta,\Theta_0) \ge c_{\mathrm{loc}} \min\Big(1, \ \|\Omega-\Omega_0\|_F^2 + \|\eta-\eta_0\|_{2,n}^2 \Big).
\end{equation}

We bound the Euclidean parameter distance using the shell indices. By orthogonality of the support sets, we can decompose the squared errors into active and inactive (noise) components:
\[
\|\Omega-\Omega_0\|_F^2 \ge \|\Omega_{\mathrm{sig}} - \Omega_{0,\mathrm{sig}}\|_F^2 + \|\Omega_{\mathrm{noise}} - \Omega_{0,\mathrm{noise}}\|_F^2.
\]

\textbf{Noise Separation ($k, m$):}
For the precision matrix, we evaluate the error over the thresholded inactive set $S_{\Omega, \mathrm{noise}} := \widehat S_\Omega(t_{\Omega,n})^c$.
For true inactive edges, $(\omega_{kk'} - \omega^0_{kk'})^2 = \omega_{kk'}^2$. For missed active edges ($(k,k') \in S_{0,\Omega} \cap S_{\Omega, \mathrm{noise}}$), the parameter is bounded by $|\omega_{kk'}| \le t_{\Omega,n}$. By the $\beta$-min condition \suppref{\ref{assum:beta-min}}{(A6)}, $|\omega^0_{kk'}| \ge 4t_{\Omega,n}$. By the reverse triangle inequality:
\[
|\omega_{kk'} - \omega^0_{kk'}| \ge |\omega^0_{kk'}| - |\omega_{kk'}| \ge 3t_{\Omega,n} \ge |\omega_{kk'}| \implies (\omega_{kk'} - \omega^0_{kk'})^2 \ge \omega_{kk'}^2.
\]
Summing over the noise set yields $\|\Omega_{\mathrm{noise}} - \Omega_{0,\mathrm{noise}}\|_F^2 \ge E_{\Omega,\mathrm{noise}}(\Omega) > k \Delta_{\Omega,n}$.

By an identical application of the $\beta$-min condition and the functional RE condition \suppref{\ref{assum:multi-RE}}{(A5)} for the mean block, missed true signals are penalized strictly more than the noise threshold $u_{B,n}$, yielding:
\[
\|\eta_{\mathrm{noise}} - \eta_{0,\mathrm{noise}}\|_{2,n}^2 \ge c_Z' E_{\eta,\mathrm{noise}}(\bm{b}) > m c_Z' \Delta_{\eta,n}.
\]
Because $\Delta_{\Omega,n} \asymp \Delta_{\eta,n} \asymp (\varepsilon_n^\dagger)^2$, the noise contribution is bounded below by $C_1 \max(k,m) (\varepsilon_n^\dagger)^2$.

\textbf{Amplitude Separation ($a$):}
If the active amplitude index is $a$, then $\max(\|\Omega_{\mathrm{sig}}\|_F, \|\eta_{\mathrm{sig}}\|_{2,n}) \ge a$. 
By the standard triangle inequality, $\|\Omega_{\mathrm{sig}} - \Omega_{0,\mathrm{sig}}\|_F \ge \|\Omega_{\mathrm{sig}}\|_F - \|\Omega_{0,\mathrm{sig}}\|_F \ge \|\Omega_{\mathrm{sig}}\|_F - C_0$.
Applying this to both blocks, the active error is bounded below by:
\[
\|\Omega_{\mathrm{sig}} - \Omega_{0,\mathrm{sig}}\|_F^2 + \|\eta_{\mathrm{sig}} - \eta_{0,\mathrm{sig}}\|_{2,n}^2 \ge (a - C_0)_+^2.
\]

Summing the noise and active amplitude lower bounds, the total Euclidean parameter distance is bounded below by the maximum of the two independent forces:
\[
\|\Omega-\Omega_0\|_F^2 + \|\eta-\eta_0\|_{2,n}^2 \ge \frac{1}{2} \Big( C_1 \max(k,m) (\varepsilon_n^\dagger)^2 + (a - C_0)_+^2 \Big).
\]
Substituting this into the global-local metric equivalence \eqref{eq:metric_equiv_exact} yields the result.
\end{proof}
\begin{lemma}[Shell Metric Entropy]
\label{lem:shell-entropy}
For any shell $S_{k,m} \subset \mathcal{G}_n \setminus \mathcal{F}_n$, let $\ell = \max(k,m) \ge 1$, 
\[
\log N\big(\sqrt{\ell} \varepsilon_n^\dagger, S_{k,m}, H\big) \le C_{\mathrm{ent}} \ell n(\varepsilon_n^\dagger)^2,
\]
for a universal constant $C_{\mathrm{ent}} > 0$.
\end{lemma}

\begin{proof}
To establish uniformly consistent tests for the shell $S_{k,m}$, which is separated from $\Theta_0$ by a distance proportional to $\sqrt{\ell}\varepsilon_n^\dagger$ (\Cref{lem:shell-separation}), we must cover the shell at a testing radius $\rho_\ell \asymp \sqrt{\ell}\varepsilon_n^\dagger$. By the global-local metric equivalence, this corresponds to covering the Euclidean parameter space at a radius $r_\ell = \sqrt{\ell} r_n$, where $r_n \asymp \varepsilon_n^\dagger$ is the base sieve tolerance radius.

For any $\Theta \in S_{k,m}$, we decompose the parameters into active signals and thresholded inactive noise: $\Omega = \Omega_{\mathrm{sig}} + \Omega_{\mathrm{noise}}$ and $\bm{b} = \bm{b}_{\mathrm{sig}} + \bm{b}_{\mathrm{noise}}$. 

By definition, the noise energies are bounded by $E_{\Omega,\mathrm{noise}}(\Omega) \le (k+1)\Delta_{\Omega,n}$ and $E_{\eta,\mathrm{noise}}(\bm{b}) \le (m+1)\Delta_{\eta,n}$. Due to the baseline inner sieve tolerances $\Delta \asymp r_n^2 / 16$, the maximum amplitude of the noise in shell $S_{k,m}$ is strictly bounded:
\[
\max\Big( \|\Omega_{\mathrm{noise}}\|_F, \, \|\eta_{\mathrm{noise}}\|_{2,n} \Big) \le \frac{\sqrt{\ell+1}}{4} r_n.
\]
Because $\ell \ge 1$, we have the strict geometric inequality ${\sqrt{\ell+1}}/{4} \le {\sqrt{\ell}}/{2}$. This guarantees that the entirety of the noise error strictly consumes less than half of our available testing radius $r_\ell$. We can therefore approximate any $\Theta \in S_{k,m}$ by a pseudo-sparse parameter $\widetilde{\Theta} = (\widetilde{\Omega}_{\mathrm{sig}} + \mathbf{0}, \, \widetilde{\bm{b}}_{\mathrm{sig}} + \mathbf{0})$, forcing the noise dimensions to exactly zero. 

By the triangle inequality, the remaining error budget for covering the active signals is:
\[
\|\Omega_{\mathrm{sig}} - \widetilde{\Omega}_{\mathrm{sig}}\|_F \le \frac{\sqrt{\ell}}{2} r_n, \quad \|\eta_{\mathrm{sig}} - \widetilde{\eta}_{\mathrm{sig}}\|_{2,n} \le \frac{\sqrt{\ell}}{2} r_n.
\]
Because $\Theta \in \mathcal{G}_n$, the active signals reside in strictly bounded combinatorial subspaces: at most $s_{\Omega,n}$ precision edges, $S_{B,n}$ active functions, and $L_{\mathrm{act},n}$ leaves. Furthermore, their maximum amplitudes are bounded by the polynomial envelopes $R_{\Omega,n}$ and $\Lambda_{B,n}$. The volumetric covering numbers for these sparse active dimensions evaluate to:
\[
\log N_{\Omega} \le s_{\Omega,n} \log\left( \frac{3 R_{\Omega,n}}{\sqrt{\ell} r_n / 2} \right), \quad \log N_{\eta} \le C L_{\mathrm{act},n} \log\left( \frac{C' \Lambda_{B,n}}{\sqrt{\ell} r_n / 2} \right).
\]
Because $R_{\Omega,n}, \Lambda_{B,n}$, and $1/r_n$ are polynomial in $n$, the logarithmic factors are bounded by $\mathcal{O}(\log n)$. Because $\ell \ge 1$, dividing by $\sqrt{\ell}$ only decreases the required entropy. Multiplying by the structural budgets $s_{\Omega,n}$ and $L_{\mathrm{act},n}$ yields a total metric entropy bounded by $C n(\varepsilon_n^\dagger)^2 \le C_{\mathrm{ent}} \ell n(\varepsilon_n^\dagger)^2$. 
\end{proof}

\subsubsection*{Proof of \Cref{thm:postconc_full} (Full Posterior Contraction)}\label{sec:proofthm2}

\begin{proof}
By \Cref{lem:shell-separation} and \Cref{lem:shell-entropy}, each shell $S_{k,m}$ is separated from $\Theta_0$ by at least $m_0 \sqrt{\ell} \varepsilon_n^\dagger$ and has metric entropy bounded by $C_{\mathrm{ent}} \ell n(\varepsilon_n^\dagger)^2$. By the standard Birg\'{e}-Le Cam testing theory for convex models \citep{Ghosal2000}, because the testing power grows quadratically with the separation while the entropy grows only linearly with $\ell$, this implies the existence of tests $\phi_{k,m}$ such that:
\begin{align*}
    \mathbb{E}_{\Theta_0} [\phi_{k,m}] & \le \exp(-C_T \ell n (\varepsilon_n ^{\dagger})^2) \\
    \sup_{\Theta \in S_{k,m}} \mathbb{E}_{\Theta} [1-\phi_{k,m}] & \le \exp(-C_T \ell n (\varepsilon_n ^{\dagger})^2), \ \text{where} \ \ell = \max(k,m).
\end{align*}
The testing constant $C_T$ can be made arbitrarily large by choosing a sufficiently large threshold $M$ in the definition of the target Hellinger ball $U_n = \{\Theta : H(\Theta,\Theta_0) > M \varepsilon_n ^\dagger \}$. 

We decompose the full posterior probability of the complement $U_n$ over our nested parameter partitions:
\[
\Pi(U_n \mid \bm{Y}) \le \underbrace{\Pi(U_n \cap \mathcal{F}_n \mid \bm{Y})}_{\text{shown already}\ \le \exp(-Cn (\varepsilon_n^{\dagger})^2) } + \Pi\Big(\bigcup_{(k,m) \neq (0,0)} S_{k,m} \Bigm| \bm{Y}\Big) + \underbrace{\Pi(\mathcal{G}_n ^{c} \mid \bm{Y})}_{\text{shown already}\ \le \exp(-C'n(\varepsilon_n ^{\dagger})^2)}.
\]
The inner and the outer tail bound follow from the proof of \Cref{lem:prior-tail-correct}. It remains to bound the expected posterior mass of the middle-band shells.

Let $D_n$ denote the denominator of the posterior, $\int p_\Theta(\bm{Y})/p_{\Theta_0}(\bm{Y}) d\Pi(\Theta)$. From our small-ball probability lemmas over the KL neighborhood $\mathcal{A}_n$, we established that with $P_{\Theta_0}$-probability tending to 1, the denominator is bounded below by $\exp(-C_D n(\varepsilon_n^\dagger)^2)$. 
Conditioning on this event, we bound the numerator for the entire middle band using the constructed tests:
\begin{align*}
\mathbb E_{\Theta_0} \Bigg[ \sum_{\max(k,m) \ge 1} (1-\phi_{k,m}) \int_{S_{k,m}} \frac{p_\Theta(\bm{Y})}{p_{\Theta_0}(\bm{Y})} d\Pi(\Theta) \Bigg] 
&\stackrel{}{=} \sum_{\max(k,m) \ge 1} \int_{S_{k,m}} \mathbb E_{\Theta_0} \left[ (1-\phi_{k,m}) \frac{p_\Theta(\bm{Y})}{p_{\Theta_0}(\bm{Y})} \right] d\Pi(\Theta) \\
&\stackrel{}{=} \sum_{\max(k,m) \ge 1} \int_{S_{k,m}} \mathbb E_{\Theta}[1-\phi_{k,m}] d\Pi(\Theta) \\
&\stackrel{}{\le} \sum_{\max(k,m) \ge 1} \left( \sup_{\Theta \in S_{k,m}} \mathbb E_{\Theta}[1-\phi_{k,m}] \right) \Pi(S_{k,m}) \\
&\stackrel{}{=} \sum_{\ell=1}^\infty \sum_{\max(k,m)=\ell} \left( \sup_{\Theta \in S_{k,m}} \mathbb E_{\Theta}[1-\phi_{k,m}] \right) \Pi(S_{k,m}) \\
&\stackrel{}{\le} \sum_{\ell=1}^\infty (2 \ell+1) \exp\big(- C_T \ell n (\varepsilon_n^\dagger)^2 \big).
\end{align*}
The factor $(2\ell+1)$ arises because there are exactly $2\ell+1$ distinct pairs of indices $(k,m)$ such that $\max(k,m) = \ell$. 
For sufficiently large $C_T$, this infinite geometric series converges and is strictly dominated by its leading term:
\[
\sum_{\ell=1}^\infty (2 \ell+1) \exp\big(- C_T \ell n (\varepsilon_n^\dagger)^2 \big) = \mathcal{O}\Big(\exp(-C_T n (\varepsilon_n ^{\dagger})^2)\Big).
\]
By choosing the testing constant $M$ such that $C_T > C_D$, the exponentially decaying numerator strictly dominates the KL denominator bound. By Markov's inequality, the posterior mass of the middle band converges to 0 in $P_{\Theta_0}$-probability. Combining all three components confirms that the full posterior distribution concentrates optimally at the rate $\varepsilon_n^\dagger$. 
\end{proof}
\subsection{Proof of Corollary 1}
\label{sec:prooftransformed-surface}
\begin{proof}
By \Cref{thm:postconc_full} and \Cref{lem:block-metric-correct}, the posterior $\Pi$ concentrates on a set
on which both components are small: there exists a constant $c>0$ such that
\begin{equation}\label{eq:En-components}
\Pi\Big(\ \|\Omega-\Omega_0\|_F>c\,\varepsilon_n^\dagger\ \text{ or }\
\|\eta-\eta_0\|_{2,n}>c\,\varepsilon_n^\dagger\ \Bigm|\ \bY\Big)
\ \to\ 0,
\end{equation}
in $P_{\Theta_0}$-probability.\footnote{This is the standard implication of metric equivalence on the
(local) neighborhood used in the proof: once $H(\Theta,\Theta_0)\lesssim \varepsilon_n^\dagger\to0$,
both $\|\Omega-\Omega_0\|_F$ and $\|\eta-\eta_0\|_{2,n}$ must be $\mathcal{O}(\varepsilon_n^\dagger)$ on the same
event under the sieve-truncated posterior.}

\paragraph{Contraction of $\bm{B}(\cdot)$.}
Let $\Delta \bm{B}(\bx):=\bm{B}(\bx)-\bm{B}_0(\bx)$ and $\Delta\eta_i:=\eta_i-\eta_{0,i}=\Delta \bm{B}(\bx_i)^\top\bz_i$.
By the functional RE condition \suppref{\ref{assum:multi-RE}}{(A5)},
\[
\|\bm{B}-\bm{B}_0\|_{F,2,n}^2
=\frac1n\sum_{i=1}^n\|\Delta \bm{B}(\bx_i)\|_F^2
\ \le\ 
\kappa_z^{-1}\cdot \frac1n\sum_{i=1}^n\|\Delta \bm{B}(\bx_i)^\top\bz_i\|_2^2
=\kappa_z^{-1}\,\|\eta-\eta_0\|_{2,n}^2.
\]
Hence,
\begin{equation}\label{eq:B-to-eta}
\|\bm{B}-\bm{B}_0\|_{F,2,n}
\ \le\ \kappa_z^{-1/2}\,\|\eta-\eta_0\|_{2,n}.
\end{equation}
Therefore, on the event in \eqref{eq:En-components} we have $\|\eta-\eta_0\|_{2,n}\le c\varepsilon_n^\dagger$, and so
\[
\|\bm{B}-\bm{B}_0\|_{F,2,n}\ \le\ c\,\kappa_z^{-1/2}\,\varepsilon_n^\dagger.
\]
Absorbing constants into $M'$ and using \eqref{eq:En-components} yields the desired contraction for $\bm{B}(\cdot)$ under the full posterior:
\[
\Pi \Big(\ \|\bm{B}-\bm{B}_0\|_{F,2,n} > M'\kappa_z^{-1/2}\varepsilon_n^\dagger \ \Bigm|\ \bY \Big)\ \to\ 0,
\]
in $P_{\Theta_0}$-probability.

\paragraph{Contraction of $G(\cdot)$.}
Fix $(\bm{B},\Omega)$ satisfying $\|\Omega-\Omega_0\|_F\le c\,\varepsilon_n^\dagger$.
Since $\varepsilon_n^\dagger\to0$, for all sufficiently large $n$ we have
$\|\Omega-\Omega_0\|_F\le \underline\lambda/4$ on this event. By Weyl's inequality and
\Cref{lem:pd-correct}, this implies
\begin{equation}\label{eq:inv-bds-cor}
\|\Omega^{-1}\|_{\text{op}}\le \frac{2}{\underline\lambda},
\qquad
\|\Omega^{-1}-\Omega_0^{-1}\|_{\text{op}}
\le \|\Omega^{-1}-\Omega_0^{-1}\|_F
\le \frac{2}{\underline\lambda^2}\,\|\Omega-\Omega_0\|_F.
\end{equation}

For each $i$,
\[
G(\bx_i)-G_0(\bx_i)
=\Omega^{-1}\big(\bm{B}(\bx_i)-\bm{B}_0 (\bx_i)\big)^\top
+\big(\Omega^{-1}-\Omega_0^{-1}\big)\bm{B}_0 (\bx_i)^\top.
\]
Using $\|AB\|_F\le \|A\|_{\text{op}}\|B\|_F$, the triangle inequality, and then averaging over $i$ gives
\begin{align*}
\|G-G_0\|_{F,2,n}
&\le
\|\Omega^{-1}\|_{\text{op}}\,\|\bm{B}-\bm{B}_0\|_{F,2,n}
+\|\Omega^{-1}-\Omega_0^{-1}\|_{\text{op}}
\Big(\frac1n\sum_{i=1}^n \|\bm{B}_0 (\bx_i)\|_F^2\Big)^{1/2}\\
&\le
\|\Omega^{-1}\|_{\text{op}}\,\|\bm{B}-\bm{B}_0\|_{F,2,n}
+\|\Omega^{-1}-\Omega_0^{-1}\|_{\text{op}}\,R_{B_0}.
\end{align*}
Now invoke \eqref{eq:B-to-eta} and combine with \eqref{eq:inv-bds-cor} to obtain
\[
\|G-G_0\|_{F,2,n}
\ \le\
\frac{2}{\underline\lambda}\,\kappa_z^{-1/2}\,\|\eta-\eta_0\|_{2,n}
+\frac{2}{\underline\lambda^2}\,R_{B_0}\,\|\Omega-\Omega_0\|_F.
\]
On the event in \eqref{eq:En-components} both $\|\eta-\eta_0\|_{2,n}$ and $\|\Omega-\Omega_0\|_F$ are at most
$c\,\varepsilon_n^\dagger$, hence
\[
\|G-G_0\|_{F,2,n}
\ \le\
C\big(\kappa_z^{-1/2}+R_{B_0}\big)\,\varepsilon_n^\dagger
\]
for a constant $C$ depending only on $\underline\lambda$. Absorbing constants into $M'$ and using
\eqref{eq:En-components} completes the proof.
\end{proof}

\setcounter{figure}{0}
\setcounter{equation}{0}
\setcounter{table}{0}
\setcounter{theorem}{0}
\setcounter{assumption}{0}

\section{Additional algorithmic details}
\label{sec:multi-algorithm}
This section collects the posterior computation details omitted from Section 3. Throughout, we suppress MCMC iteration superscripts for readability. 
\subsection{Outcome-wise pseudo-response representation}
The key computational simplification is that, conditional on the remaining outcomes and on \(\Omega\), the multivariate Gaussian likelihood reduces to a scalar Gaussian working likelihood.

Fix \(r\in\{1,\ldots,q\}\). Writing \(\bm E_i=(E_{i1},\ldots,E_{iq})^\top\), we  expand the quadratic form
\begin{align*}
\bm E_i^\top \Omega \bm E_i
&=
\sum_{a=1}^q \sum_{b=1}^q E_{ia}\Omega_{ab}E_{ib} \\
&=
\Omega_{rr}E_{ir}^2
+
2E_{ir}\sum_{k\neq r}\Omega_{rk}E_{ik}
+
\sum_{a\neq r}\sum_{b\neq r}E_{ia}\Omega_{ab}E_{ib}.
\end{align*}
Completing the square yields
\[
\bm E_i^\top \Omega \bm E_i
=
\Omega_{rr}
\left(
E_{ir}
+
\Omega_{rr}^{-1}\sum_{k\neq r}\Omega_{rk}E_{ik}
\right)^2
+
C_{i,-r},
\]
where \(C_{i,-r}\) is constant with respect to \(E_{ir}\). Therefore,
\[
E_{ir}\mid \bm E_{i,-r},\Omega
\sim
\mathcal N\!\left(
-\Omega_{rr}^{-1}\sum_{k\neq r}\Omega_{rk}E_{ik},
\ \Omega_{rr}^{-1}
\right).
\]
Substituting \(E_{ir}=Y_{ir}-\eta_{ir}\) gives the scalar pseudo-response
\begin{equation}\label{eq:multi-pseudoresponse}
\widetilde Y_{ir}
:=
Y_{ir}
+
\Omega_{rr}^{-1}\sum_{k\neq r}\Omega_{rk}E_{ik},
\end{equation}
for which
\begin{equation}\label{eq:multi-working-gaussian}
\widetilde Y_{ir}\mid \bm E_{i,-r},\Omega
\sim
\mathcal N\big(\eta_{ir},\Omega_{rr}^{-1}\big).
\end{equation}
Hence, conditional on \(\Omega\), each outcome-specific update reduces exactly to a univariate Gaussian BART regression with working variance \(\Omega_{rr}^{-1}\).

\subsection{Conditional tree updates for \texorpdfstring{$B_{jr}$}{Bjr}}

Fix an outcome \(r\) and predictor \(j\). The \(t\)th tree in the \((j,r)\)th ensemble is updated conditional on all other trees. Define the leave-one-tree-out partial residual
\begin{equation}\label{eq:multi-partial-resid}
\widetilde r_{ir}^{(j,t)}
=
\widetilde Y_{ir}
-
\sum_{j'\neq j} z_{ij'} B_{j'r}(\bm X_i)
-
\sum_{t'\neq t} z_{ij}\,
g_{jrt'}(\bm X_i;\mathcal T_{jrt'},\mathcal M_{jrt'}).
\end{equation}
Conditional on the remaining trees, the working model for tree \((j,r,t)\) is
\begin{equation}\label{eq:multi-tree-working}
\widetilde r_{ir}^{(j,t)}
=
z_{ij}\,\mu_{jrt,\ell(\bm X_i;\mathcal T_{jrt})}
+
\varepsilon_{ir},
\quad
\varepsilon_{ir}\sim \mathcal N(0,\Omega_{rr}^{-1}).
\end{equation}

\subsubsection*{Leaf sufficient statistics}
Let \(I\subset\{1,\ldots,n\}\) denote the set of observations routed to a given leaf. We denote the following quantities
\(n_I = |I|,\) \(A_I = \sum_{i\in I} z_{ij}^2,\) 
\(B_I = \sum_{i\in I} z_{ij}\,\widetilde r_{ir}^{(j,t)},\) and \(R_I = \sum_{i\in I} \bigl(\widetilde r_{ir}^{(j,t)}\bigr)^2.\)

\subsubsection*{Integrated node marginal likelihood}

For a leaf parameter \(\mu\), the likelihood-prior pair is given by
\(
\widetilde r_{ir}^{(j,t)} \mid \mu
\sim
\mathcal N(z_{ij}\mu,\Omega_{rr}^{-1}),
\
\mu\sim \mathcal N(0,\sigma_{jr}^2).
\)
Integrating out \(\mu\) yields the exact marginal likelihood contribution of node \(I\):
\begin{equation}
\label{eq:multi-node-marginal}
\begin{split}
m(I)
&=
\int
\left[
\prod_{i\in I}
\phi\left(\widetilde r_{ir}^{(j,t)}; z_{ij}\mu,\Omega_{rr}^{-1}\right)
\right]
\phi(\mu;0,\sigma_{jr}^2)\, d\mu \\
&=
(2\pi)^{-n_I/2}\,
\Omega_{rr}^{\,n_I/2}\,
\bigl(1+\sigma_{jr}^2\Omega_{rr}A_I\bigr)^{-1/2}
\exp\left\{
-\frac{\Omega_{rr}}{2}R_I
+
\frac{\Omega_{rr}^2\sigma_{jr}^2 B_I^2}
{2(1+\sigma_{jr}^2\Omega_{rr}A_I)}
\right\}.
\end{split}
\end{equation}
For a tree \(\mathcal T\), the integrated marginal likelihood is the product over its terminal nodes:
\(\mathcal L_{\mathrm{marg}}(\mathcal T)
=\prod_{\ell\in\mathcal L(\mathcal T)} m(I_\ell).\)

\subsubsection*{Metropolis--Hastings update for the tree structure}

We update \(\mathcal T_{jrt}\) using standard local grow/prune proposals. If \(\mathcal T^\star\) is the proposed tree, the acceptance probability is
\begin{equation}
\label{eq:multi-tree-mh}
\alpha(\mathcal T,\mathcal T^\star)
=
1\wedge
\frac{p(\mathcal T^\star)}{p(\mathcal T)}
\cdot
\frac{q(\mathcal T^\star\to\mathcal T)}{q(\mathcal T\to\mathcal T^\star)}
\cdot
\frac{\mathcal L_{\mathrm{marg}}(\mathcal T^\star)}
{\mathcal L_{\mathrm{marg}}(\mathcal T)},
\end{equation}
where \(p(\mathcal T)\) is the BART tree prior and \(q(\cdot\to\cdot)\) is the proposal kernel.

\subsubsection*{Gaussian full conditional for the leaf parameters}

Conditional on the accepted tree structure, the leaf parameters are independent across leaves. For a leaf \(I\),

\begin{equation}
\label{eq:multi-leaf-full}
\mu_I\mid \ldots \sim \mathcal N(m_I^\star,V_I^\star),
\quad
V_I^\star = \left(\Omega_{rr}A_I + \sigma_{jr}^{-2}\right)^{-1},
\quad
m_I^\star = V_I^\star\,\Omega_{rr} B_I.
\end{equation}

\subsection{Updates for the modifier-splitting probabilities}

For each ensemble \((j,r)\), let \(c_{jrk}\) denote the number of internal nodes, across all \(M\) trees in that ensemble, that split on modifier \(X_k\), \(k=1,\ldots,d\). Let \(C_{jr}=\sum_{k=1}^d c_{jrk}.\)
Conditionally on the concentration parameter \(\theta_{jr}\), the splitting-probability vector
\(\bm\pi_{jr}=(\pi_{jr,1},\ldots,\pi_{jr,d})^\top\) has the conjugate full conditional
\begin{equation}
\label{eq:multi-pi-full}
\bm\pi_{jr}\mid \theta_{jr},\mathbf c_{jr}
\sim
\mathrm{Dirichlet}\!\left(
\frac{\theta_{jr}}{d}+c_{jr,1},
\ldots,
\frac{\theta_{jr}}{d}+c_{jr,d}
\right).
\end{equation}
In practice, this is sampled by drawing
\[
g_k\mid \ldots
\sim
\mathrm{Ga}\!\left(\frac{\theta_{jr}}{d}+c_{jr,k},1\right),
\quad
\pi_{jr,k}=\frac{g_k}{\sum_{\ell=1}^d g_\ell}.
\]

\subsubsection*{Metropolis--Hastings update for \texorpdfstring{$\theta_{jr}$}{thetajr}}

Integrating out \(\bm\pi_{jr}\) yields the Dirichlet--Multinomial marginal
\begin{equation}
\label{eq:multi-theta-target}
p(\theta_{jr}\mid \mathbf c_{jr})
\propto
\frac{\Gamma(\theta_{jr})}{\Gamma(\theta_{jr}+C_{jr})}
\prod_{k=1}^d
\frac{\Gamma\left(\theta_{jr}/d+c_{jr,k}\right)}
{\Gamma(\theta_{jr}/d)}
\cdot
\frac{1}{(d+\theta_{jr})^{d+1}} .
\end{equation}
We update \(\theta_{jr}\) by a random-walk Metropolis step on the log scale. Writing \(\vartheta_{jr}=\log\theta_{jr}\), propose
\(
\vartheta_{jr}^\star \sim \mathcal N(\vartheta_{jr},\sigma_\theta^2),
\
\theta_{jr}^\star = e^{\vartheta_{jr}^\star},
\)
and accept with probability
\( \alpha(\theta_{jr},\theta_{jr}^\star)
= 1\wedge \frac{p(\theta_{jr}^\star\mid \mathbf c_{jr})\,\theta_{jr}^\star}
{p(\theta_{jr}\mid \mathbf c_{jr})\,\theta_{jr}}.\)

\subsection{Horseshoe updates for the coefficient-function ensembles}

For each ensemble \((j,r)\), let
\(
K_{jr} := \sum_{t=1}^M |\mathcal L(\mathcal T_{jrt})|
\)
denote the total number of leaf parameters in the ensemble, and let
\(
S_{jr}
:=
\sum_{t=1}^M \sum_{\ell\in\mathcal L(\mathcal T_{jrt})}
\mu_{jrt\ell}^2.
\)
Under the horseshoe prior
\[
\mu_{jrt\ell}\mid \lambda_{jr},\tau_B
\sim
\mathcal N\!\left(0,\frac{\tau_B^2\lambda_{jr}^2}{M}\right),
\quad
\lambda_{jr}\sim \mathcal C^+(0,1),
\quad
\tau_B\sim \mathcal C^+(0,1),
\]
we use the inverse-gamma augmentation of \citet{Makalic_2016}:
\[
\lambda_{jr}^2\mid \nu_{jr}\sim \mathcal{IG}\!\left(\frac12,\frac{1}{\nu_{jr}}\right),
\quad
\nu_{jr}\sim \mathcal{IG}\!\left(\frac12,1\right),
\]
\[
\tau_B^2\mid \xi_B\sim \mathcal{IG}\!\left(\frac12,\frac{1}{\xi_B}\right),
\quad
\xi_B\sim \mathcal{IG}\!\left(\frac12,1\right).
\]

\subsubsection*{Local ensemble-specific scale \texorpdfstring{$\lambda_{jr}^2$}{lambdajr2}}

Conditioning on all leaf parameters in ensemble \((j,r)\) gives
\begin{equation}
\label{eq:multi-lambda-full}
\lambda_{jr}^{2}\mid \ldots
\sim
\mathcal{IG}\left(
\frac{K_{jr}+1}{2},
\ \frac{M S_{jr}}{2\tau_B^2}+\frac{1}{\nu_{jr}}
\right).
\end{equation}
The auxiliary variable is updated from
\begin{equation}
\label{eq:multi-nu-full}
\nu_{jr}\mid \lambda_{jr}^2
\sim
\mathcal{IG}\left(
1,\ 1+\frac{1}{\lambda_{jr}^2}
\right).
\end{equation}

\subsubsection*{Global coefficient-function scale \texorpdfstring{$\tau_B^2$}{tauB2}}

Let
\(
K_{\mathrm{tot}}:=\sum_{j=1}^p\sum_{r=1}^q K_{jr},
\
W_{jr}:=\frac{M S_{jr}}{\lambda_{jr}^2}.
\)
Then
\begin{equation}
\label{eq:multi-tauB-full}
\tau_B^{2}\mid \ldots
\sim
\mathcal{IG}\left(
\frac{K_{\mathrm{tot}}+1}{2},
\ \frac12\sum_{j=1}^p\sum_{r=1}^q W_{jr}+\frac{1}{\xi_B}
\right),
\end{equation}
and \[\xi_B\mid \tau_B^2
\sim
\mathcal{IG}\left(
1, 1+\frac{1}{\tau_B^2}
\right).\]

\subsection{Graphical Horseshoe updates for the precision matrix}\label{sec:GHSdetails}

After updating all coefficient-function ensembles, recompute the residuals
\[
\bm E_i
=
\bm Y_i - B(\bm X_i)^\top \bm Z_i,
\quad i=1,\ldots,n,
\]
and form the residual scatter matrix
\begin{equation}
\label{eq:multi-resid-scatter}
S_E
=
\sum_{i=1}^n \bm E_i \bm E_i^\top .
\end{equation}
Conditional on the current mean functions, the posterior kernel for \(\Omega\) is
\begin{equation}
\label{eq:multi-omega-kernel}
p(\Omega\mid B,\text{data})
\propto
\mathbbm 1\{\Omega\succ0\}\,
|\Omega|^{n/2}
\exp\left\{-\frac12\operatorname{tr}(S_E\Omega)\right\}
\pi(\Omega\mid \Lambda,\tau_\Omega),
\end{equation}
where \(\pi(\Omega\mid \Lambda,\tau_\Omega)\) is the Graphical Horseshoe prior \citep{LiJCGS}:
\[
\omega_{rs}\mid \lambda_{rs},\tau_\Omega
\sim
\mathcal N(0,\tau_\Omega^2\lambda_{rs}^2),
\quad r<s,
\]
with half-Cauchy priors on the local scales \(\lambda_{rs}\) and global scale \(\tau_\Omega\).

\subsubsection*{Column-wise block Gibbs update}

We update \(\Omega\) one column at a time. Fix \(c\in\{1,\ldots,q\}\). Partition
\[
\Omega=
\begin{pmatrix}
\Omega_{-c,-c} & \omega_{-c,c}\\
\omega_{-c,c}^\top & \omega_{cc}
\end{pmatrix},
\quad
S_E=
\begin{pmatrix}
S_{-c,-c} & s_{-c,c}\\
s_{-c,c}^\top & s_{cc}
\end{pmatrix},
\]
and write
\[
\Sigma_{-c,-c} := \Omega_{-c,-c}^{-1},
\quad
\Lambda_c := \mathrm{diag}(\lambda_{1c}^2,\ldots,\lambda_{(c-1)c}^2,\lambda_{(c+1)c}^2,\ldots,\lambda_{qc}^2).
\]
Following \citet{LiJCGS}, reparameterize
\[
\beta := \omega_{-c,c},
\quad
\gamma := \omega_{cc} - \beta^\top \Sigma_{-c,-c}\beta.
\]
Then the full conditionals are
\begin{equation}
\label{eq:multi-beta-full}
\beta\mid \ldots
\sim
\mathcal N\bigl(-C_c\,s_{-c,c},\ C_c\bigr),
\quad
C_c
=
\left(
s_{cc}\Sigma_{-c,-c}
+
(\tau_\Omega^2 \Lambda_c)^{-1}
\right)^{-1},
\end{equation}
and
\begin{equation}
\label{eq:multi-gamma-full}
\gamma\mid \ldots
\sim
\mathrm{Ga}\left(
\frac n2 + a_0, \frac{s_{cc}+2b_0}{2}
\right).
\end{equation}
The updated column is then recovered by
\begin{equation}
\label{eq:multi-omega-reconstruct}
\omega_{-c,c}=\beta,
\quad
\omega_{cc}
=
\gamma + \beta^\top \Sigma_{-c,-c}\beta.
\end{equation}
Because \(\gamma>0\), this update preserves positive definiteness.

\subsubsection*{Local and global Graphical Horseshoe scales}

Using the same \citet{Makalic_2016} augmentation, write
\[
\lambda_{rs}^2\mid \nu_{rs}
\sim
\mathcal{IG}\left(\frac12,\frac{1}{\nu_{rs}}\right),
\quad
\nu_{rs}\sim
\mathcal{IG}\left(\frac12,1\right),
\]
and
\[
\tau_\Omega^2\mid \xi_\Omega
\sim
\mathcal{IG}\!\left(\frac12,\frac{1}{\xi_\Omega}\right),
\quad
\xi_\Omega\sim
\mathcal{IG}\left(\frac12,1\right).
\]
For each off-diagonal element \(\omega_{rs}\), \(r<s\), the local scale update is
\begin{equation}
\label{eq:multi-gh-local}
\lambda_{rs}^2\mid \ldots
\sim
\mathcal{IG}\left(
1,\ \frac{\omega_{rs}^2}{2\tau_\Omega^2} + \frac{1}{\nu_{rs}}
\right),
\quad
\nu_{rs}\mid \lambda_{rs}^2
\sim
\mathcal{IG}\left(
1,\ 1+\frac{1}{\lambda_{rs}^2}
\right).
\end{equation}
Letting \(m_\Omega = q(q-1)/2\) denote the number of off-diagonal entries, the global network scale update is
\begin{equation}
\label{eq:multi-gh-global}
\tau_\Omega^2\mid \ldots
\sim
\mathcal{IG}\left(
\frac{m_\Omega+1}{2},
\ \frac12\sum_{r<s}\frac{\omega_{rs}^2}{\lambda_{rs}^2} + \frac{1}{\xi_\Omega}
\right),
\end{equation}
followed by
\begin{equation}
\label{eq:multi-gh-xi}
\xi_\Omega\mid \tau_\Omega^2
\sim
\mathcal{IG}\left(
1, 1+\frac{1}{\tau_\Omega^2}
\right).
\end{equation}

\subsection{MCMC diagnostics for \texttt{multiVCBART}}
\label{sec:supp_mcmc_diagnostics}

As a standard diagnostic for the \texttt{multiVCBART} sampler, we examined trace
plots from one representative Friedman-type simulation replicate. We ran four
independent Markov chains for \(20000\) iterations, discarding the first \(1000\)
iterations of each chain as burn-in. Rather than attempting to monitor every
tree-specific parameter, which is not meaningful in an overparameterized BART
ensemble, we monitor representative functionals of the posterior mean surface.
Specifically, \Cref{fig:trace_eta_multivcbart} displays trace plots for
\(\eta_{1,1}\) and \(\eta_{1,2}\), the fitted conditional means for a fixed
training observation under outcomes \(y_1\) and \(y_2\), respectively.

The four chains show substantial visual overlap for both fitted-mean quantities,
suggesting stable sampling of the mean-function component. For these two
diagnostic functionals, the corresponding \(\hat R\) values were \(1.06\) for
\(\eta_{1,1}\) and \(1.08\) for \(\eta_{1,2}\). Thus, while the trace plots
indicate broadly reasonable mixing for representative fitted means, we view these
diagnostics as qualitative evidence of sampler stability rather than as an
exhaustive convergence assessment for all latent tree components.

\begin{figure}[t]
    \centering
    \includegraphics[width=0.8\linewidth]{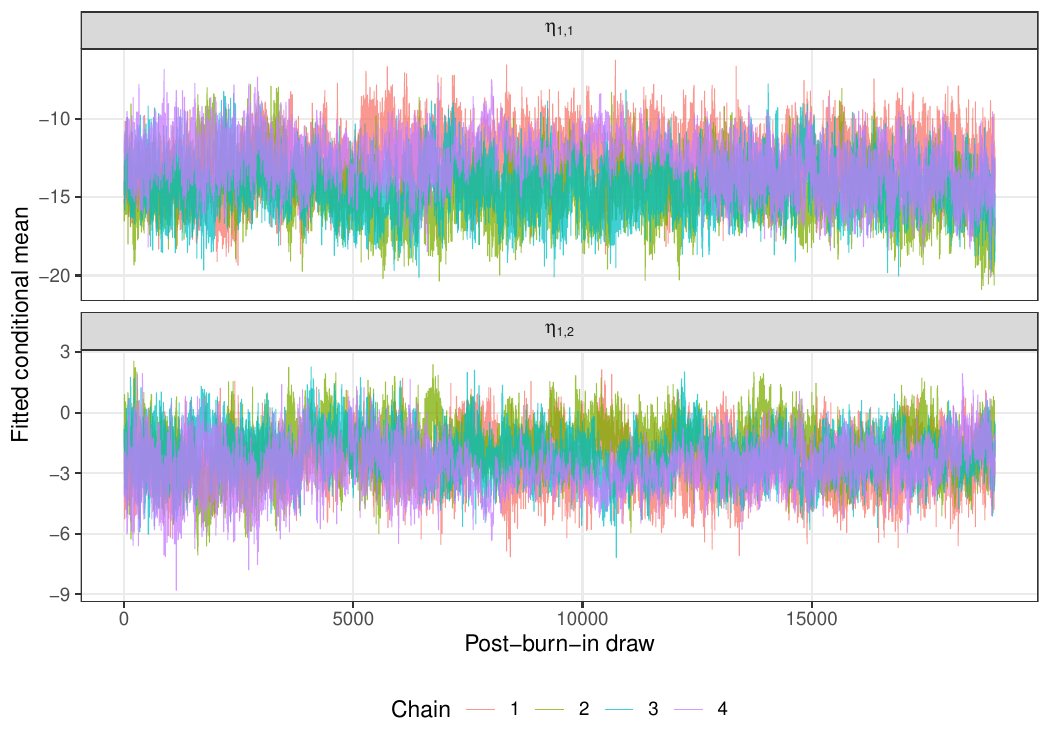}
    \caption{Trace plots for \(\eta_{1,1}\) and \(\eta_{1,2}\), the fitted
    conditional means for a representative training observation under outcomes
    \(y_1\) and \(y_2\), respectively. Four independent \texttt{multiVCBART}
    chains were run for \(20000\) iterations, with the first \(1000\) iterations
    discarded as burn-in.}
    \label{fig:trace_eta_multivcbart}
\end{figure}

\setcounter{figure}{0}
\setcounter{equation}{0}
\setcounter{table}{0}
\setcounter{theorem}{0}
\setcounter{assumption}{0}

\section{Additional experimental results}
\label{sec:multi-experiments}
\subsection{SUR-type simulation setting}
To complement the Friedman benchmark in \Cref{sec:experiments}, we next consider a multivariate regression design in the spirit of
Zellner’s seemingly unrelated regressions \citep{ZellnerSUR}, where each equation is driven by a different subset of
predictors while the outcomes remain correlated through the error covariance.
This non-overlapping signal design is useful for two reasons: (i) first, it isolates the benefit of modelling cross-outcome dependence, because the signal covariates are
disjoint across outcomes, improvements cannot be attributed to sharing the same covariate effects, (ii) second, it stress-tests methods that impose shared partitions or shared tree structures across responses, like \texttt{mvBART}; 
when each response depends on a distinct subset of covariates, a single common partition can be overly
restrictive, a point emphasized in the \texttt{suBART} simulation discussion in \citet{esser2025seeminglyunrelatedbayesianadditive}.

We generate $d=100$ modifiers
$\bx_i=(x_{i1},\ldots,x_{id})^\top$ with $x_{ik}\stackrel{\mathrm{iid}}{\sim}\mathcal{U}(-1,1)$ and set $q=10$. Similar to the first setting, again we take $\bm{z}_i = \bm{1}.$
The mean vector $\boldsymbol\eta_0(\bx_i)=(\eta_{0,1}(\bx_i),\ldots,\eta_{0,10}(\bx_i))^\top$
is defined by ten nonlinear component functions:
\begin{align*}
\eta_{0,1}(\bx_i) &= 4\sin\!\big(\pi x_{i1}x_{i2}\big), &
\eta_{0,2}(\bx_i) &= 3\cos\!\big(\pi x_{i3}\big),\\
\eta_{0,3}(\bx_i) &= 4(x_{i4}^2-0.33), &
\eta_{0,4}(\bx_i) &= 3x_{i5},\\
\eta_{0,5}(\bx_i) &= 4\exp\big(-2x_{i6}^2\big), &
\eta_{0,6}(\bx_i) &= 3x_{i7}x_{i8},\\
\eta_{0,7}(\bx_i) &= 4\big(|x_{i9}|-0.5\big), &
\eta_{0,8}(\bx_i) &= 3\sin\big(\pi x_{i10}\big),\\
\eta_{0,9}(\bx_i) &= 4x_{i11}^3, &
\eta_{0,10}(\bx_i) &= 3x_{i12}+2x_{i13}.
\end{align*}
Hence, each outcome depends on its own dedicated predictor subset, and the remaining $87$ covariates are pure noise.
Similar to the first experiment, we generate responses as
\(
\bm{Y}_i=\boldsymbol\eta_0(\bx_i)+\boldsymbol\varepsilon_i, \ \boldsymbol\varepsilon_i\stackrel{\mathrm{iid}}{\sim}\mathcal N_q(\mathbf 0,\Sigma_0),
\)
where $\Sigma_0=4\,\Omega_0^{-1}$ injects substantial correlated noise. The true precision matrix
$\Omega_0$ is sparse with edges
$(1,2),(2,3),(3,4)$ having weight $-0.4$,
$(5,6),(6,7)$ having weight $+0.5$,
$(8,9)$ with weight $-0.4$, and $(9,10)$ with weight $+0.4$, and unit diagonals. This creates multiple small dependence clusters while remaining sparse, mimicking a
graphical-SUR residual structure.

We compare \texttt{multiVCBART} against \texttt{suBART}, independent
\texttt{BART}, the two-stage \texttt{BART-GHS} procedure, and
\texttt{BayesSUR}. We exclude \texttt{mvBART} from our evaluations since it is not designed to handle more than \(q=2\) outcomes.

We evaluate predictive performance using RMSE and CRPS, and uncertainty quantification using nominal
$95\%$ predictive interval coverage, consistent with the metrics described in the Friedman-type simulation in \Cref{sec:experiments}, across $25$ replications. 
\begin{table}[ht]
\centering
\caption{SUR-type simulation results with \(d=100\) modifiers, and \(q=10\) outcomes. Reported are average test RMSE, average test CRPS, empirical coverage of nominal \(95\%\) predictive intervals, Frobenius precision-matrix recovery error \(\|\hat{\Omega}-\Omega_0\|_F\), and computational time in seconds. Best RMSE, CRPS, precision recovery, and time are bolded; for coverage, the value closest to \(0.95\) is bolded.}
\label{tab:zellner_ortho}
\footnotesize
\begin{tabular}{lccccc}
\toprule
Method & RMSE$_{\text{test}}$ & CRPS$_{\text{test}}$ & PI Coverage$_{\text{test}}$ & $\|\hat{\Omega} - \Omega_0\|_F$ & Time (s) \\
\midrule
\texttt{multiVCBART} & \textbf{0.796} & \textbf{0.517} & 0.827 & \textbf{0.279} & \textbf{30.997} \\
\texttt{suBART}      & 1.039 & 0.606 & 0.942 & 0.614 & 226.023 \\
\texttt{BART}        & 1.049 & 0.619 & \textbf{0.944} & 0.441 & 215.882 \\
\texttt{BART-GHS}    & 1.053 & 0.621 & 0.905 & 0.548 & 219.121 \\
\texttt{BayesSUR}    & 2.304 & 1.320 & 0.918 & 0.360 & 75.600 \\
\bottomrule
\end{tabular}
\end{table}

\Cref{tab:zellner_ortho} shows that \texttt{multiVCBART} achieves the best
predictive accuracy among all methods, attaining the smallest test RMSE and
CRPS. The poor performance of \texttt{BayesSUR} is expected since it enforces linear conditional means, which is fundamentally misspecified for the nonlinear generating functions above. Among tree-based competitors, \texttt{suBART}, independent \texttt{BART}, and \texttt{BART-GHS} all capture aspects of the nonlinear mean structure, but \texttt{multiVCBART} yields a clear additional gain. This is consistent with the intended role of the horseshoe regularization in suppressing noisy coefficient surfaces while jointly leveraging residual dependence through \(\Omega\).


Independent \texttt{BART} and \texttt{suBART} achieve empirical predictive interval coverage closest to the nominal \(95\%\) level, with coverages \(0.944\) and \(0.942\), respectively. In contrast, \texttt{multiVCBART} undercovers, with empirical coverage \(0.827\), despite having the best RMSE and CRPS. This suggests that \texttt{multiVCBART} produces substantially sharper predictive distributions in this setting, but the resulting intervals are somewhat too narrow. 

A similar pattern appears in precision-matrix recovery.
Among the methods considered, \texttt{multiVCBART} attains the smallest Frobenius error, while \texttt{BayesSUR} is competitive but slightly worse, and \texttt{suBART} is substantially less accurate.
For independent \texttt{BART}, there is no joint multivariate residual model, so its reported \(\Omega\) estimate is constructed using only the univariate residual variance draws from each outcome fit. Specifically, we form a diagonal precision matrix whose entries are the reciprocals of the posterior residual variance draws for the separate outcome-specific BART models.
Thus, the \texttt{BART} benchmark is unable to recover off-diagonal residual dependence by construction, and its Frobenius error should be interpreted relative to this restricted diagonal approximation.
Overall, these results indicate that accurate estimation of \(\Omega_0\) requires not only modeling cross-outcome dependence, but also fitting the nonlinear mean structure well enough that the remaining residual variation reflects the latent graphical dependence.

Finally, \texttt{multiVCBART} is also the fastest method in this updated
experiment. This timing advantage is consistent with the computational discussion in \Cref{sec:comp} -- here \(\bm z_i=\bm 1\), so the model updates only \(q=10\) coefficient-function ensembles rather than \(pq\) high-dimensional covariate--outcome surfaces. In this regime, the precision-based pseudo-response update avoids dense multivariate tree updates while keeping the number of BART ensembles modest, making \texttt{multiVCBART} both statistically and computationally attractive.

\subsection{Causal inference example}
We next consider a sparse multivariate causal inference design inspired by the first data generating process in \citet[Section 4]{mcjames2023bayesiancausalforestsmultivariate}, but modified to include a substantial number of null covariates. The original design is favorable to their multivariate tree-based causal estimator \texttt{mvbcf} because the two outcomes share similar prognostic structure and treatment-effect modifiers. To make the setting more challenging and more reflective of high-dimensional observational studies, we retain the same outcome-generating mechanism while increasing the ambient dimension from \(p=10\) to \(p=50\) by appending 40 independent noise variables.

Specifically, for each Monte Carlo replication we generate \(n \in \{500,1000\}\) training observations and an independent test set of size \(1000\).
The first five predictors are continuous with
\(X_1,\ldots,X_5 \stackrel{\mathrm{iid}}{\sim} \mathcal{U}(0,1),\) the next three are binary,
\(X_6,X_7,X_8 \stackrel{\mathrm{iid}}{\sim} \mathrm{Bern}(0.5),\) the next two are ordinal categorical,
\(X_9,X_{10} \in \{0,1,2,3,4\},\) and the remaining $40$ predictors are independent noise covariates generated from \(\mathcal{U}(0,1)\).
Treatment assignment follows
\(
T \mid \bm X \sim \mathrm{Bern}(X_4),
\)
so that \(X_4\) acts as a confounder.

The two potential outcome surfaces are governed by shared nonlinear prognostic structure but distinct treatment effects.
The prognostic mean functions are
\begin{align*}
\mu_1(\bm X)
&=
300 + 10\Bigl\{
11\sin(\pi X_1X_2)
+18(X_3-0.5)^2
+10X_4
+12X_6
+X_9
\Bigr\},\\
\mu_2(\bm X)
&=
300 + 10\Bigl\{
9\sin(\pi X_1X_2)
+22(X_3-0.5)^2
+14X_4
+8X_6
+X_9
\Bigr\},
\end{align*}
while the heterogeneous treatment effects are
\(\tau_1(\bm X) = 10(2X_4+2X_5),\) and \(\tau_2(\bm X) = 10(X_4+3X_5).\)
The observed bivariate outcome is then generated as
\[
\bm{Y}
=
\begin{pmatrix}
Y_1\\
Y_2
\end{pmatrix}
=
\begin{pmatrix}
\mu_1(\bm{X}) + T\tau_1(\bm X)\\
\mu_2(\bm{X}) + T\tau_2(\bm X)
\end{pmatrix}
+
\boldsymbol{\varepsilon},
\quad
\boldsymbol{\varepsilon} \sim N_2(\mathbf 0, 50^2 I_2).
\]

We compare \texttt{multiVCBART} against \texttt{mvbcf}, separate univariate \texttt{BART} fits, and \texttt{mvBART}.
Following the setup in \citet{mcjames2023bayesiancausalforestsmultivariate}, we first estimate the propensity score using a \texttt{BART} model and then include the estimated propensity score as an additional covariate for all competing methods.

Performance is averaged over $50$ Monte Carlo replications and evaluated using the precision in estimation of heterogeneous effects (PEHE) defined as $\sqrt{n^{-1}\sum_{i=1}^{n} (\hat{\tau}(\bx_i) - \tau(\bx_i))^2}$, continuous ranked probability score (CRPS) for $\tau$, and empirical \(95\%\) interval coverage.


\Cref{tab:sparse_dgp_results} summarizes average causal inference performance across the two outcomes under Simulation Setting 3. At \(n=500\), \texttt{multiVCBART} achieves the smallest average PEHE, the smallest average CRPS, and near-nominal \(95\%\) interval coverage. This suggests a clear advantage in the moderate-sample regime, where aggressive shrinkage of the appended noise covariates and joint multivariate learning remain especially beneficial. At \(n=1000\), the differences between the leading methods become noticeably narrower. \texttt{multiVCBART} continues to achieve the smallest average PEHE, indicating the most accurate recovery of the treatment-effect surfaces overall, whereas \texttt{mvbcf} attains the best average CRPS and the highest empirical coverage. This pattern is plausible given that the design of the data-generating mechanism is adapted from a benchmark that is structurally favorable to \texttt{mvbcf}, since the two outcomes share similar prognostic structure and treatment-effect modifiers, while the residual errors are simulated independently \((\Omega_0 \propto I_2)\). Consequently, there is no residual cross-outcome dependence for \texttt{multiVCBART} to exploit, and the advantage of explicitly modeling \(\Omega\) is reduced.
\begin{table}[htbp]
    \centering
    \caption{Average causal inference performance metrics on the Simulation Setting 3 benchmark with \(p=50\). Results represent the mean metric across both outcomes, averaged over 50 Monte Carlo replications. Best metrics are bolded.}
    \label{tab:sparse_dgp_results}
    \footnotesize
    \begin{tabular}{cl ccc}
        \toprule
        \multirow{1}{*}{$n$} & \multirow{1}{*}{Method} 
        & {PEHE} & {CRPS} & {95\% Cov} \\
        \midrule
        \multirow{4}{*}{500}
        & \texttt{multiVCBART} & \textbf{11.13} & \textbf{6.94} & 0.93 \\
        & \texttt{mvbcf}       & 12.75 & 7.67 & \textbf{0.95} \\
        & \texttt{BART}        & 11.83 & 7.29 & 0.91 \\
        & \texttt{mvBART}      & 13.40 & 9.83 & 0.92 \\
        \midrule
        \multirow{4}{*}{1000}
        & \texttt{multiVCBART} & \textbf{8.91} & 5.21 & 0.89 \\
        & \texttt{mvbcf}       & 9.06 & \textbf{5.18} & \textbf{0.97} \\
        & \texttt{BART}        & 9.70 & 5.64 & 0.88 \\
        & \texttt{mvBART}      & 9.29 & 6.60 & 0.86 \\
        \bottomrule
    \end{tabular}
\end{table}

\subsection{Additional details for the GDSC analysis}
\begin{figure}[t]
    \centering
    \includegraphics[width=0.8\linewidth]{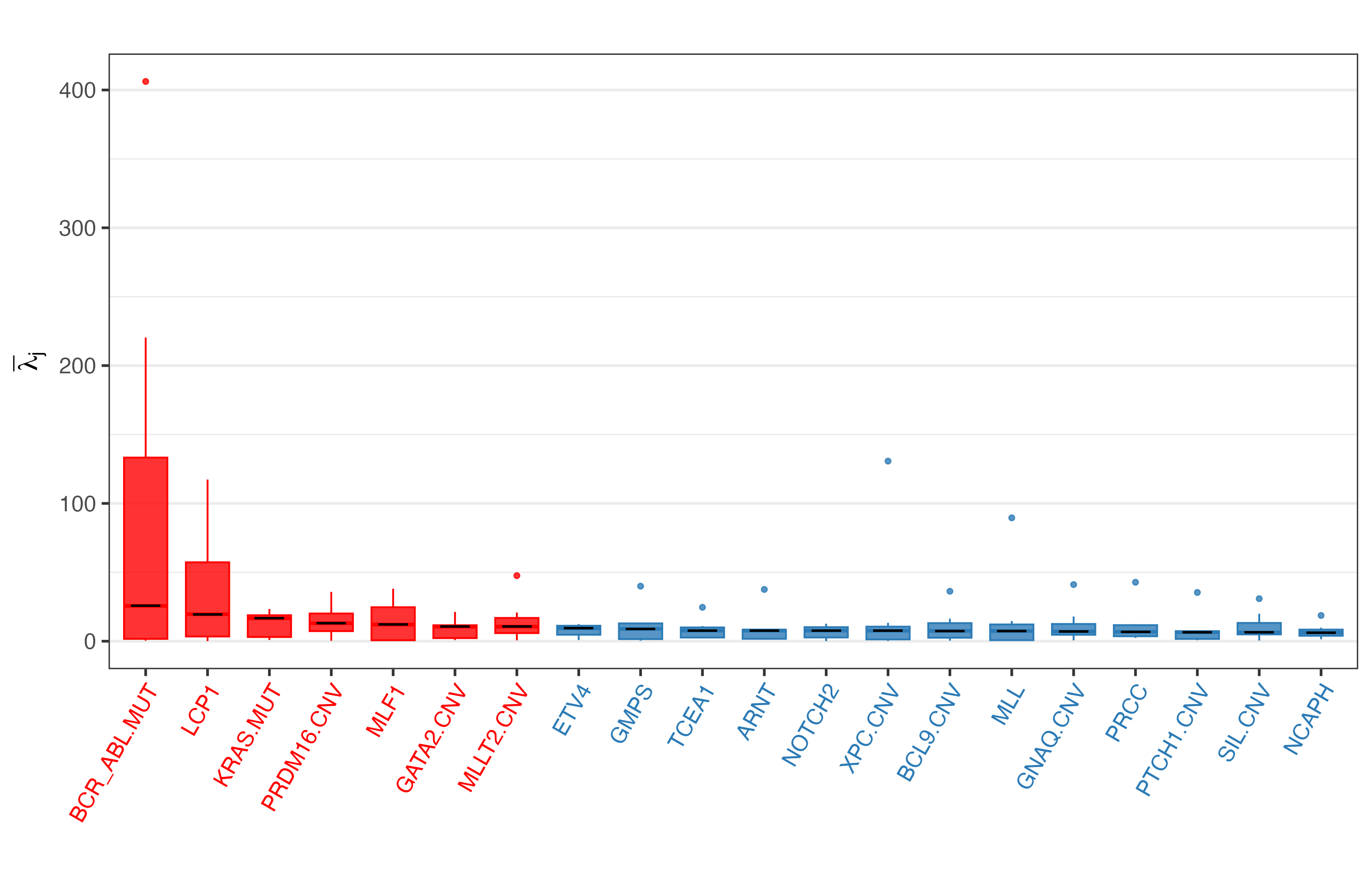}
    \caption{Boxplots of posterior draws of \(\bar{\lambda}_j=q^{-1}\sum_{r=1}^q\lambda_{jr}\) for the top twenty omics features in the GDSC analysis, ranked by decreasing posterior median. Features with posterior median exceeding \(10\) are highlighted.}
    \label{fig:gdsc_lambda}
\end{figure}
The ranking in \Cref{fig:gdsc_lambda} is dominated by \texttt{BCR\_ABL.MUT}, whose posterior local-scale
distribution is clearly separated from the remaining features. The other
top-ranked features show more moderate and overlapping evidence of relevance, suggesting a small number of strong signals followed by a broader set of weaker candidate biomarkers.

\subsection{Sensitivity analyses}\label{sec:sensitivity}
The sensitivity results in \Cref{tab:tree_sensitivity,tab:sigmaB_sensitivity} indicate that the predictive performance of \texttt{multiVCBART} is reasonably stable across the hyperparameter settings considered. Varying the number of trees $M$ has a more visible effect on both accuracy and computation. In this experiment, the smaller ensemble with $M=10$ gives the lowest RMSE and CRPS and is also the fastest, while increasing to $M=40$ leads to higher computational cost without improving prediction. Across all three choices of $M$, however, predictive interval coverage remains close to the nominal $95\%$ level, suggesting that uncertainty quantification is not highly sensitive to this tuning parameter.

The results for the global coefficient-surface scale $\sigma_{B}$ exhibit greater stability. The three choices $\sigma_{B}\in \{0.2,1,5\}$ produce similar RMSE, CRPS, and coverage, with $\sigma_{B}=0.2$ yielding a modest improvement in RMSE and CRPS and $\sigma_{B}=5$ yielding slightly higher coverage. The default value $\sigma_{B}=1$ is computationally fastest in this comparison and provides performance close to the best setting. Overall, these experiments suggest that \texttt{multiVCBART} is not overly sensitive to moderate changes in the ensemble size or global shrinkage scale, although very large ensembles can add substantial computational cost without clear predictive benefit.

\begin{table}[H]
\centering
\caption{Sensitivity of \texttt{multiVCBART} to the number of trees \(M\) in the Friedman-type simulation. Reported are average test RMSE, test CRPS, empirical coverage of nominal \(95\%\) predictive intervals, and runtime in seconds. The default setting is \(M=20\).}
\label{tab:tree_sensitivity}
\footnotesize
\begin{tabular}{lcccc}
\toprule
\(M\) & RMSE$_{\text{test}}$ & CRPS$_{\text{test}}$ & PI Coverage$_{\text{test}}$ & Time (s) \\
\midrule
10 & \textbf{2.247} & \textbf{1.192} & \textbf{0.943} & \textbf{62.568} \\
20 & 2.493 & 1.341 & 0.940 & 77.974 \\
40 & 2.775 & 1.529 & 0.941 & 221.946 \\
\bottomrule
\end{tabular}
\end{table}

\begin{table}[H]
\centering
\caption{Sensitivity of \texttt{multiVCBART} to the global coefficient-surface shrinkage scale \(\sigma_{B}\) in the Friedman-type simulation. Reported are average test RMSE, test CRPS, empirical coverage of nominal \(95\%\) predictive intervals, and runtime in seconds. The default setting is \(\sigma_{B}=1\).}
\label{tab:sigmaB_sensitivity}
\footnotesize
\begin{tabular}{lcccc}
\toprule
\(\sigma_{B}\) & RMSE$_{\text{test}}$ & CRPS$_{\text{test}}$ & PI Coverage$_{\text{test}}$ & Time (s) \\
\midrule
0.2 & \textbf{2.353} & \textbf{1.262} & 0.937 & 126.124 \\
1   & 2.493 & 1.341 & 0.940 & \textbf{77.974} \\
5   & 2.490 & 1.348 & \textbf{0.942} & 125.093 \\
\bottomrule
\end{tabular}
\end{table}

\end{document}